\documentclass[a4paper,twocolumn,accepted=2025-02-17]{quantumarticle}
\pdfoutput=1

\usepackage[verbose=true]{microtype}
\usepackage{amsmath}
\usepackage{amssymb}
\usepackage{amsthm}
\usepackage{mathtools}
\usepackage[table,dvipsnames]{xcolor}
\usepackage{caption}
\usepackage{bbm}
\usepackage{enumitem}
\usepackage{bbding}
\usepackage{multirow}

\usepackage{caption}
\usepackage{subcaption}

\usepackage{graphicx}
\usepackage{float}
\usepackage{color}
\usepackage[colorlinks=true, citecolor=color3, linkcolor=color3,urlcolor=color4]{hyperref}

\usepackage{cleveref}

\usepackage{tikz}
\usetikzlibrary{backgrounds}
\usepackage{pgfplots}
\usepgfplotslibrary{patchplots,colormaps}
\pgfplotsset{compat=1.9}
\usetikzlibrary{3d}
\usetikzlibrary{fadings}
\usetikzlibrary{shadows.blur,shapes}
\usetikzlibrary{3d,backgrounds,calc}

\usepackage{braket}

\tikzfading[name=fade out, 
    inner color=transparent!0,
    outer color=transparent!100]   

\usepackage{mdframed}

\newtheorem{theorem}{Theorem}
\newtheorem{lemma}[theorem]{Lemma}
\newtheorem{corollary}[theorem]{Corollary}
\newtheorem{proposition}[theorem]{Proposition}

\newtheorem{conjecture}[theorem]{Conjecture}
\newtheorem{definition}[theorem]{Definition}

\theoremstyle{definition}

\newtheorem{example}[theorem]{Example}
\newtheorem{remark}[theorem]{Remark}

\usepackage{tikz-3dplot}
\usetikzlibrary{decorations.pathreplacing,calligraphy}

\renewcommand{\epsilon}{\varepsilon}

\newcommand{\Diag}{\textrm{Diag}}

\newcommand{\N}{\mathbb{N}}

\newcommand{\tr}{\mathrm{tr}}

\newcommand{\Cbraket}[2]{\braket{#1 \, | \, #2}}

\usepackage{contour}
\usepackage[normalem]{ulem}

\newcommand{\myuline}[1]{%
  \uline{\smash{#1}}%
}

\usepackage{xcolor}

\DeclareMathOperator{\CPTP}{\mathrm{CPTP}}

\DeclareMathOperator{\trank}{\operatorname{t-rank}}

\DeclareMathOperator{\rank}{\operatorname{rank}}
\DeclareMathOperator{\brank}{\myuline{\operatorname{rank}}}
\DeclareMathOperator{\symmrank}{\operatorname{symm-rank}}
\DeclareMathOperator{\bsymmrank}{\myuline{\operatorname{symm-rank}}}

\DeclareMathOperator{\osr}{\operatorname{osr}}
\DeclareMathOperator{\tiosr}{\operatorname{ti-osr}}
\DeclareMathOperator{\bosr}{\myuline{\operatorname{osr}}}
\DeclareMathOperator{\btiosr}{\myuline{\operatorname{ti-osr}}}

\DeclareMathOperator{\psdrank}{\operatorname{psd-rank}}
\DeclareMathOperator{\bpsdrank}{\myuline{\operatorname{psd-rank}}}
\DeclareMathOperator{\symmpsdrank}{\operatorname{symm-psd-rank}}
\DeclareMathOperator{\bsymmpsdrank}{\myuline{\operatorname{symm-psd-rank}}}

\DeclareMathOperator{\psdosr}{\operatorname{psd-osr}}
\DeclareMathOperator{\tipsdosr}{\operatorname{ti-psd-osr}}
\DeclareMathOperator{\bpsdosr}{\myuline{\operatorname{psd-osr}}}
\DeclareMathOperator{\btipsdosr}{\myuline{\operatorname{ti-psd-osr}}}

\DeclareMathOperator{\nnrank}{\operatorname{nn-rank}}
\DeclareMathOperator{\bnnrank}{\myuline{\operatorname{nn-rank}}}
\DeclareMathOperator{\symmnnrank}{\operatorname{symm-nn-rank}}

\DeclareMathOperator{\nnosr}{\operatorname{nn-osr}}
\DeclareMathOperator{\tinnosr}{\operatorname{ti-nn-osr}}
\DeclareMathOperator{\bnnosr}{\myuline{\operatorname{nn-osr}}}
\DeclareMathOperator{\btinnosr}{\myuline{\operatorname{ti-nn-osr}}}

\DeclareMathOperator{\purirank}{\operatorname{puri-rank}}
\DeclareMathOperator{\bpurirank}{\myuline{\operatorname{puri-rank}}}
\DeclareMathOperator{\symmpurirank}{\operatorname{symm-puri-rank}}

\DeclareMathOperator{\puriosr}{\operatorname{puri-osr}}
\DeclareMathOperator{\tipuriosr}{\operatorname{ti-puri-osr}}

\DeclareMathOperator{\seprank}{\operatorname{sep-rank}}
\DeclareMathOperator{\bseprank}{\myuline{\operatorname{sep-rank}}}
\DeclareMathOperator{\symmseprank}{\operatorname{symm-sep-rank}}

\DeclareMathOperator{\seposr}{\operatorname{sep-osr}}
\DeclareMathOperator{\tiseposr}{\operatorname{ti-sep-osr}}

\DeclareMathOperator{\CCorr}{\mathsf{CCorr}}
\DeclareMathOperator{\CQCorr}{\mathsf{CQCorr}}
\DeclareMathOperator{\QQCorr}{\mathsf{QQCorr}}

\newcommand{\facetF}{\mathcal{\widetilde{F}}}
\newcommand{\facetG}{\mathcal{\widetilde{G}}}
\newcommand{\facetH}{\mathcal{\widetilde{H}}}

\newcommand{\sfacetF}{\mathcal{\wt[0.35ex]{F}}}
\newcommand{\sfacetG}{\mathcal{\wt[0.35ex]{G}}}
\newcommand{\sfacetH}{\mathcal{\wt[0.35ex]{H}}}

\newcommand{\id}{\mathrm{id}}
\newcommand{\one}{\mathbbm{1}}

\newcommand{\Lin}{\mathrm{Lin}}

\usepackage{tabto}

\def\thesubsection{\arabic{section}.\arabic{subsection}}
\makeatletter
\def\p@subsection{}
\def\p@subsubsection{}
\makeatother

\let\originalleft\left
\let\originalright\right
\renewcommand{\left}{\mathopen{}\mathclose\bgroup\originalleft}
\renewcommand{\right}{\aftergroup\egroup\originalright}

\newcommand\xqed[1]{%
  \leavevmode\unskip\penalty9999 \hbox{}\nobreak\hfill
  \quad\hbox{#1}}
\newcommand\demo{\xqed{$\triangle$}}

\definecolor{color4}{HTML}{2B3467}
\definecolor{color3}{HTML}{EB455F}
\definecolor{color1}{HTML}{BAD7E9}
\definecolor{color2}{HTML}{FCFFE7}

\usepackage{array}

\makeatletter
\newcommand*\wt[2][0.1ex]{%
        \begingroup
        \mathchoice{\wt@helper{#1}{#2}{\displaystyle}{\textfont}}
                   {\wt@helper{#1}{#2}{\textstyle}{\textfont}}
                   {\wt@helper{#1}{#2}{\scriptstyle}{\scriptfont}}
                   {\wt@helper{#1}{#2}{\scriptscriptstyle}{\scriptscriptfont}}%
        \endgroup
        #2%
}
\newcommand*\wt@helper[4]{%
        \def\currentfont{\the#41}%
        \def\currentskewchar{\char\the\skewchar\currentfont}%
        \setbox\tw@\hbox{\currentfont#2\currentskewchar}%
        \dimen@ii\wd\tw@
        \setbox\tw@\hbox{\currentfont#2{}\currentskewchar}%
        \advance\dimen@ii-\wd\tw@
        \rlap{\raisebox{-#1}{$\m@th#3\kern\dimen@ii\widetilde{\phantom{#2}}$}}%
}
\makeatother

\begin{document}

	\title{Border Ranks of Positive and Invariant Tensor Decompositions: Applications to Correlations}

	\author{Andreas Klingler}
	\email{Andreas.Klingler@univie.ac.at}
	\affiliation{Institute for Theoretical Physics, Technikerstr.\ 21a,  A-6020 Innsbruck, Austria}
	\affiliation{Faculty of Mathematics, Oskar-Morgenstern-Platz 1,  A-1090 Wien, Austria}
	
	\author{Tim Netzer}
	\affiliation{Department of Mathematics, Technikerstr.\ 13,  A-6020 Innsbruck, Austria}
	
	\author{Gemma De les Coves}
	\affiliation{Institute for Theoretical Physics, Technikerstr.\ 21a,  A-6020 Innsbruck, Austria}

\maketitle

\begin{abstract}
The matrix rank and its positive versions are robust for small approximations, 
i.e.\ they do not decrease under small perturbations. 
In contrast, the multipartite tensor rank can collapse for arbitrarily small errors, 
i.e.\ there may be a gap between rank and border rank, 
leading to instabilities in the optimization over sets with fixed tensor rank. 
Can multipartite positive ranks also collapse for small perturbations? 
In this work, we prove that multipartite positive and invariant tensor decompositions exhibit gaps between rank and border rank, including tensor rank purifications and cyclic separable decompositions. 
We also prove a correspondence between positive decompositions and membership in certain sets of multipartite probability distributions, and leverage the gaps between rank and border rank to prove that these correlation sets are not closed. 
It follows that testing membership of probability distributions arising from resources like translational invariant Matrix Product States is impossible in finite time.
Overall, this work sheds light on the instability of ranks 
and the unique behavior of bipartite systems. 
\end{abstract}

\section{Introduction}

It is well-known that low-rank approximations of matrices are well-behaved: 
For every matrix, there is a best low-rank approximation with a fixed error, 
and every element closer to the original matrix must be of larger rank. 
In other words, the approximate rank
$$ 
\rank^{\varepsilon}(T) \coloneqq \inf_{W \in B_{\varepsilon}(T)} \rank(W)
$$
where $B_{\varepsilon}(T)$ is the $\varepsilon$-ball around an element $T$ for a given norm $\Vert \cdot \Vert$, coincides with the exact rank when $\varepsilon > 0$ is small enough.

\begin{figure}[t]
\begin{tikzpicture}[scale=0.8]

\draw[fill=color1, draw=none]  plot[smooth, tension=0.7] coordinates {(-1,0) (0,0.5) (3,0.8) (5,0) (3.5, -1.4) (1,-1.6) (-1,-1) (-1,0)};

\node at (4,0) {$\mathcal{V}^{\otimes n}$};

\node at (-0.3,-0.35) {$T$};

\filldraw (0,-0.5) circle (2pt);

\filldraw (0.8,-0.95) circle (1.5pt);
\node at (1.1,-0.7) {$T_{\varepsilon}$};

\draw[dashed] (3,-0.5) to[in=30, out=-180] (2,-1) to[in=-30, out=-150] (0,-0.5);

\draw (0,-2) -- (0,-3.5);
\draw (-0.1,-3.4) -- (3,-3.4);

\draw[color3] (0,-3) -- (2.8,-3);
\draw[color3,fill=white] (0,-3) circle (2pt);

\draw[color3,fill=color3] (0,-2.5) circle (2pt);

\node[color=color3, anchor=west,text=black] at (3,-3) {$\trank(T_{\varepsilon})$};
\node[color=color3, anchor=east,text=black] at (-0.2,-2.5) {$\trank(T)$};

\node at (3.2,-3.5) {$\varepsilon$};

\draw[dotted] (0.8, -1.15) -- (0.8,-2.9);
\draw[dotted] (0, -0.7) -- (0,-1.9);

\end{tikzpicture}
\caption{\textbf{Border rank}. Given a tensor $T$ in an $n$-fold tensor product space and a certain type of rank $\trank$, if there exists a family of tensors $(T_{\varepsilon})_{\varepsilon > 0}$ such that $T_{\varepsilon} \to T$ for $\varepsilon \to 0$ and $\trank(T_{\varepsilon}) < \trank(T)$ for all $\varepsilon > 0$, we say that $\trank$ exhibits a gap between rank and border rank.
}
\label{fig:borderRank}
\end{figure}

\begin{figure*}[thb]\centering
\begin{tikzpicture}[scale=1.04]

\draw (0,1.8) -- (3,0);

\draw (0,0) -- (15.5,0);
\draw (1,-1.5) -- (15.5,-1.5);
\draw (1,-3) -- (15.5,-3);

\draw (3,2.5) -- (3,-4.5);
\draw (5.5,2.5) -- (5.5,-4.5);
\draw (8,2.5) -- (8,-4.5);
\draw (10.5,2.5) -- (10.5,-4.5);
\draw (13,2.5) -- (13,-4.5);

\node[anchor=west] at (1.2,-0.75) {$\rank$};

\node[anchor=west] at (1.2,-2) {$\psdrank$};
\node[anchor=west] at (1.2,-2.5) {$\purirank$};

\node[anchor=west] at (1.2,-3.5) {$\nnrank$};
\node[anchor=west] at (1.2,-4) {$\seprank$};

\node[anchor=west] at (0.2,0.8) {\textsc{Type}};
\node[anchor=west] at (0.2,0.4) {\textsc{of rank}};

\node[anchor=east] at (2.9,2.1) {\textsc{Decomposition}};
\node[anchor=east] at (2.9,1.7) {\textsc{type}};

\node at (4.25,2.1) {Standard};
\node at (6.75,2.1) {Symmetric};
\node at (9.25,2.1) {Cyclic};
\node at (11.75,2.3) {Translational};
\node at (11.75,1.9) {invariant};
\node at (14.25,2.1) {Tree};

\begin{scope}[xshift=4.25cm , yshift=0.9cm]
\draw[fill=gray!50!white] (18:0.5cm) -- (90:0.5cm) -- (18+2*72:0.5cm) --
(18+3*72:0.5cm) -- (18+4*72:0.5cm) -- cycle;
\filldraw (18:0.5cm) circle (1.2pt);
\filldraw (18+72:0.5cm) circle (1.2pt);
\filldraw (18+2*72:0.5cm) circle (1.2pt);
\filldraw (18+3*72:0.5cm) circle (1.2pt);
\filldraw (18+4*72:0.5cm) circle (1.2pt);
\end{scope}

\begin{scope}[xshift=6.75cm , yshift=0.9cm]
\draw[fill=gray!50!white] (18:0.5cm) -- (90:0.5cm) -- (18+2*72:0.5cm) --
(18+3*72:0.5cm) -- (18+4*72:0.5cm) -- cycle;
\filldraw (18:0.5cm) circle (1.2pt);
\filldraw (18+72:0.5cm) circle (1.2pt);
\filldraw (18+2*72:0.5cm) circle (1.2pt);
\filldraw (18+3*72:0.5cm) circle (1.2pt);
\filldraw (18+4*72:0.5cm) circle (1.2pt);

\draw[stealth-stealth, color3, thick] (30:0.58cm) to[out=118, in=-10] (78:0.58cm);
\draw[stealth-stealth, color3, thick, rotate=72] (30:0.58cm) to[out=118, in=-10] (78:0.58cm);
\draw[stealth-stealth, color3, thick, rotate=2*72] (30:0.58cm) to[out=118, in=-10] (78:0.58cm);
\draw[stealth-stealth, color3, thick, rotate=3*72] (30:0.58cm) to[out=118, in=-10] (78:0.58cm);
\draw[stealth-stealth, color3, thick, rotate=4*72] (30:0.58cm) to[out=118, in=-10] (78:0.58cm);s
\end{scope}

\begin{scope}[xshift=9.25cm , yshift=0.9cm]
\draw (18:0.5cm) -- (90:0.5cm) -- (18+2*72:0.5cm) --
(18+3*72:0.5cm) -- (18+4*72:0.5cm) -- cycle;
\filldraw (18:0.5cm) circle (1.2pt);
\filldraw (18+72:0.5cm) circle (1.2pt);
\filldraw (18+2*72:0.5cm) circle (1.2pt);
\filldraw (18+3*72:0.5cm) circle (1.2pt);
\filldraw (18+4*72:0.5cm) circle (1.2pt);
\end{scope}

\begin{scope}[xshift=11.75cm , yshift=0.9cm]
\draw (18:0.5cm) -- (90:0.5cm) -- (18+2*72:0.5cm) --
(18+3*72:0.5cm) -- (18+4*72:0.5cm) -- cycle;
\filldraw (18:0.5cm) circle (1.2pt);
\filldraw (18+72:0.5cm) circle (1.2pt);
\filldraw (18+2*72:0.5cm) circle (1.2pt);
\filldraw (18+3*72:0.5cm) circle (1.2pt);
\filldraw (18+4*72:0.5cm) circle (1.2pt);

\draw[-stealth, color3, thick] (30:0.58cm) to[out=118, in=-10] (78:0.58cm);
\draw[-stealth, color3, thick, rotate=72] (30:0.58cm) to[out=118, in=-10] (78:0.58cm);
\draw[-stealth, color3, thick, rotate=2*72] (30:0.58cm) to[out=118, in=-10] (78:0.58cm);
\draw[-stealth, color3, thick, rotate=3*72] (30:0.58cm) to[out=118, in=-10] (78:0.58cm);
\draw[-stealth, color3, thick, rotate=4*72] (30:0.58cm) to[out=118, in=-10] (78:0.58cm);
\end{scope}

\begin{scope}[xshift=14.25cm, yshift=0.9cm]

\draw (0,0.5) -- (0.5,0) -- (0.75,-0.5);
\draw (0.5,0) -- (0.25,-0.5);

\draw (0,0.5) -- (0,0);

\draw (0,0.5) -- (-0.5,0) -- (-0.25,-0.5);
\draw (-0.5,0) -- (-0.75,-0.5);

\filldraw (0,0.5) circle (1.2pt);
\filldraw (0.5,0) circle (1.2pt);
\filldraw (-0.5,0) circle (1.2pt);
\filldraw (0,0) circle (1.2pt);
\filldraw (-0.75,-0.5) circle (1.2pt);
\filldraw (-0.25,-0.5) circle (1.2pt);
\filldraw (0.25,-0.5) circle (1.2pt);
\filldraw (0.75,-0.5) circle (1.2pt);

\end{scope}

\draw[thick, -stealth] (0.5,-0.75) -- (0.5, -3.75);

\node[rotate=90] at (0,-2.25) {\scriptsize Imposing stronger};
\node[rotate=90] at (0.25,-2.25) {\scriptsize positivity constraints};


\begin{scope}[yshift=-0.75cm, xshift=4.1cm]
\node[anchor=west] at (-1, 0.25) {Yes for $n \geq 3$};
\node[anchor=west] at (-1, -0.25) {\cite{Bi80, La11}};
\end{scope}

\begin{scope}[yshift=-0.75cm, xshift=6.6cm]
\node[anchor=west] at (-1, 0.25) {Yes for $n \geq 3$};
\node[anchor=west] at (-1, -0.25) {\cite{La11}};
\end{scope}

\begin{scope}[yshift=-0.75cm, xshift=9.1cm]
\node[anchor=west] at (-1, 0.25) {Yes for $n \geq 3$};
\node[anchor=west] at (-1, -0.25) {\cite{Ba21,Ch20, La12}};
\end{scope}

\begin{scope}[yshift=-0.75cm, xshift=11.6cm]
\node[anchor=west] at (-1, 0.25) {Yes for $n \geq 3$};
\node[anchor=west] at (-1, -0.25) {\cite{Ba21,Ch20}};
\end{scope}

\begin{scope}[yshift=-0.75cm, xshift=14.1cm]
\node[anchor=west] at (-1, 0.25) {No};
\node[anchor=west] at (-1, -0.25) {\cite{Ba21, La12}};
\end{scope}


\begin{scope}[yshift=-2.25cm, xshift=4.1cm]
\node[anchor=west] at (-1, 0.25) {\textbf{Yes} for $n \geq 5$};
\node[anchor=west] at (-1, -0.25) {(\Cref{ssec:standardBorder})};
\end{scope}

\begin{scope}[yshift=-2.25cm, xshift=6.6cm]
\node[anchor=west] at (-1, 0.25) {\textbf{Yes} for $n \geq 3$};
\node[anchor=west] at (-1, -0.25) {(\Cref{ssec:standardBorder})};
\end{scope}

\begin{scope}[yshift=-2.25cm, xshift=9.1cm]
\node at (0, 0) {?};
\end{scope}

\begin{scope}[yshift=-2.25cm, xshift=11.6cm]
\node[anchor=west] at (-1, 0.25) {\textbf{Yes} for $n \geq 17$};
\node[anchor=west] at (-1, -0.25) {(\Cref{ssec:tiBorder})};
\end{scope}

\begin{scope}[yshift=-2.25cm, xshift=14.1cm]
\node[anchor=west] at (-1, 0.25) {\bf{No}};
\node[anchor=west] at (-1, -0.25) {(\Cref{thm:treeBorderRank})};
\end{scope}


\begin{scope}[yshift=-3.75cm, xshift=4.1cm]
\node[anchor=west] at (-1, 0.25) {\bf{No}};
\node[anchor=west] at (-1, -0.25) {(\Cref{thm:nnDecNoBorderRank})};
\end{scope}

\begin{scope}[yshift=-3.75cm, xshift=6.6cm]
\node[anchor=west] at (-1, 0.25) {\bf{No}};
\node[anchor=west] at (-1, -0.25) {(\Cref{thm:nnDecNoBorderRank})};
\end{scope}

\begin{scope}[yshift=-3.75cm, xshift=9.1cm]
\node[anchor=west] at (-1, 0.25) {\textbf{Yes} for $n \geq 3$};
\node[anchor=west] at (-1, -0.25) {(\Cref{ssec:cyclic})};
\end{scope}

\begin{scope}[yshift=-3.75cm, xshift=11.6cm]
\node[anchor=west] at (-1, 0.25) {\textbf{Yes} for $n \geq 5$};
\node[anchor=west] at (-1, -0.25) {(\Cref{ssec:tiBorder})};
\end{scope}

\begin{scope}[yshift=-3.75cm, xshift=14.1cm]
\node[anchor=west] at (-1, 0.25) {\bf{No}};
\node[anchor=west] at (-1, -0.25) {(\Cref{thm:treeBorderRank})};
\end{scope}

\end{tikzpicture}
\captionsetup{width=.7\linewidth, margin=0.5cm}
\caption{
\textbf{Is there a gap between rank and border rank in an $n$-fold tensor product space?} This table summarizes known results and the contributions of this paper (marked boldface): We prove that gaps persist when imposing positivity constrains corresponding to quantum correlation scenarios (second row), 
and that certain gaps disappear for stronger positivity constrains corresponding to classical correlation scenarios (third row). 
The types of ranks and of decompositions are defined in \Cref{sec:posTensorDec}.}
\label{fig:borderRankResults}
\end{figure*}

Surprisingly, the multipartite tensor rank behaves very differently: There exist tensors $T$ where the \emph{border rank} 
$$ \brank(T) \coloneqq \lim_{\varepsilon \to 0} \rank^{\varepsilon}(T)$$
is \emph{strictly smaller} than the rank of $T$ (see \Cref{fig:borderRank}). For the mathematician, this means that the rank is not lower semi-continuous. 
This is equivalent to the statement that the set of tensors whose rank is upper bounded by a constant $r$
$$
\mathcal{T} \coloneqq \left\{T \in \mathcal{V}^{\otimes n}: \rank(T) \leq r \right\}
$$
is topologically not closed, since there are sequences in $\mathcal{T}$ whose limit is not in $\mathcal{T}$. As a consequence, optimization problems over such sets, such as computing an optimal low-rank approximation, are generally ill-posed \cite{Si08}. It is known that 
 tensor decompositions with three or more local spaces exhibit a gap between rank and border rank \cite{Bi80, La17, Zu17}, 
 and so do tensor network decompositions containing loops \cite{La12, Ch20, Ba21, Ch21c},  
 where some of these results concern symmetric decompositions of invariant tensors. We refer to \cite{Co20} for an elaborate treatment of the topological properties of tensor ranks. For a perspective on optimization problems over fixed-rank tensors via the condition number, we refer to \cite{Be23}.

What about \emph{positive} multipartite tensors? 
Positive semidefinite and nonnegative tensors represent quantum and classical notions of positivity, respectively, and are thus central in quantum information theory and classical probability theory. 
Yet, the tensor decompositions mentioned above are insufficient to represent positive elements, 
as it is generally impossible to guarantee the global positivity  \cite{De15, Kl14}.  
This can be circumvented by introducing locally positive constraints, as in the local purification form and the separable decomposition for mixed states \cite{De13c, De19, De19d}, 
or the nonnegative, positive semidefinite and square-root decomposition for entrywise nonnegative tensors \cite{Fa14, Ja13, De19,De19d, Gl19}. 
Each of these gives rise to a notion of positive rank, 
which characterizes the complexity of linear vs.\  semidefinite programs \cite{Ya91, Go12} or classical vs.\ quantum correlation complexities \cite{Ja13, Fi12, Ja17}.
Can multipartite positive tensor ranks, possibly with invariance, exhibit a gap with its border rank? 

In this work, we prove that several locally positive and invariant decompositions exhibit a gap between rank and border rank, as summarized in \Cref{fig:borderRankResults}. 
This includes positive and/or symmetric versions of Matrix Product States (MPS) and Matrix Product Operators (MPO). 
We also prove that the ranks of these decompositions determine whether a probability distribution can be generated by certain classes of states via local measurements; 
for example, we prove that the positive semidefinite rank of a nonnegative tensor captures the amount of entanglement needed to generate multipartite probability distributions by local measurements (see \Cref{fig:quantum_correlation}).
We then leverage the gaps to border ranks together with this connection to quantum correlations to show that: 
\begin{enumerate}[label=(\roman*)]
	\item If a tensor network has a loop, computing the best approximation with fixed positive rank is ill-posed. Specifically, given a mixed state $\rho$ (of rank larger than $r$), there is generally no mixed state $\sigma$ which is the best approximation among all decompositions with a positive rank bounded by $r$, 
	because for any   $\varepsilon > 0$ there is an $\varepsilon$-close mixed state of rank $r$, 
	while the rank of $\rho$ is strictly larger than $r$.
	
	\item The set of probability distributions generated by a multipartite state with local measurements (\Cref{fig:quantum_correlation} (a)) is not closed. Thus, it is impossible to verify the necessity of a certain resource state from sampling the distribution, even in arbitrarily many rounds. The same applies to generating multipartite mixed states from local quantum channels (\Cref{fig:quantum_correlation} (b)).
	
	\item We provide correlation scenarios where the quantum case is fragile with respect to approximations, while the classical case is robust. This shows a novel type of separation between these two scenarios.
\end{enumerate}

\begin{figure}
\begin{tikzpicture}

\begin{scope}[scale=0.9]
\begin{scope}
\draw (18:0.5cm) -- (90:0.5cm) -- (18+2*72:0.5cm) --
(18+3*72:0.5cm) -- (18+4*72:0.5cm) -- cycle;
\filldraw (18:0.5cm) circle (1.2pt);
\filldraw (18+72:0.5cm) circle (1.2pt);
\filldraw (18+2*72:0.5cm) circle (1.2pt);
\filldraw (18+3*72:0.5cm) circle (1.2pt);
\filldraw (18+4*72:0.5cm) circle (1.2pt);

\draw[-stealth, color3, thick] (30:0.58cm) to[out=118, in=-10] (78:0.58cm);
\draw[-stealth, color3, thick, rotate=72] (30:0.58cm) to[out=118, in=-10] (78:0.58cm);
\draw[-stealth, color3, thick, rotate=2*72] (30:0.58cm) to[out=118, in=-10] (78:0.58cm);
\draw[-stealth, color3, thick, rotate=3*72] (30:0.58cm) to[out=118, in=-10] (78:0.58cm);
\draw[-stealth, color3, thick, rotate=4*72] (30:0.58cm) to[out=118, in=-10] (78:0.58cm);

\node at (0,-1.2) {$\psdrank(P) \leq r$};

\node at (1.35,0) {$\Longleftrightarrow$};

\node at (-1.2,0) {(a)};

\end{scope}

\begin{scope}[xshift=2.4cm, yshift=0.4cm, scale=0.38]

\draw[fill=color1, draw=none] (-1,-0.5) rectangle (11.5,2.4);
\node at (10.5, 0.5) {$\ket{\psi}$};




\draw[color3, thick, -stealth] (0.7,1.7) to[in=160,out=20] (2.3,1.7);
\draw[color3, thick, -stealth] (2.7,1.7) to[in=160,out=20] (4.3,1.7);
\draw[color3, thick, -stealth] (4.7,1.7) to[in=160,out=20] (6.3,1.7);
\draw[color3, thick, -stealth] (6.7,1.7) to[in=160,out=20] (8.3,1.7);

\draw[thick, rounded corners] (-0.5,0.5) -- (3.5,0.5) node[midway,below, yshift=-0.05cm] {\scriptsize $r$} -- (9.5,0.5) -- (9.5,1.5) -- (-0.5,1.5) -- cycle;

\draw[-stealth] (0.5,0.5) -- (0.5,-1.5);
\draw[-stealth] (2.5,0.5) -- (2.5,-1.5);
\draw[-stealth] (4.5,0.5) -- (4.5,-1.5);
\draw[-stealth] (6.5,0.5) -- (6.5,-1.5);
\draw[-stealth] (8.5,0.5) -- (8.5,-1.5);

\draw[-stealth] (0.5,-3) -- (0.5,-4);
\draw[-stealth] (2.5,-3) -- (2.5,-4);
\draw[-stealth] (4.5,-3) -- (4.5,-4);
\draw[-stealth] (6.5,-3) -- (6.5,-4);
\draw[-stealth] (8.5,-3) -- (8.5,-4);

\draw[decorate, decoration={brace}] (8.75,-4.2) -- (0.25,-4.2) node[midway, below, yshift=-0.1cm, color=black] {$P$};

\draw[fill=white] (0,0) rectangle (1,1);
\draw[fill=white] (2,0) rectangle (3,1);
\draw[fill=white] (4,0) rectangle (5,1);
\draw[fill=white] (6,0) rectangle (7,1);
\draw[fill=white] (8,0) rectangle (9,1);

\draw[fill=white] (-0.25,-3) rectangle (1.25,-1.5);
\draw[fill=white] (1.75,-3) rectangle (3.25,-1.5);
\draw[fill=white] (3.75,-3) rectangle (5.25,-1.5);
\draw[fill=white] (5.75,-3) rectangle (7.25,-1.5);
\draw[fill=white] (7.75,-3) rectangle (9.25,-1.5);


\begin{scope}[shift={(0.5,-2.9)}]
\draw[thick,domain=45:135, xshift=0] plot ({0.7*cos(\x)}, {0.7*sin(\x)});
\draw[line width=0.1cm, white, rotate around={-15:(0,0.1)}] (0,0.2) -- (0,1.1);
\draw[-stealth, rotate around={-15:(0,0.1)}] (0,0.2) -- (0,1.1);
\end{scope}

\begin{scope}[shift={(2.5,-2.9)}]
\draw[thick,domain=45:135, xshift=0] plot ({0.7*cos(\x)}, {0.7*sin(\x)});
\draw[line width=0.1cm, white, rotate around={-15:(0,0.1)}] (0,0.2) -- (0,1.1);
\draw[-stealth, rotate around={-15:(0,0.1)}] (0,0.2) -- (0,1.1);
\end{scope}

\begin{scope}[shift={(4.5,-2.9)}]
\draw[thick,domain=45:135, xshift=0] plot ({0.7*cos(\x)}, {0.7*sin(\x)});
\draw[line width=0.1cm, white, rotate around={-15:(0,0.1)}] (0,0.2) -- (0,1.1);
\draw[-stealth, rotate around={-15:(0,0.1)}] (0,0.2) -- (0,1.1);
\end{scope}

\begin{scope}[shift={(6.5,-2.9)}]
\draw[thick,domain=45:135, xshift=0] plot ({0.7*cos(\x)}, {0.7*sin(\x)});
\draw[line width=0.1cm, white, rotate around={-15:(0,0.1)}] (0,0.2) -- (0,1.1);
\draw[-stealth, rotate around={-15:(0,0.1)}] (0,0.2) -- (0,1.1);
\end{scope}

\begin{scope}[shift={(8.5,-2.9)}]
\draw[thick,domain=45:135, xshift=0] plot ({0.7*cos(\x)}, {0.7*sin(\x)});
\draw[line width=0.1cm, white, rotate around={-15:(0,0.1)}] (0,0.2) -- (0,1.1);
\draw[-stealth, rotate around={-15:(0,0.1)}] (0,0.2) -- (0,1.1);
\end{scope}

\end{scope}

\end{scope}


\begin{scope}[yshift=-3cm, scale=0.9]

\begin{scope}
\draw (18:0.5cm) -- (90:0.5cm) -- (18+2*72:0.5cm) --
(18+3*72:0.5cm) -- (18+4*72:0.5cm) -- cycle;
\filldraw (18:0.5cm) circle (1.2pt);
\filldraw (18+72:0.5cm) circle (1.2pt);
\filldraw (18+2*72:0.5cm) circle (1.2pt);
\filldraw (18+3*72:0.5cm) circle (1.2pt);
\filldraw (18+4*72:0.5cm) circle (1.2pt);

\draw[-stealth, color3, thick] (30:0.58cm) to[out=118, in=-10] (78:0.58cm);
\draw[-stealth, color3, thick, rotate=72] (30:0.58cm) to[out=118, in=-10] (78:0.58cm);
\draw[-stealth, color3, thick, rotate=2*72] (30:0.58cm) to[out=118, in=-10] (78:0.58cm);
\draw[-stealth, color3, thick, rotate=3*72] (30:0.58cm) to[out=118, in=-10] (78:0.58cm);
\draw[-stealth, color3, thick, rotate=4*72] (30:0.58cm) to[out=118, in=-10] (78:0.58cm);

\node at (0,-1.2) {$\purirank(\rho) \leq r$};

\node at (1.35,0) {$\Longleftrightarrow$};

\node at (-1.2,0) {(b)};

\end{scope}

\begin{scope}[xshift=2.4cm, yshift=0.4cm, scale=0.38]

\draw[fill=color1, draw=none] (-1,-0.5) rectangle (11.5,2.4);
\node at (10.5, 0.5) {$\ket{\psi}$};


%

\draw[color3, thick, -stealth] (0.7,1.7) to[in=160,out=20] (2.3,1.7);
\draw[color3, thick, -stealth] (2.7,1.7) to[in=160,out=20] (4.3,1.7);
\draw[color3, thick, -stealth] (4.7,1.7) to[in=160,out=20] (6.3,1.7);
\draw[color3, thick, -stealth] (6.7,1.7) to[in=160,out=20] (8.3,1.7);

\draw[thick, rounded corners] (-0.5,0.5) -- (3.5,0.5) node[midway,below, yshift=-0.05cm] {\scriptsize $r$} -- (9.5,0.5) -- (9.5,1.5) -- (-0.5,1.5) -- cycle;

\draw[-stealth] (0.5,0.5) -- (0.5,-1.5);
\draw[-stealth] (2.5,0.5) -- (2.5,-1.5);
\draw[-stealth] (4.5,0.5) -- (4.5,-1.5);
\draw[-stealth] (6.5,0.5) -- (6.5,-1.5);
\draw[-stealth] (8.5,0.5) -- (8.5,-1.5);

\draw[-stealth] (0.5,-3) -- (0.5,-4);
\draw[-stealth] (2.5,-3) -- (2.5,-4);
\draw[-stealth] (4.5,-3) -- (4.5,-4);
\draw[-stealth] (6.5,-3) -- (6.5,-4);
\draw[-stealth] (8.5,-3) -- (8.5,-4);

\draw[decorate, decoration={brace}] (8.75,-4.2) -- (0.25,-4.2) node[midway, below, yshift=-0.1cm, color=black] {$\rho$};

\draw[fill=white] (0,0) rectangle (1,1);
\draw[fill=white] (2,0) rectangle (3,1);
\draw[fill=white] (4,0) rectangle (5,1);
\draw[fill=white] (6,0) rectangle (7,1);
\draw[fill=white] (8,0) rectangle (9,1);

\draw[fill=white] (-0.25,-3) rectangle (1.25,-1.5) node[midway] {$\mathcal{E}_1$};
\draw[fill=white] (1.75,-3) rectangle (3.25,-1.5) node[midway] {$\mathcal{E}_2$};
\draw[fill=white] (3.75,-3) rectangle (5.25,-1.5) node[midway] {$\mathcal{E}_3$};
\draw[fill=white] (5.75,-3) rectangle (7.25,-1.5) node[midway] {$\mathcal{E}_4$};
\draw[fill=white] (7.75,-3) rectangle (9.25,-1.5) node[midway] {$\mathcal{E}_5$};
\end{scope}

\end{scope}

\end{tikzpicture}
\caption{\textbf{Implications of border ranks for correlations}.  (a) The translational invariant (t.i.) cyclic psd-rank characterizes the minimal bond dimension to generate this distribution via a t.i.\ MPS together with local measurements, where $P$ is a nonnegative tensor representing an $n$-partite probability distribution. The gap between rank and border rank implies that the set of probability distributions generated in such a way is not closed.
The same applies when replacing the cyclic graph and translational invariance by other decomposition geometries and symmetries (see \Cref{thm:quantumCorrMainText}). (b) Gaps between border rank and rank also imply that the set of $n$-partite density matrices arising via MPS of bond dimension $r$ and local quantum channels is not closed. }
\label{fig:quantum_correlation}
\end{figure}

This paper is structured as follows. In \Cref{sec:posTensorDec}, we review positive tensor decompositions on nonnegative tensors and positive semidefinite multipartite matrices. 
In \Cref{sec:results}, we present the results summarized in \Cref{fig:borderRankResults}. In \Cref{sec:applications}, we provide applications of positive tensor decompositions and implications of the existence of border ranks. 
In \Cref{sec:Conclusion}, we conclude and discuss open questions.  

The appendix is structured as follows. In \Cref{app:OmegaG-decompositions}, we review positive and invariant tensor decompositions on hypergraph structures, so-called $(\Omega,G)$-decompositions. In \Cref{app:rankInequalities}, \Cref{app:NonnegativeDec_CausalStructures}, \Cref{app:CorrelationCorrespondenceAppendix}, and \Cref{app:nnDecNoBorderRank}, we present proofs of lemmas used in the main text. In \Cref{app:POVMQChannel} and \Cref{app:BolzanoWeierstrass}, we review basic facts of the topology of quantum channels and positive operator-valued measures.

\section{Positive tensor decompositions}
\label{sec:posTensorDec}

 In the following we review several tensor decomposition on entrywise nonnegative tensors by varying them in three different aspects: local positivity (\Cref{ssec:multipartite}), arrangement of indices (\Cref{ssec:arbitraryArrangements}), and symmetry (\Cref{ssec:symmetricDec}). A summary of all the different notions of ranks together with their abbreviations is provided in \Cref{fig:rank_abbrv}.
We will later generalize these decompositions to multipartite positive semidefinite (psd) matrices (\Cref{ssec:matrixDec}). Finally, we introduce in \Cref{ssec:borderRank} the notion of a border rank.

\begin{figure*}[thb]\centering
\begin{tikzpicture}[scale=1.04]

\draw (0,1.8) -- (3,0);

\draw (0,0) -- (11,0);
\draw (1,-0.5) -- (11,-0.5);
\draw (1,-1) -- (11,-1);

\draw (3,1.5) -- (3,-1.5);
\draw (5,1.5) -- (5,-1.5);
\draw (7,1.5) -- (7,-1.5);
\draw (9,1.5) -- (9,-1.5);

\node[anchor=east] at (2.9,-0.25) {\footnotesize no constraint};

\node[anchor=east] at (2.9,-0.75) {\footnotesize positive semidefinite};

\node[anchor=east] at (2.9,-1.25) {\footnotesize nonnegative};

\node[anchor=west] at (-0.2,0.8) {\footnotesize \textsc{Positivity}};
\node[anchor=west] at (-0.2,0.4) {\footnotesize \textsc{constraint}};

\node[anchor=east] at (2.9,1.7) {\footnotesize \textsc{Decomposition}};
\node[anchor=east] at (2.9,1.3) {\footnotesize \textsc{type}};

\node at (4,1.5) {\footnotesize Standard};
\node at (6,1.5) {\footnotesize Symmetric};
\node at (8,1.5) {\footnotesize Cyclic};
\node at (10,1.7) {\footnotesize Translational};
\node at (10,1.4) {\footnotesize invariant};

\begin{scope}[xshift=4cm , yshift=0.7cm, scale=0.75]
\draw[fill=gray!50!white] (18:0.5cm) -- (90:0.5cm) -- (18+2*72:0.5cm) --
(18+3*72:0.5cm) -- (18+4*72:0.5cm) -- cycle;
\filldraw (18:0.5cm) circle (1.2pt);
\filldraw (18+72:0.5cm) circle (1.2pt);
\filldraw (18+2*72:0.5cm) circle (1.2pt);
\filldraw (18+3*72:0.5cm) circle (1.2pt);
\filldraw (18+4*72:0.5cm) circle (1.2pt);
\end{scope}

\begin{scope}[xshift=6cm , yshift=0.7cm, scale=0.75]
\draw[fill=gray!50!white] (18:0.5cm) -- (90:0.5cm) -- (18+2*72:0.5cm) --
(18+3*72:0.5cm) -- (18+4*72:0.5cm) -- cycle;
\filldraw (18:0.5cm) circle (1.2pt);
\filldraw (18+72:0.5cm) circle (1.2pt);
\filldraw (18+2*72:0.5cm) circle (1.2pt);
\filldraw (18+3*72:0.5cm) circle (1.2pt);
\filldraw (18+4*72:0.5cm) circle (1.2pt);

\draw[stealth-stealth, color3, thick] (30:0.58cm) to[out=118, in=-10] (78:0.58cm);
\draw[stealth-stealth, color3, thick, rotate=72] (30:0.58cm) to[out=118, in=-10] (78:0.58cm);
\draw[stealth-stealth, color3, thick, rotate=2*72] (30:0.58cm) to[out=118, in=-10] (78:0.58cm);
\draw[stealth-stealth, color3, thick, rotate=3*72] (30:0.58cm) to[out=118, in=-10] (78:0.58cm);
\draw[stealth-stealth, color3, thick, rotate=4*72] (30:0.58cm) to[out=118, in=-10] (78:0.58cm);
\end{scope}

\begin{scope}[xshift=8cm , yshift=0.7cm, scale=0.75]
\draw (18:0.5cm) -- (90:0.5cm) -- (18+2*72:0.5cm) --
(18+3*72:0.5cm) -- (18+4*72:0.5cm) -- cycle;
\filldraw (18:0.5cm) circle (1.2pt);
\filldraw (18+72:0.5cm) circle (1.2pt);
\filldraw (18+2*72:0.5cm) circle (1.2pt);
\filldraw (18+3*72:0.5cm) circle (1.2pt);
\filldraw (18+4*72:0.5cm) circle (1.2pt);
\end{scope}

\begin{scope}[xshift=10cm , yshift=0.7cm, scale=0.75]
\draw (18:0.5cm) -- (90:0.5cm) -- (18+2*72:0.5cm) --
(18+3*72:0.5cm) -- (18+4*72:0.5cm) -- cycle;
\filldraw (18:0.5cm) circle (1.2pt);
\filldraw (18+72:0.5cm) circle (1.2pt);
\filldraw (18+2*72:0.5cm) circle (1.2pt);
\filldraw (18+3*72:0.5cm) circle (1.2pt);
\filldraw (18+4*72:0.5cm) circle (1.2pt);

\draw[-stealth, color3, thick] (30:0.58cm) to[out=118, in=-10] (78:0.58cm);
\draw[-stealth, color3, thick, rotate=72] (30:0.58cm) to[out=118, in=-10] (78:0.58cm);
\draw[-stealth, color3, thick, rotate=2*72] (30:0.58cm) to[out=118, in=-10] (78:0.58cm);
\draw[-stealth, color3, thick, rotate=3*72] (30:0.58cm) to[out=118, in=-10] (78:0.58cm);
\draw[-stealth, color3, thick, rotate=4*72] (30:0.58cm) to[out=118, in=-10] (78:0.58cm);
\end{scope}

%


\begin{scope}[yshift=-0.25cm, xshift=4cm]
\node at (0, 0) {\scriptsize $\rank$};
\end{scope}

\begin{scope}[yshift=-0.25cm, xshift=6cm]
\node at (0, 0) {\scriptsize $\symmrank$};
\end{scope}

\begin{scope}[yshift=-0.25cm, xshift=8cm]
\node at (0, 0) {\scriptsize $\osr$};
\end{scope}

\begin{scope}[yshift=-0.25cm, xshift=10cm]
\node at (0, 0) {\scriptsize $\tiosr$};
\end{scope}


\begin{scope}[yshift=-0.75cm, xshift=4cm]
\node at (0, 0) {\scriptsize $\psdrank$};
\end{scope}

\begin{scope}[yshift=-0.75cm, xshift=6cm]
\node at (0, 0) {\scriptsize $\symmpsdrank$};
\end{scope}

\begin{scope}[yshift=-0.75cm, xshift=8cm]
\node at (0, 0) {\scriptsize $\psdosr$};
\end{scope}

\begin{scope}[yshift=-0.75cm, xshift=10cm]
\node at (0, 0) {\scriptsize $\tipsdosr$};
\end{scope}


\begin{scope}[yshift=-1.25cm, xshift=4cm]
\node at (0, 0) {\scriptsize $\nnrank$};
\end{scope}

\begin{scope}[yshift=-1.25cm, xshift=6cm]
\node at (0, 0) {\scriptsize $\symmnnrank$};
\end{scope}

\begin{scope}[yshift=-1.25cm, xshift=8cm]
\node at (0, 0) {\scriptsize $\nnosr$};
\end{scope}

\begin{scope}[yshift=-1.25cm, xshift=10cm]
\node at (0, 0) {\scriptsize $\tinnosr$};
\end{scope}

\end{tikzpicture}
\captionsetup{width=.7\linewidth, margin=0.5cm}
\caption{\textbf{Abbreviations of all the different notions of ranks.} The columns represent different arrangements of indices in the decomposition. The rows distinguish between the positivity constraints imposed on the decomposition. The term $\osr$ stands for the \emph{operator Schmidt rank}, the abbreviation ti stands for \emph{translational invariant}, and symm stands for \emph{symmetric}.}
\label{fig:rank_abbrv}
\end{figure*}

Bipartite locally positive tensor decompositions arise from matrix factorizations. For instance, given a matrix $T \in \mathcal{M}_d(\mathbb{C}) \cong \mathbb{C}^d \otimes \mathbb{C}^d$, $T$ has rank at most $r$ if there exists a decomposition
\begin{equation} 
\label{eq:MatrixRankTensorRank}
T = \sum_{\alpha = 1}^{r} \ket{v_{\alpha}} \bra{\bar{w}_{\alpha}},
\end{equation}
where $\bra{\bar{w}_{\alpha}} \coloneqq \ket{w_{\alpha}}^t$.
Equivalently, seeing $T$ as an element in $\mathbb{C}^d \otimes \mathbb{C}^d$, the matrix factorization translates into
$$T = \sum_{\alpha = 1}^{r} \ket{v_{\alpha}} \otimes \ket{w_{\alpha}}.$$

Restricting to entrywise nonnegative vectors, i.e.\ $\Cbraket{j}{v_{\alpha}}, \Cbraket{j}{w_{\alpha}} \geq 0$ for every $j \in \{1,\ldots,d\}$, we obtain the nonnegative matrix factorization \cite{Co93b}. We denote the minimal $r$ attained by such a decomposition as $\nnrank(T)$. Replacing the nonnegative vectors by psd matrices, we obtain the psd matrix factorization, given by
$$ T_{ij} = \sum_{\alpha, \beta = 1}^r \left(A_i\right)_{\alpha, \beta} \cdot \left(B_j\right)_{\alpha, \beta} = \tr\left(A_i B_j^{t}\right)$$
where $A_i, B_j \in \mathcal{M}_{r}^{+}(\mathbb{C})$ are hermitian psd matrices of size $r \times r$. The positive semidefinite rank $\psdrank(T)$ is the minimum size $r \in \mathbb{N}$ necessary to generate such a decomposition. 

These three ranks are related via the inequalities (see \cite{Fa14})
\begin{equation}\label{eq:rank_ineq1}
\rank(T) \leq \psdrank(T)^2
\end{equation}
and
\begin{equation}\label{eq:rank_ineq2}
\psdrank(T), \rank(T) \leq \nnrank(T).
\end{equation}
Note that inequalities in the other direction do not exist in general without including the dimension $d$.

\subsection{Multipartite tensor decompositions}
\label{ssec:multipartite}

Bipartite decompositions generalize to multipartite ones by considering more vectors and keeping the sum over a single index, i.e.\
$$ T = \sum_{\alpha=1}^{r} \ket{v_{\alpha}} \otimes \ket{w_{\alpha}} \otimes \cdots \otimes \ket{z_{\alpha}}.$$
With the additional constraint of entrywise nonnegative local vectors we obtain the \emph{standard nonnegative} multipartite tensor decomposition. The psd-decomposition generalizes to the multipartite case as
\begin{equation}
\label{eq:SigmaNpsdDec}
T_{j_1,\ldots, j_n} =  \sum_{\alpha, \beta = 1}^r \left(A_{j_1}\right)_{\alpha, \beta} \cdot \left(B_{j_2}\right)_{\alpha, \beta} \cdots \left(Z_{j_n}\right)_{\alpha, \beta},
\end{equation}
where $A_j, B_j, \ldots, Z_j \in \mathcal{M}_r^{+}(\mathbb{C})$ are psd \cite{Ja17, De19d}.

By introducing multiple summation indices, we obtain various decompositions of multipartite tensors into local elements. For example, a cyclic decomposition is given by
\begin{equation}
\label{eq:CyclicDec}
T = \sum_{\alpha_1, \ldots, \alpha_n = 1}^{r} \ket{x_{\alpha_1, \alpha_2}} \otimes \ket{y_{\alpha_2, \alpha_3}} \otimes \cdots \otimes \ket{z_{\alpha_n, \alpha_1}},
\end{equation}
This is known as a Matrix Product State (MPS) decomposition \cite{Pe07}. Restricting again to nonnegative tensors gives rise to the nonnegative cyclic decompositions \cite{De19d, Gl19}, also called Stochastic Matrix Product State \cite{Te10}.
Similarly,
\begin{equation}
\label{eq:CyclicPsdDecs}
T_{j_1, \ldots, j_n} = \hspace*{-0.1cm} \sum_{\alpha_i,\beta_i = 1}^{r} \hspace*{-0.1cm} \left(A_{j_1}\right)_{\alpha_1, \alpha_2; \beta_1, \beta_2} \cdots \left(Z_{j_n}\right)_{\alpha_n, \alpha_1; \beta_n, \beta_1},
\end{equation}
where $A_i, \ldots, Z_i \in \mathcal{M}_{r^2}^{+}(\mathbb{C})$ are psd, defines the cyclic psd-decomposition \cite{De19,Gl19}.

Analogously to the bipartite case, these notions of ranks are related via the inequalities in Equations \eqref{eq:rank_ineq1} and \eqref{eq:rank_ineq2}.

\subsection{Arbitrary arrangements of indices}
\label{ssec:arbitraryArrangements}

More generally, one can define any arrangement of indices with a hypergraph $\Omega$ by associating the nodes to the local spaces and the (hyper-)edges to the summation indices. This includes all so far mentioned decompositions (see \Cref{fig:hypergraphDec} for examples). Given a hypergraph $\Omega$, we define the corresponding ranks as $\rank_\Omega, \nnrank_\Omega, \psdrank_\Omega$ \cite{De19d}. One particular subclass of decompositions considered in this paper arises from trees (i.e.\ connected graphs containing no loops). Note that the notion of a tree has no meaning for hypergraphs. While tree tensor decompositions do not show a gap between border rank and rank, some positive tensor decompositions exhibit a gap already for the simplest hypergraph, namely a simplex. For a detailed treatment of $\Omega$-decompositions we refer to \Cref{app:OmegaG-decompositions}, where $\Omega$-decompositions are defined via \emph{weighted simplicial complexes} $\Omega$ representing the hypergraph structure.

\begin{figure}
\begin{tikzpicture}

\begin{scope}[scale=0.7]
\draw[fill=gray!50!white] (0,0) -- (1,0) -- (0.5,0.86) -- cycle;

\node at (0.5,0.29) {\tiny $\alpha$};
\draw[fill=black] (0,0) circle (1.7pt);
\draw[fill=black] (1,0) circle (1.7pt);
\draw[fill=black] (0.5,0.86) circle (1.7pt);

\node at (2.5,0.46) {$\displaystyle \sum_{\alpha=1}^{r}$};
\node[anchor=west] at (3,0.46){$ \ket{x_{\alpha}} \otimes  \ket{y_{\alpha}} \otimes \ket{z_{\alpha}} $};
\node at (-1,0.4) {(a)};

\end{scope}

\begin{scope}[yshift=-2cm, scale=0.7]
\draw (0,0) -- (1,0) -- (0.5,0.86) -- cycle;

\draw[fill=black] (0,0) circle (1.7pt);
\draw[fill=black] (1,0) circle (1.7pt);
\draw[fill=black] (0.5,0.86) circle (1.7pt);

\node at (2.5,0.46) {$\displaystyle \sum_{\alpha,\beta,\gamma=1}^{r}$};
\node[anchor=west] at (3,0.46) {$\ket{x_{\alpha,\beta}} \otimes  \ket{y_{\beta,\gamma}} \otimes \ket{z_{\gamma, \alpha}} $};
\node at (-1,0.4) {(b)};

\node at (0.5,-0.25) {\tiny $\alpha$};
\node at (0.95,0.5) {\tiny $\beta$};
\node at (0.05,0.5) {\tiny $\gamma$};

\end{scope}

\begin{scope}[yshift=-4cm, scale=0.7]

\draw (-0.3,0) -- (0.45,0.3);
\draw (1.2,0.6) -- (1.45,0.7);
\draw[dotted] (0.45,0.3) -- (1.2,0.6);

\node at (-0.18,0.25) {\tiny $\alpha_1$};
\node at (0.7,0.1) {\tiny $\alpha_2$};
\node at (0.65,0.7) {\tiny $\alpha_{n-1}$};

\draw[fill=black] (-0.3,0) circle (1.7pt);
\draw[fill=black] (0.2,0.2) circle (1.7pt);
\draw[fill=black] (1.45,0.7) circle (1.7pt);


\node at (2.5,0.46) {$\displaystyle \sum_{\alpha_1, \ldots, \alpha_{n-1}=1}^{r}$};
\node[anchor=west] at (3,0.5) {$ \ket{x_{\alpha_1}} \otimes  \ket{y_{\alpha_1, \alpha_2}} \otimes \cdots \otimes \ket{z_{\alpha_{n-1}}} $};
\node at (-1,0.4) {(c)};

\end{scope}

\begin{scope}[yshift=-6cm, scale=0.7]
\draw (0.5,0.86) -- (-0.2,0);
\draw (0.5,0.86) -- (0.5,0);
\draw (0.5,0.86) -- (1.2,0);

\draw[fill=black] (-0.2,0) circle (1.7pt);
\draw[fill=black] (0.5,0) circle (1.7pt);
\draw[fill=black] (1.2,0) circle (1.7pt);
\draw[fill=black] (0.5,0.86) circle (1.7pt);

\node at (2.5,0.46) {$\displaystyle \sum_{\alpha,\beta,\gamma=1}^{r}$};
\node[anchor=west] at (3,0.46) {$ \ket{x_{\alpha}} \otimes  \ket{y_{\beta}} \otimes \ket{z_{\gamma}} \otimes \ket{w_{\alpha,\beta,\gamma}} $};
\node at (-1,0.4) {(d)};

\node at (0,0.5) {\tiny $\alpha$};
\node at (1,0.5) {\tiny $\gamma$};
\node at (0.67,0.3) {\tiny $\beta$};

\end{scope}


\end{tikzpicture}
\caption{\textbf{Examples of decomposition types.} Hypergraph structures $\Omega$ give rise to different tensor decompositions. Each vertex corresponds to a local space and each (hyper-)edge corresponds to an index in the decomposition. If $\Omega$ is the $n$-simplex, we obtain standard tensor decomposition (a), for the $n$-cycle, we obtain the cyclic decomposition (b), for a line of length $n$ we get the decomposition in (c) and for a tree, we obtain a tree tensor network (d).}
\label{fig:hypergraphDec}
\end{figure}

\subsection{Symmetric decompositions}
\label{ssec:symmetricDec}

Many tensor decompositions admit a symmetrized version, where we define a symmetric decomposition by constraining the local elements to being the same. For example, the symmetric tensor decomposition is given by
$$ T = \sum_{\alpha = 1}^{r} \ket{v_{\alpha}} \otimes \ket{v_{\alpha}} \otimes \cdots \otimes \ket{v_{\alpha}}.$$
A psd-decomposition as given in Equation \eqref{eq:SigmaNpsdDec} is called symmetric if all local matrices $A_{j}, B_{j},\ldots$ are identical
for every fixed $j$. In this situation, we denote the corresponding ranks by $\symmrank$ for the unconstrained, $\symmpsdrank$ for the positive semidefinite, and $\symmnnrank$ for the nonnegative decomposition.

For a cyclic decomposition, its translational invariant (t.i.) version is given by
\begin{equation}
\label{eq:CyclicTiDec}
 T = \sum_{\alpha_1, \ldots, \alpha_n = 1}^{r} \ket{v_{\alpha_1, \alpha_2}} \otimes \cdots \otimes \ket{v_{\alpha_n, \alpha_1}}
\end{equation}
Similarly, a t.i.\ psd-decomposition is given by a decomposition as in Equation \eqref{eq:CyclicPsdDecs} where all local matrices $A_j, B_j, \ldots$ are identical for every fixed $j$.
We denote the corresponding ranks by $\tiosr$, $\tipsdosr$, and $\tinnosr$.

Symmetric decompositions generalize to arbitrary structures $\Omega$ by introducing a group action $G$ on the set of vertices $\{1, \ldots, n\}$ that is compatible with the hypergraph structure (see \Cref{fig:symmetricDec} for an example). This gives rise to $(\Omega,G)$-decompositions, each with its corresponding rank. We refer to \Cref{app:OmegaG-decompositions} for further details.

\begin{figure}
\centering
\begin{tikzpicture}[scale=0.8]

\draw (18:0.5cm) -- (90:0.5cm) -- (18+2*72:0.5cm) --
(18+3*72:0.5cm) -- (18+4*72:0.5cm) -- cycle;
\filldraw (18:0.5cm) circle (1.2pt);
\filldraw (18+72:0.5cm) circle (1.2pt);
\filldraw (18+2*72:0.5cm) circle (1.2pt);
\filldraw (18+3*72:0.5cm) circle (1.2pt);
\filldraw (18+4*72:0.5cm) circle (1.2pt);

\draw[-stealth, color3, thick] (30:0.58cm) to[out=118, in=-10] (78:0.58cm);
\draw[-stealth, color3, thick, rotate=72] (30:0.58cm) to[out=118, in=-10] (78:0.58cm);
\draw[-stealth, color3, thick, rotate=2*72] (30:0.58cm) to[out=118, in=-10] (78:0.58cm);
\draw[-stealth, color3, thick, rotate=3*72] (30:0.58cm) to[out=118, in=-10] (78:0.58cm);
\draw[-stealth, color3, thick, rotate=4*72] (30:0.58cm) to[out=118, in=-10] (78:0.58cm);

\draw[-stealth, color4, thick] (-82:0.35cm) to[out=55, in=-163] (-26:0.35cm);
\draw[-stealth, color4, thick, rotate=72] (-82:0.35cm) to[out=55, in=-163] (-26:0.35cm);
\draw[-stealth, color4, thick, rotate=2*72] (-82:0.35cm) to[out=55, in=-163] (-26:0.35cm);
\draw[-stealth, color4, thick, rotate=3*72] (-82:0.35cm) to[out=55, in=-163] (-26:0.35cm);
\draw[-stealth, color4, thick, rotate=4*72] (-82:0.35cm) to[out=55, in=-163] (-26:0.35cm);

\node[anchor=west] at (0.7,0) {$\displaystyle \sum_{\alpha_1, \ldots, \alpha_n = 1}^{r} \ket{v_{\alpha_1, \alpha_2}} \otimes \ket{v_{\alpha_2, \alpha_3}} \otimes \cdots \otimes \ket{v_{\alpha_n, \alpha_1}}$};

\draw[-stealth, color4, thick] (3.8,0.3) to[out=45, in=135] (5,0.3);
\draw[-stealth, color4, thick] (5.9,0.3) to[out=45, in=135] (7,0.3);
\draw[-stealth, color4, thick] (8.5,-0.35) to[out=-155, in=-25] (3.8,-0.35);

\end{tikzpicture}
\caption{A group action on a hypergraph $\Omega$ is a group action on the vertices (red), that induces a well-defined group action on the edges (green). The group action gives rise to constraints on the local tensors. In this example the group action gives rise to the constraint $\ket{v_{\alpha,\beta}^{[i]}} = \ket{v_{\alpha,\beta}^{[j]}}$ for $i, j \in \{1,\ldots,n\}$ (blue).}
\label{fig:symmetricDec}
\end{figure}

\subsection{Positive and invariant decompositions of multipartite psd matrices}
\label{ssec:matrixDec}

There is a similar hierarchy of tensor decompositions for psd matrices $\rho \in \mathcal{M}_d(\mathbb{C})^{\otimes n}$ containing the unconstrained tensor decomposition\footnote{An unconstrained tensor decomposition refers to the fact that the decomposition contains no positivity constraints in its local terms.}
\begin{equation}
\label{eq:MatrixDec}
 \rho = \sum_{\alpha = 1}^r A_{\alpha} \otimes B_{\alpha} \otimes \cdots \otimes C_{\alpha}
\end{equation}
where $A_{\alpha}, B_{\alpha}, C_{\alpha} \in \mathcal{M}_d(\mathbb{C})$,
the separable decomposition, where $A_{\alpha}, B_{\alpha}, C_{\alpha}$ are restricted to being psd matrices, and the local purification form, which consists of a factorization $\rho = L L^{\dag}$ together with an unconstrained tensor decomposition of $L$. We denote the corresponding ranks by $\rank(\rho)$, $\seprank(\rho)$, and $\purirank(\rho)$. In the unconstrained case and the separable case, $\rank(\rho)$ and $\seprank(\rho)$ are the minimal $r \in \mathbb{N}$ such that a decomposition of the form \eqref{eq:MatrixDec} exists. The purification rank $\purirank(\rho)$ is the minimal $\rank(L)$ among all purifications of $\rho$. One can define arbitrary $\Omega$-decompositions by arranging the indices of the local matrices differently. We refer to \Cref{app:OmegaG-decompositions} for details.

The hierarchy of positive tensor decompositions embeds in the framework of types of matrix tensor decompositions as follows. Choosing
\begin{equation}
\label{eq:rhoT}
\rho_T \coloneqq \sum_{j_1, \ldots, j_n=1}^{d} T_{j_1, \ldots, j_n} \ket{j_1, \ldots, j_n}\bra{j_1, \ldots, j_n},
\end{equation}
every unconstrained tensor decomposition of $\rho_T$ can be made into an unconstrained tensor decomposition of $T$ by ignoring the off-diagonal terms. The same is true for the separable and nonnegative decomposition as well as the positive semidefinite decomposition and the purification form. In particular, we have that
\begin{align*} \rank(T) &= \rank(\rho_T), \\ \nnrank(T) &= \seprank(\rho_T), \\ \psdrank(T) &= \purirank(\rho_T).
\end{align*}
including all different arrangements of indices as well as symmetries. For a detailed exposition and a proof of the correspondence we refer to \cite[Theorem 38]{De19} and \cite[Theorem 43]{De19d}.

If a certain type of rank for nonnegative tensors exhibits a gap between border rank and rank, then this gap is also present for the corresponding matrix tensor rank. For example, for any tensor $T$ with $\bpsdrank(T) < \psdrank(T)$ we have that
\begin{align*}
\bpurirank(\rho_T) &\leq \bpsdrank(T) \\ &< \psdrank(T) \\ &= \purirank(\rho_T)
\end{align*}
where we have used the correspondence above in the first and the last step.

\subsection{Border ranks}
\label{ssec:borderRank}

The rank of a tensor relates to complexity measures in various descriptions \cite{La11, Fa22, Ja17}. Are these complexities robust to approximations? In other words, given an approximation error $\varepsilon > 0$, how does the smallest rank of an approximation
$$ \rank^{\varepsilon}(T) \coloneqq \inf_{\Vert M - T \Vert < \varepsilon} \rank(M)$$
relate to its exact counterpart $\rank(T)$? For a fixed approximation error, Ref. \cite{De20} shows that for certain $\ell_p$-norms and Schatten $p$-norms, the ranks of the corresponding decompositions
$\nnrank^{\varepsilon}$ and $\seprank^{\varepsilon}$ can be upper bounded independently of the system dimension, in contrast to their exact counterparts. However, these bounds depend on $\varepsilon$ and therefore diverge if $\varepsilon \to 0$.

The border rank \cite{La11, Bi80,Zu17} is an asymptotic measure for the approximate rank\footnote{Note that the function $\varepsilon \mapsto \textrm{rank}^{\varepsilon}(T)$
 is constant for  $\varepsilon$ small enough.}, i.e.\
$$ \brank(T) = \lim_{\varepsilon \to 0} \rank^{\varepsilon}(T).$$
That is, $\brank(T) \leq r$ if and only if there exists a converging sequence of tensors $T_n \to T$ such that $\rank(T_n) \leq r$. By definition, we have $$\rank^{\varepsilon}(T) \leq \brank(T) \leq \rank(T)$$ for every fixed approximation error $\varepsilon > 0$. Note that the border rank is independent of the choice of the norms for convergence since all norms are equivalent in finite dimensional vector spaces.\footnote{That is, there exist constants $A, B > 0$, such that $A \Vert \cdot \Vert_1 \leq  \Vert \cdot \Vert_2  \leq B \Vert \cdot \Vert_1$}. Therefore, if $T_n \to T$ converges, it also converges in any other norm.
Moreover, the border rank is independent of normalization. If $\brank(T) = r$, then we also have $\brank(T/\Vert T\Vert) = r$ for any norm.

The border rank of all other notions of ranks presented in this section is defined similarly. We denote them by placing an underline on the usual identifier of the corresponding rank.

It is known that in the case of the matrix rank (i.e.\ the bipartite tensor rank), we have that
$$\brank(T) = \rank(T)$$ for every bipartite tensor $T$. This follows from the fact that the rank corresponds to the size of the largest submatrix of $T$ with non-zero determinant and that the determinant is continuous. Alternatively, the statement follows from the Eckart--Young approximation theorem \cite{Ec36}.

However, for multipartite tensor decompositions, there are tensors $T$ such that
$$\brank(T) < \rank(T).$$
Whenever such a situation appears, we will say that there is a \emph{gap between rank and border rank}. Moreover, if there exists a sequence of multipartite tensors $(T_n)_{n \in \mathbb{N}}$ with $T_n \in (\mathbb{C}^d)^{\otimes n}$ such that $\brank(T_n)$ is constant and $\rank(T_n) \to \infty$, we will say that there is a \emph{border rank separation}.

Results on the existence and non-existence of gaps between border ranks and ranks include unconstrained tensor decompositions with various arrangement of indices \cite{La11,Ch20, Ba21}, bipartite locally positive decompositions \cite{Fa14}, and the standard nonnegative tensor decomposition \cite{Qi16}. For completeness, we give a brief overview of these results in the next section.

\section{Gaps between rank and border rank}
\label{sec:results}

Here we present the results summarized in \Cref{fig:borderRankResults}. Throughout, gaps between ranks and border ranks are established by giving explicit examples of tensors exhibiting them.

\subsection{Standard tensor decompositions}
\label{ssec:standardBorder}
Since the matrix rank does not exhibit a gap between border rank and rank, systems of size $n=3$ are the smallest examples with a gap between border rank and rank. While this has been studied for the standard and symmetric tensor rank, we extend this investigations in this subsection also to positive tensor decompositions for nonnegative tensors and positive semidefinite matrices.

For the standard (unconstrained) tensor decomposition, the unnormalized $n$-partite $W$-state
$$W_n \coloneqq \ket{0,\ldots, 0,1} + \ket{0,\ldots, 1,0} + \ldots + \ket{1,0, \ldots, 0}$$
exhibits a gap between border rank and rank as well as between symmetric border rank and rank for system sizes $n \geq 3$. Specifically, for $\varepsilon > 0$, the family of tensors
\begin{equation} 
\label{eq:W-BorderRank}
W_{n}^{\varepsilon} = \frac{1}{\varepsilon}\left(\ket{0} + \varepsilon \ket{1}\right)^{\otimes n} - \frac{1}{\varepsilon} \ket{0, \ldots, 0}
\end{equation}
imply that
$$\bsymmrank(W_n) = \brank(W_n) = 2$$
since $\lim_{\varepsilon \to 0} W_{n}^{\varepsilon} = W_{n}$.
On the other hand, ${\rank(W_n) = n}$ (see \Cref{prop:WnLoweBound} in \Cref{app:rankInequalities} for a self-contained proof).\footnote{Alternatively, this statement can be shown for the symmetric rank by using Sylvester's theorem for binary forms \cite{Sy52, Co11}. The statement for the standard tensor rank follows by the validity of Comon's conjecture for tensors with border rank $2$ \cite{Ba13}.}

We now show that the $W_n$-tensor exhibits a gap between rank and border rank for the psd-rank. For this purpose, we first present a lower bound of the $\psdrank$ using the following lemma and then compute an upper bound of the border rank. 

\begin{lemma}
Let $T$ be a nonnegative tensor. Then,
\label{lem:rank-psdrank-bound}
$$ \rank(T) \leq \psdrank(T)^2$$
The same bound holds for arbitrary arrangements of indices $\Omega$ with arbitrary symmetry-constraints.
\end{lemma}

We refer to \Cref{lem:rank-psdrank-bound_general} in \Cref{app:rankInequalities} for a proof. \Cref{lem:rank-psdrank-bound} implies that $\psdrank(W_5) \geq 3$ as well as $\psdrank(W_n) = \Omega(\sqrt{n})$.\footnote{i.e.\ $\psdrank$ is asymptotically lower bounded by $D \cdot \sqrt{n}$ for some constant $D$} It is not known if this lower bound is tight.

For $\varepsilon > 0$ the family of tensors $\widetilde{W}_n^{\varepsilon}$ defined by psd matrices
\renewcommand{\arraystretch}{1}
\begin{equation} \label{eq:defAmatrices} A^{\varepsilon}_0 = \frac{C}{\sqrt[n-1]{\varepsilon}} \left(\begin{array}{cc} 1 & e^{\frac{i \pi}{n}} \\ e^{-\frac{i \pi}{n}} & 1 \end{array}\right), \quad A^{\varepsilon}_1 = \varepsilon \left(\begin{array}{rr} 1 & 1 \\ 1 & 1 \end{array}\right)
\end{equation}
for a suitable constant $C \in \mathbb{R}$ provides an arbitrarily close approximation of $W_n$ which proves that $$\bpsdrank(W_n) = \bsymmpsdrank(W_n) = 2.$$
In other words, there is a border rank separation for $n$-partite psd-decompositions with $n \geq 5$.

In the $3$-partite case, a dimension analysis shows that $\symmpsdrank(W_3) \geq 3$, thereby also leading to a gap between border rank and rank (see \Cref{prop:symmetryAnalysisWState} in the appendix). However, for the non-symmetric case in the tripartite scenario the existence of a gap between border rank and rank is still open (see \Cref{conj:psdrank} and \Cref{sec:Conclusion}).

In contrast to the psd-decomposition, the nonnegative (and subsequently also the separable) decomposition exhibit no gap between border rank and rank in the $n$-partite case for arbitrary $n$, as we will see in \Cref{sec:noBorderRanks}. The special case of nonnegative decompositions has already been shown in \cite{Li09}.

\subsection{Border ranks for cyclic decompositions}
\label{ssec:cyclic}

We now consider the cyclic decomposition with and without translational invariance. Barthel et al.\ \cite{Ba21} show that for the operator Schmidt rank, there is a gap between border rank and rank for the \emph{two-domain state}, given by
\begin{align*}
\tau \coloneqq &\sum_{\alpha=1}^{k} \ket{\alpha, \alpha}^{\otimes n} \\ + &\sum_{i = 0}^{n-2} \sum_{\alpha \neq \beta = 1}^{k} \ket{\alpha, \alpha}^{\otimes i} \otimes \ket{\alpha, \beta} \otimes \ket{\beta, \beta}^{\otimes (n-i-2)} \otimes \ket{\beta, \alpha}.
\end{align*}
In particular, they prove that $\bosr(\tau) \leq k < \osr(\tau)$. 

The construction in \cite{Ba21} also leads to a gap between border rank and rank for nonnegative cyclic decompositions, which we briefly review now. Let $\varepsilon > 0$ and define, for every $\alpha, \beta \in \{1, \ldots, k\}$, the nonnegative vectors
$$\ket{v_{\alpha,\beta}^{\varepsilon}} = \varepsilon \ket{\alpha, \beta} + (1-\varepsilon) \delta_{\alpha,\beta} \ket{\alpha,\beta}$$
where $\delta_{\alpha,\beta}$ is the Kronecker-delta, as well as
$$\ket{w_{\alpha,\beta}^{\varepsilon}} = \delta_{\alpha,\beta} \ket{\alpha, \beta} + \frac{1}{\varepsilon}(1-\delta_{\alpha,\beta}) \ket{\alpha,\beta}.$$
Setting
$$\tau^{\varepsilon} = \sum_{\alpha_i = 1}^{k} \ket{v^{\varepsilon}_{\alpha_1, \alpha_2}} \otimes \ket{v^{\varepsilon}_{\alpha_2, \alpha_3}} \otimes \cdots \otimes \ket{v^{\varepsilon}_{\alpha_{n-1}, \alpha_n}} \otimes \ket{w^{\varepsilon}_{\alpha_n, \alpha_1}}$$
we obtain $\tau_{\varepsilon} = \tau + \mathcal{O}(\varepsilon)$
and therefore $\bnnosr(\tau) \leq k$. This implies the following chain of inequalities
$$\bosr(\tau) \leq \bnnosr(\tau) \leq k < \osr(\tau) \leq \nnosr(\tau)$$
where the strict inequality is shown in \cite[Proposition 5]{Ba21} and the inequalities between $\osr$ and $\nnosr$ hold because the latter is a constrained version of the former.

While \Cref{lem:rank-psdrank-bound} cannot be employed to prove a gap between border psd-rank and psd-rank since the gap between border rank and rank is too small, we will see that a gap also appears for t.i.\ psd-decompositions. For this reason, we expect the same behavior for the non-symmetric case.

\begin{conjecture}
\label{conj:psdrank}
There is a nonnegative tensor $T$ such that $\bpsdosr(T) < \psdosr(T)$.
\end{conjecture}

\subsection{Translational invariant cyclic decomposition}
\label{ssec:tiBorder}

We now prove the existence of gaps between border rank and rank for cyclic t.i.\ decompositions. We obtain border rank separations for all types of decompositions. 

First, for the unconstrained tensor decomposition, the rank of the t.i. decomposition of $W_n$ increases with system size, while $\osr(W_n) = 2$ \cite[Example 22]{De19}. There are several lower bounds regarding the t.i.\ operator Schmidt rank. While P\'erez-Garc\'ia et al. \cite{Pe07} show that $\tiosr(W_n) \geq \Omega(n^{1/3})$ using Wielandt's inequality \cite{Sa09}, the tightest lower bound \cite[Proposition 23]{De19} is given to the best of our knowledge by
\begin{equation}
\label{eq:WnlowerBound}
\tiosr(W_n) \geq \sqrt{n}
\end{equation}
using the irreducible form of MPS \cite{De17}.
On the other hand, Christandl et al. \cite{Ch20} show that $\btiosr(W_n) = 2$ by defining an approximate decomposition with  $\ket{v_{12}^{\varepsilon}} = \ket{v_{21}^{\varepsilon}} = 0$ and
$$\ket{v_{11}^{\varepsilon}} = \frac{1}{\varepsilon^{1/n}} \begin{pmatrix} 1 \\ \varepsilon \end{pmatrix} \text{ and } \ket{v_{22}^{\varepsilon}} = \frac{1}{\varepsilon^{1/n}} \begin{pmatrix} (-1)^{\frac{1}{n}} \\ 0 \end{pmatrix}$$
for arbitrary $\varepsilon > 0$. This proves a separation between border rank and rank for the t.i.\ cyclic decompositions when $n \geq 5$.

For the t.i.\ psd-decomposition, we obtain $\btipsdosr(W_n) \leq \bsymmpsdrank(W_n) = 2$ by defining the psd matrices
$$\left(B_{j}^{\varepsilon}\right)_{\alpha, \alpha'; \beta, \beta'} = \delta_{\alpha,\alpha'} \cdot \delta_{\beta, \beta'} \cdot \left(A_{j}^{\varepsilon}\right)_{\alpha, \beta}$$
where $A_{j}^{\varepsilon}$ is defined in Equation \eqref{eq:defAmatrices}. We obtain
\begin{align*} & \sum_{\alpha_i, \beta_i = 1}^{2} \left(B_{{j_1}}^{\varepsilon}\right)_{\alpha_1, \alpha_2; \beta_1, \beta_2} \cdots \left(B_{{j_n}}^{\varepsilon}\right)_{\alpha_n, \alpha_1; \beta_n, \beta_1} \\ =& \ \left(W_n\right)_{j_1, \ldots, j_n} + \mathcal{O}(\varepsilon^{1 + \frac{1}{n-1}}).
\end{align*}
Using \Cref{lem:rank-psdrank-bound} together with \Cref{eq:WnlowerBound} we have that 
\begin{equation}
\label{eq:tipsdrankLowerBound}
\tipsdosr(W_n) \geq \Omega(n^{1/4})
\end{equation}
and in particular $\tipsdosr(W_n) \geq 3$ as soon as $n \geq 17$. This proves the separation between border rank and rank for the t.i.\ cyclic psd-decomposition.

For the t.i.\ nonnegative decomposition we construct a tensor with a separation between border rank and rank for every odd $n \geq 5$. Consider again the tensor $W_n$. By the previous discussion we have 
$$ \tinnosr(W_n) \geq \tiosr(W_n) \geq \sqrt{n}.$$
We will now show that 
$$\btinnosr(W_n) \leq 2$$
if $n$ is odd. To this end, we use in the following the representation of a nonnegative cyclic decomposition
$$T_{j_1, \ldots, j_n} = \tr(A_{j_1} \cdots A_{j_n})$$
where $A_j \in \mathcal{M}_{r}(\mathbb{C})$ and $(A_j)_{\alpha,\beta} \geq 0$. This representation coincides with the cyclic decomposition in Equation \eqref{eq:CyclicTiDec} by setting
$$ \left(A_{j}\right)_{\alpha,\beta} = \Cbraket{j}{v_{\alpha,\beta}}.$$
It follows that the rank of the decomposition is specified by the size of the matrices $A_j$. Let
$$A_{0}^{\varepsilon} = \frac{1}{\sqrt[n-1]{\varepsilon}} \left(\begin{array}{cc} 0 & 1\\ 1 & 0\end{array}\right) = \frac{1}{\sqrt[n-1]{\varepsilon}} P_{\tau} \quad \quad A_{1}^{\varepsilon} = \varepsilon I_2$$
be multiples of a nonnegative representation of the cyclic group on $\{1,2\}$, where $\tau$ is the permutation $1 \mapsto 2$ and $2 \mapsto 1$ and $P_{\tau}$ the corresponding permutation matrix. We have
\begin{align*} \left(\widehat{W}^{\varepsilon}_n\right)_{j_1, \ldots, j_n} \coloneqq& \ \frac{1}{2} \tr\left(A_{j_1}^{\varepsilon} \cdots A_{j_n}^{\varepsilon}\right) \\ =& \ \frac{1}{2} \left\{\begin{array}{cl} 0 & \text{: } j_1 + \ldots + j_n \textrm{ even} \\[0.2cm]  \varepsilon^{k - 1 + \frac{k-1}{n-1}} & \text{: } j_1 + \ldots + j_n \textrm{ odd} \end{array}\right. 
\end{align*}
where $k \coloneqq j_1 + \cdots + j_n$. This implies that $\widehat{W}^{\varepsilon}_n = \frac{1}{2} W_n + \mathcal{O}(\varepsilon^2)$. We conclude that for $n \geq 5$ odd, we have 
\begin{align*} \btinnosr(W_n) &= 2 < \sqrt{5} \\ &\leq \tiosr(W_n) \leq \tinnosr(W_n).
\end{align*}
This construction generalizes to every $n$ and $p | (n-1)$ by replacing $\{1,2\}$ with $\{1, \ldots, p\}$, and $\tau$ by the translation on $\{1, \ldots, p\}$. Since the corresponding permutation matrices  $A_{0}^{\varepsilon}$ and $A_{1}^{\varepsilon}$ are of size $p \times p$, it follows that $\btinnosr(W_n) \leq p$.

\subsection{Decomposing multipartite psd matrices}
\label{ssec:borderRanksMatrixDecompositions}

The three types of positive decompositions for nonnegative tensors are related to the three positive decompositions for multipartite psd matrices (see \Cref{ssec:matrixDec}). This enables us to translate gaps between border ranks and ranks for positive tensor decompositions to gaps between border rank and rank for multipartite psd matrices. Given a tensor $T$, such that $\bpsdrank(T) < \psdrank(T)$, the matrix $\rho_T$ (Equation \eqref{eq:rhoT}) satisfies
\begin{align*}
\bpurirank(\rho_T) &\leq  \bpsdrank(T)\\
&< \psdrank(T) \\ &= \purirank(\rho_T).
\end{align*}
and thereby exhibits a gap between border rank and rank for $\purirank$. Analogously, one obtains gaps for matrix tensor decompositions whenever there is a gap in the corresponding tensor decomposition. This strategy results in gaps between border rank and rank for $\symmpurirank$, $\tipuriosr$, $\seposr$, and $\tiseposr$.

\subsection{Non-existence of border rank gaps}
\label{sec:noBorderRanks}

In the following, we provide two cases where no gaps between border rank and rank appear. First, we establish that for standard tensor decompositions (i.e.\ only containing one summation index), the $\nnrank$, $\symmnnrank$, $\seprank$, and the $\symmseprank$ do not exhibit a gap. Second, we prove that $\Omega$-decompositions arising from a tree $\Omega$ do not exhibit gaps between rank and border rank regardless of the positivity constraints.

The proof strategy is similar in both cases. We first show that every decomposition can be reduced to a particular normalized version (i.e.\ every local element satisfies a normalization constraint). Then, we apply the Bolzano--Weierstraß Theorem to the local elements to guarantee that every sequence of decompositions obtained from a converging sequence of global elements converges to a decomposition of the same rank.

\subsubsection{Nonnegative and separable standard tensor decompositions}

Let us now show that $\nnrank$, $\symmnnrank$, $\seprank$, and the $\symmseprank$ do not exhibit a gap between rank and border rank.

\begin{theorem}
\label{thm:nnDecNoBorderRank}
Let $(\rho_k)_{k \in \mathbb{N}}$ be a sequence of $n$-partite separable matrices with limit $\rho_k \to \rho$ and $\seprank(\rho_k) \leq r$ for every $k$. Then,
$$\seprank(\rho) \leq r$$
The same statement holds for $\symmseprank$. It also holds for sequences of nonnegative tensors together with $\nnrank$, and $\symmnnrank$.
\end{theorem}

Since the nonnegative decomposition corresponds to the separable decomposition of a diagonal matrix, it suffices to show the statement for separable decompositions. We will now give a sketch of the proof; for details and the general statement, we refer to \Cref{app:nnDecNoBorderRank}.
Without loss of generality, let $\rho_k$ be a sequence of $n$-partite separable matrices with a separable decomposition
$$ \rho_k = \sum_{\alpha = 1}^{r} \rho^{[1]}_{\alpha, k} \otimes \rho^{[2]}_{\alpha, k} \otimes \cdots \otimes \rho^{[n]}_{\alpha, k}.$$
This sequence can be chosen to be normalized, i.e.\ ${\Vert \rho^{[i]}_{\alpha,k} \Vert \leq C \Vert \rho \Vert}$, so that $(\rho^{[i]}_{\alpha,k})_{k \in \mathbb{N}}$ is a bounded sequence. The Bolzano--Weierstraß Theorem (\Cref{thm:BolzanoWeierstrass} in the appendix) implies the existence of a limiting point of a subsequence leading to a separable decomposition of $\rho$ with $\seprank(\rho) \leq r$.

This generalizes the result in \cite{Qi16}, which shows that the multipartite nonnegative standard tensor decomposition does not exhibit a gap between rank and border rank.

\subsubsection{Tree tensor network decompositions}

Tensor networks without local positivity show border rank phenomena if and only if they contain loops in the hypergraph $\Omega$ that specifies the decomposition structure \cite{Ba21}. In particular, if a hypergraph $\Omega$ is a tree, the corresponding unconstrained tensor network decomposition exhibits no gap between rank and border rank (see also \Cref{app:treeDecompositions} for a brief review of this result). In the following we will prove that the same is the case for positive tensor networks. We show the following:

\TabPositions{0pt, 0.35\linewidth, 0.45\linewidth}
\begin{theorem}
\label{thm:treeBorderRank}
Let $\Omega$ be a tree and $T$ a nonnegative tensor, $\rho_1$ an $n$-partite separable matrix and $\rho_2$ an $n$-partite psd matrix. Then, the following holds:
\begin{enumerate}[label=(\roman*)]
	\item\label{treeI} $\bnnrank_{\Omega}(T)$ \tab $=$ \tab $\nnrank_{\Omega}(T)$
	\item\label{treeII} $\bpsdrank_{\Omega}(T)$ \tab $=$ \tab $\psdrank_{\Omega}(T)$
	\item\label{treeIII} $\bseprank_{\Omega}(\rho_1)$ \tab $=$ \tab $\seprank_{\Omega}(\rho_1)$
	\item\label{treeIV} $\bpurirank_{\Omega}(\rho_2)$ \tab $=$ \tab $\purirank_{\Omega}(\rho_2)$
\end{enumerate}
\end{theorem}

For the proof of \ref{treeI} and \ref{treeIII} we refer to \Cref{ssec:sepTree}, for the proof of \ref{treeII} and \ref{treeIV} we refer to \Cref{ssec:puriTree}.

\begin{figure}[t]
\label{fig:proofStrategySepBounded}
\begin{tikzpicture}[scale=0.7]

\draw [fill=color2, draw=white] plot [smooth cycle, tension=1.2] coordinates {(-3.8,3.5) (-3,5) (-1,6.2) (1.8,6) (2.6,4.4) (0.5,4) (-1, 5) (-2.8, 2.9)};
\draw [fill=color1, draw=white] plot [smooth cycle, tension=1.2] coordinates {(-1.8, 3.5) (-1, 1.8) (-2.6,1.8)};
\draw [fill=color3, draw=white] plot [smooth cycle, tension=1.2] coordinates {(-0.2, 3.5) (0.4, 2.4) (-0.8,2.2)};

\draw[thick] (0,6) -- (0.8,5);
\draw[thick] (0,6) -- (0.8,5);
\draw[thick] (0,6) -- (2,5);
\draw[thick] (0,6) -- (-2,5);
\draw[draw=gray!50!white] (-2,5) -- (-1,4);
\draw[thick] (-2,5) -- (-1.5,4.5);
\draw[thick] (-2,5) -- (-3,4);
\draw[draw=gray!50!white] (-1,4) -- (-1.8,3);
\draw[thick] (-1.8,3) -- (-1.4,3.5);
\draw[draw=gray!50!white] (-1,4) -- (-0.2,3);
\draw[thick] (-0.2,3) -- (-0.6,3.5);
\draw[thick] (-1.8,3) -- (-2.1,2);
\draw[thick] (-1.8,3) -- (-1.5,2);

\draw[thick] (0.8,5) -- (0.5, 4.7);
\draw[thick] (0.8,5) -- (0.9, 4.7);
\draw[thick] (0.8,5) -- (1.1, 4.7);

\draw[thick] (2,5) -- (1.7, 4.7);
\draw[thick] (2,5) -- (2.1, 4.7);
\draw[thick] (2,5) -- (2.3, 4.7);

\draw[thick] (-3,4) -- (-3.3, 3.7);
\draw[thick] (-3,4) -- (-3.1, 3.7);
\draw[thick] (-3,4) -- (-2.7, 3.7);

\draw[thick] (-0.2,3) -- (0.1, 2.7);
\draw[thick] (-0.2,3) -- (-0.3, 2.7);
\draw[thick] (-0.2,3) -- (-0.5, 2.7);

\draw[draw=gray!50!white] (-1,4) -- (-0.8,4.2);

\node at (0.8,4.3) {$\cdots$};
\node at (2,4.3) {$\cdots$};
\node at (-3,3.3) {$\cdots$};
\node at (-0.2,2.3) {$\cdots$};

\node[gray!50!white] at (-0.55,4.15) {\scriptsize $n$};

\draw[fill=color2,draw=none] (0,6) circle (2.5pt);
\draw[fill=color2,draw=none] (0.8,5) circle (2.5pt);
\draw[fill=color2,draw=none] (2,5) circle (2.5pt);
\draw[fill=color2,draw=none] (-2,5) circle (5pt);
\draw[fill=white, draw=none] (-1,4) circle (2.5pt);
\draw[fill=color2,draw=none] (-3,4) circle (2.5pt);
\draw[fill=color1,draw=none] (-1.8,3) circle (5pt);
\draw[fill=color3,draw=none] (-0.2,3) circle (5pt);
\draw[fill=color1,draw=none] (-2.1,2) circle (2.5pt);
\draw[fill=color1,draw=none] (-1.5,2) circle (2.5pt);

\draw[fill=black] (0,6) circle (1.5pt);
\draw[fill=black] (0.8,5) circle (1.5pt);
\draw[fill=black] (2,5) circle (1.5pt);
\draw[fill=black] (-2,5) circle (3pt);
\draw[fill=gray!50!white,draw=gray!50!white] (-1,4) circle (1.5pt);
\draw[fill=black] (-3,4) circle (1.5pt);
\draw[fill=black] (-1.8,3) circle (3pt);
\draw[fill=black] (-0.2,3) circle (3pt);
\draw[fill=black] (-2.1,2) circle (1.5pt);
\draw[fill=black] (-1.5,2) circle (1.5pt);

\draw[-stealth, thick] (0.5,3.5) -- (2.5,2.5) node[midway,right, xshift=0.1cm,yshift=0.2cm] {\scriptsize Induction step}; 

\begin{scope}[xshift=5.1cm,yshift=-3cm]

\draw[thick] (0,6) -- (0.8,5);
\draw[thick] (0,6) -- (0.8,5);
\draw[thick] (0,6) -- (2,5);
\draw[thick] (0,6) -- (-2,5);
\draw[thick] (-2,5) -- (-1,4);
\draw[thick] (-2,5) -- (-3,4);
\draw[thick] (-1,4) -- (-1.8,3);
\draw[thick] (-1,4) -- (-0.2,3);
\draw[thick] (-1.8,3) -- (-2.1,2);
\draw[thick] (-1.8,3) -- (-1.5,2);

\draw[thick, color4, -stealth, xshift=-0.15cm, yshift=0.15cm] (-1.6,3.25) -- (-1.2,3.75);
\draw[thick, color4, -stealth, xshift=0.15cm, yshift=0.15cm] (-1.7,4.7) -- (-1.3,4.3);
\draw[thick, color4, -stealth, xshift=0.15cm, yshift=0.15cm] (-0.4,3.25) -- (-0.8,3.75);

\draw[thick] (0.8,5) -- (0.5, 4.7);
\draw[thick] (0.8,5) -- (0.9, 4.7);
\draw[thick] (0.8,5) -- (1.1, 4.7);

\draw[thick] (2,5) -- (1.7, 4.7);
\draw[thick] (2,5) -- (2.1, 4.7);
\draw[thick] (2,5) -- (2.3, 4.7);

\draw[thick] (-3,4) -- (-3.3, 3.7);
\draw[thick] (-3,4) -- (-3.1, 3.7);
\draw[thick] (-3,4) -- (-2.7, 3.7);

\draw[thick] (-0.2,3) -- (0.1, 2.7);
\draw[thick] (-0.2,3) -- (-0.3, 2.7);
\draw[thick] (-0.2,3) -- (-0.5, 2.7);

\draw[thick] (-1,4) -- (-0.7,4.3);

\node at (0.8,4.3) {$\cdots$};
\node at (2,4.3) {$\cdots$};
\node at (-3,3.3) {$\cdots$};
\node at (-0.2,2.3) {$\cdots$};

\node at (-0.4,4.2) {\scriptsize $n$};

\draw[fill=white,draw=none] (0,6) circle (2.5pt);
\draw[fill=white,draw=none] (0.8,5) circle (2.5pt);
\draw[fill=white,draw=none] (2,5) circle (2.5pt);
\draw[fill=white,draw=none] (-2,5) circle (2.5pt);
\draw[fill=white, draw=none] (-1,4) circle (5pt);
\draw[fill=white,draw=none] (-3,4) circle (2.5pt);
\draw[fill=white,draw=none] (-1.8,3) circle (2.5pt);
\draw[fill=white,draw=none] (-0.2,3) circle (2.5pt);
\draw[fill=white,draw=none] (-2.1,2) circle (2.5pt);
\draw[fill=white,draw=none] (-1.5,2) circle (2.5pt);

\draw[fill=black] (0,6) circle (1.5pt);
\draw[fill=black] (0.8,5) circle (1.5pt);
\draw[fill=black] (2,5) circle (1.5pt);
\draw[fill=black] (-2,5) circle (1.5pt);
\draw[fill=black] (-1,4) circle (3pt);
\draw[fill=black] (-3,4) circle (1.5pt);
\draw[fill=black] (-1.8,3) circle (1.5pt);
\draw[fill=black] (-0.2,3) circle (1.5pt);
\draw[fill=black] (-2.1,2) circle (1.5pt);
\draw[fill=black] (-1.5,2) circle (1.5pt);
\end{scope}

\end{tikzpicture}
\caption{\textbf{Proof idea of \Cref{lem:boundSepDecTree}}. Assuming that the separable matrix $\rho$ attains a normalized decomposition on subtrees of size at most $n-1$ (i.e.\ every local family of matrices is normalized except one), we obtain a normalized decomposition of size $n$ by shifting the weight to the added vertex.}
\label{fig:proofSketchTrees}
\end{figure}

Again, the proof contains the same two steps: First, every element attains a decomposition with bounded local elements; second, every bounded sequence contains a limiting sequence which gives rise to an optimal decomposition. In the following, we will give a brief sketch of the proof for the separable decomposition. The existence of a normalized decomposition follows from the following lemma.

\begin{lemma}
\label{lem:boundSepDecTree}
 Let $\Omega$ be a tree and $\rho$ an $n$-partite separable matrix with $\seprank_{\Omega}(\rho) \leq r$ and bounded trace-norm $\Vert \rho \Vert_{1} \leq 1$.
There exists a separable decomposition of rank $r$
such that the trace-norm of each local matrix is bounded by $1$.
\end{lemma}

We refer to \Cref{ssec:sepTree} for the proof of \Cref{lem:boundSepDecTree}.
Following the same reasoning as for the separable standard tensor decomposition, \Cref{lem:boundSepDecTree} together with the Bolzano--Weierstraß Theorem implies that for every sequence $\rho_k \to \rho$ with $\seprank(\rho_k) \leq r$ for every $k$ it follows that $\seprank(\rho) \leq r$ which proves \ref{treeIII} in \Cref{thm:treeBorderRank}.

Note that \Cref{lem:boundSepDecTree} implies that every separable density matrix attains an optimal decomposition
$$\rho = \sum_{\alpha_1, \ldots, \alpha_n = 1}^{r} \rho_{\alpha_1}^{[1]} \otimes \rho_{\alpha_1, \alpha_2}^{[2]} \otimes \cdots \otimes \rho_{\alpha_{n-2}, \alpha_{n-1}}^{[n-1]} \otimes \rho_{\alpha_{n-1}}^{[n]}$$
with $\Vert \rho_{\alpha,\beta}^{[i]} \Vert_1 \leq 1$, since the summation indices are arranged according to a line. In contrast, a cyclic decomposition is in general not normalizable since this violates the gap between rank and border rank for separable decompositions shown in \Cref{ssec:cyclic} and \Cref{ssec:borderRanksMatrixDecompositions}.

We now present the proof idea of \Cref{lem:boundSepDecTree}.  Every tree can be built up from smaller ones, and, by induction over such decompositions, the weight of local elements can be transferred to neighboring vertices in the larger tree (see \Cref{fig:proofSketchTrees}).  For a single system, the statement is trivial. For $n-1 \to n$ consider the graph which arises when removing the vertex $n$. The tree factorizes into smaller trees of size at most $n-1$. By the induction hypothesis, we can assume that every subtree leads to a family of normalized decompositions where each local vector is normalized by $1$, except for the local vectors in the systems adjacent to $n$ which are normalized by $\Vert \rho \Vert_{1}$. We finally show that it attains a decomposition with local matrices normalized by $1$ for systems $1, \ldots, n-1$ and normalized by $\Vert \rho \Vert_{1} \leq 1$ for system $n$ by adding system $n$ and transferring the weights to system $n$.

\section{Consequences}
\label{sec:applications}

Let us now present two implications of the existence of gaps between rank and border rank in multipartite positive tensor decompositions. 
In \Cref{ssec:rankCorrelation} we show a correspondence between positive tensor decompositions and quantum correlation sets. The gaps between border ranks and ranks then imply that certain sets of quantum correlations are not closed.
In \Cref{ssec:borderSeparation} we prove that gaps also lead to new types of separations between positive tensor ranks.

\begin{figure}[t]
\begin{tikzpicture}[scale=0.9]

\begin{scope}


\node[anchor=west] at (-1.2,-2.2) {\small (a) \textbf{Multipartite classical}};

\draw[fill=gray!20!white] (18:0.5cm) -- (90:0.5cm) -- (18+2*72:0.5cm) --
(18+3*72:0.5cm) -- (18+4*72:0.5cm) -- cycle;
\filldraw (18:0.5cm) circle (1.2pt);
\filldraw (18+72:0.5cm) circle (1.2pt);
\filldraw (18+2*72:0.5cm) circle (1.2pt);
\filldraw (18+3*72:0.5cm) circle (1.2pt);
\filldraw (18+4*72:0.5cm) circle (1.2pt);

\node at (0,-1.2) {$\nnrank(P) \leq r$};

\node at (1.15,0) {$\Longleftrightarrow$};


\begin{scope}[xshift=1.8cm,yshift=-0.2cm]

\draw[-stealth] (2.4,0.9) to[out=180,in=90] (0.4,0.1);
\draw[-stealth] (2.4,0.9) to[out=-155,in=90] (1.4,0.1);
\draw[-stealth] (2.4,0.9) to[out=-90,in=90] (2.4,0.1);
\draw[-stealth] (2.4,0.9) to[out=-25,in=90] (3.4,0.1);
\draw[-stealth] (2.4,0.9) to[out=0,in=90] (4.4,0.1);

\draw[fill=white, draw=none] (2.4,0.9) circle (0.4cm);
\draw[fill=white] (2.4,0.9) circle (0.28cm);
\node at (2.4,0.9) {$\Lambda$};

\node at (3.9,1.3) {with $\Lambda \in [r]$};

\draw (0,-0.5) rectangle (0.8,0) node[midway] {\scriptsize $X_1 | \Lambda$};
\draw (1,-0.5) rectangle (1.8,0) node[midway] {\scriptsize $X_2 | \Lambda$};
\draw (2,-0.5) rectangle (2.8,0) node[midway] {\scriptsize $X_3 | \Lambda$};
\draw (3,-0.5) rectangle (3.8,0) node[midway] {\scriptsize $X_4 | \Lambda$};
\draw (4,-0.5) rectangle (4.8,0) node[midway] {\scriptsize $X_5 | \Lambda$};

\draw[-stealth] (0.4,-0.5) -- (0.4,-0.9);
\draw[-stealth] (1.4,-0.5) -- (1.4,-0.9);
\draw[-stealth] (2.4,-0.5) -- (2.4,-0.9);
\draw[-stealth] (3.4,-0.5) -- (3.4,-0.9);
\draw[-stealth] (4.4,-0.5) -- (4.4,-0.9);

\draw[decorate, decoration={brace}] (4.5,-1) -- (0.3,-1) node[midway, pos=0.25, below, yshift=-0.1cm, color=black] {\footnotesize $P \in \CCorr(n,d,r)$};

\end{scope}
\end{scope}

\begin{scope}[yshift=-5cm]


\node[anchor=west] at (-1.2,-2.4) {\small (b) \textbf{Multipartite quantum-classical}};

\draw (0,0.4) -- (-0.4,0);
\draw (-0.4,0) -- (-0.6,-0.4);
\draw (0,0.4) -- (0.4,0);
\draw (-0.4,0) -- (-0.2,-0.4);

\filldraw (0,0.4) circle (1.2pt);
\filldraw (-0.4,0) circle (1.2pt);
\filldraw (0.4,0) circle (1.2pt);
\filldraw (-0.6,-0.4) circle (1.2pt);
\filldraw (-0.2,-0.4) circle (1.2pt);

\node at (0,-1.2) {$\psdrank_{\Omega}(P) \leq r$};

\node at (1.15,0) {$\Longleftrightarrow$};


\begin{scope}[xshift=2.2cm, yshift=0.2cm, scale=0.38]

\draw[fill=color1, draw=none] (-0.7,-0.7) rectangle (9.7,4);

\node[anchor=west] at (0.1,2.9) {with $\rank_{\Omega}(\ket{\psi}) \leq r$};

\draw[fill=white] (-0.25,-0.2) rectangle (9.25,2);

\node at (3.25,0.9) {$\ket{\psi} \ \sim$};

\begin{scope}[scale=1.6, xshift=3.7cm, yshift=0.55cm]
\draw (0,0.4) -- (-0.4,0);
\draw (-0.4,0) -- (-0.6,-0.4);
\draw (0,0.4) -- (0.4,0);
\draw (-0.4,0) -- (-0.2,-0.4);

\filldraw (0,0.4) circle (1.2pt);
\filldraw (-0.4,0) circle (1.2pt);
\filldraw (0.4,0) circle (1.2pt);
\filldraw (-0.6,-0.4) circle (1.2pt);
\filldraw (-0.2,-0.4) circle (1.2pt);
\end{scope}

\draw[-stealth] (0.5,-0.2) -- (0.5,-1.5);
\draw[-stealth] (2.5,-0.2) -- (2.5,-1.5);
\draw[-stealth] (4.5,-0.2) -- (4.5,-1.5);
\draw[-stealth] (6.5,-0.2) -- (6.5,-1.5);
\draw[-stealth] (8.5,-0.2) -- (8.5,-1.5);

\draw[-stealth] (0.5,-3) -- (0.5,-4);
\draw[-stealth] (2.5,-3) -- (2.5,-4);
\draw[-stealth] (4.5,-3) -- (4.5,-4);
\draw[-stealth] (6.5,-3) -- (6.5,-4);
\draw[-stealth] (8.5,-3) -- (8.5,-4);

\draw[fill=white] (-0.25,-3) rectangle (1.25,-1.5);
\draw[fill=white] (1.75,-3) rectangle (3.25,-1.5);
\draw[fill=white] (3.75,-3) rectangle (5.25,-1.5);
\draw[fill=white] (5.75,-3) rectangle (7.25,-1.5);
\draw[fill=white] (7.75,-3) rectangle (9.25,-1.5);

\begin{scope}[shift={(0.5,-2.9)}]
\draw[thick,domain=45:135, xshift=0] plot ({0.7*cos(\x)}, {0.7*sin(\x)});
\draw[line width=0.1cm, white, rotate around={-15:(0,0.1)}] (0,0.2) -- (0,1.1);
\draw[-stealth, rotate around={-15:(0,0.1)}] (0,0.2) -- (0,1.1);
\end{scope}

\begin{scope}[shift={(2.5,-2.9)}]
\draw[thick,domain=45:135, xshift=0] plot ({0.7*cos(\x)}, {0.7*sin(\x)});
\draw[line width=0.1cm, white, rotate around={-15:(0,0.1)}] (0,0.2) -- (0,1.1);
\draw[-stealth, rotate around={-15:(0,0.1)}] (0,0.2) -- (0,1.1);
\end{scope}

\begin{scope}[shift={(4.5,-2.9)}]
\draw[thick,domain=45:135, xshift=0] plot ({0.7*cos(\x)}, {0.7*sin(\x)});
\draw[line width=0.1cm, white, rotate around={-15:(0,0.1)}] (0,0.2) -- (0,1.1);
\draw[-stealth, rotate around={-15:(0,0.1)}] (0,0.2) -- (0,1.1);
\end{scope}

\begin{scope}[shift={(6.5,-2.9)}]
\draw[thick,domain=45:135, xshift=0] plot ({0.7*cos(\x)}, {0.7*sin(\x)});
\draw[line width=0.1cm, white, rotate around={-15:(0,0.1)}] (0,0.2) -- (0,1.1);
\draw[-stealth, rotate around={-15:(0,0.1)}] (0,0.2) -- (0,1.1);
\end{scope}

\begin{scope}[shift={(8.5,-2.9)}]
\draw[thick,domain=45:135, xshift=0] plot ({0.7*cos(\x)}, {0.7*sin(\x)});
\draw[line width=0.1cm, white, rotate around={-15:(0,0.1)}] (0,0.2) -- (0,1.1);
\draw[-stealth, rotate around={-15:(0,0.1)}] (0,0.2) -- (0,1.1);
\end{scope}

\draw[decorate, decoration={brace}] (8.9,-4.1) -- (0.1,-4.1) node[midway, pos=0.12, below, yshift=-0.1cm, color=black] {\footnotesize $P \in \CQCorr_{\Omega}(n,d,r)$};

\end{scope}
\end{scope}

\begin{scope}[yshift=-10cm]


\node[anchor=west] at (-1.2,-2.4) {\small (c) \textbf{Multipartite quantum-quantum}};

\draw (0,0.4) -- (-0.4,0);
\draw (-0.4,0) -- (-0.6,-0.4);
\draw (0,0.4) -- (0.4,0);
\draw (-0.4,0) -- (-0.2,-0.4);

\filldraw (0,0.4) circle (1.2pt);
\filldraw (-0.4,0) circle (1.2pt);
\filldraw (0.4,0) circle (1.2pt);
\filldraw (-0.6,-0.4) circle (1.2pt);
\filldraw (-0.2,-0.4) circle (1.2pt);

\node at (0,-1.2) {$\purirank_{\Omega}(\rho) \leq r$};

\node at (1.15,0) {$\Longleftrightarrow$};


\begin{scope}[xshift=2.2cm, yshift=0.2cm, scale=0.38]

\draw[fill=color1, draw=none] (-0.7,-0.7) rectangle (9.7,4);

\node[anchor=west] at (0.1,2.9) {with $\rank_{\Omega}(\ket{\psi}) \leq r$};

\draw[fill=white] (-0.25,-0.2) rectangle (9.25,2);

\node at (3.25,0.9) {$\ket{\psi} \ \sim$};

\begin{scope}[scale=1.6, xshift=3.7cm, yshift=0.55cm]
\draw (0,0.4) -- (-0.4,0);
\draw (-0.4,0) -- (-0.6,-0.4);
\draw (0,0.4) -- (0.4,0);
\draw (-0.4,0) -- (-0.2,-0.4);

\filldraw (0,0.4) circle (1.2pt);
\filldraw (-0.4,0) circle (1.2pt);
\filldraw (0.4,0) circle (1.2pt);
\filldraw (-0.6,-0.4) circle (1.2pt);
\filldraw (-0.2,-0.4) circle (1.2pt);
\end{scope}

\draw[-stealth] (0.5,-0.2) -- (0.5,-1.5);
\draw[-stealth] (2.5,-0.2) -- (2.5,-1.5);
\draw[-stealth] (4.5,-0.2) -- (4.5,-1.5);
\draw[-stealth] (6.5,-0.2) -- (6.5,-1.5);
\draw[-stealth] (8.5,-0.2) -- (8.5,-1.5);

\draw[-stealth] (0.5,-3) -- (0.5,-4);
\draw[-stealth] (2.5,-3) -- (2.5,-4);
\draw[-stealth] (4.5,-3) -- (4.5,-4);
\draw[-stealth] (6.5,-3) -- (6.5,-4);
\draw[-stealth] (8.5,-3) -- (8.5,-4);

\draw[fill=white] (-0.25,-3) rectangle (1.25,-1.5) node[midway] {$\mathcal{E}_1$};
\draw[fill=white] (1.75,-3) rectangle (3.25,-1.5) node[midway] {$\mathcal{E}_2$};
\draw[fill=white] (3.75,-3) rectangle (5.25,-1.5) node[midway] {$\mathcal{E}_3$};
\draw[fill=white] (5.75,-3) rectangle (7.25,-1.5) node[midway] {$\mathcal{E}_4$};
\draw[fill=white] (7.75,-3) rectangle (9.25,-1.5) node[midway] {$\mathcal{E}_5$};

\draw[decorate, decoration={brace}] (8.9,-4.1) -- (0.1,-4.1) node[midway,pos=0.12, below, yshift=-0.1cm, color=black] {\footnotesize $\rho \in \QQCorr_{\Omega}(n,d,r)$};

\end{scope}
\end{scope}

\end{tikzpicture}
\caption{\textbf{Positive ranks characterize correlation scenarios.} (a) For a nonnegative tensor $P$ representing an $n$-partite probability distribution, the standard nonnegative rank characterizes the range of a shared hidden variable to generate $P$ via local conditional probability distributions. (b) $\psdrank_{\Omega}(P)$ characterizes the minimal $\rank_{\Omega}$ of a resource state $\ket{\psi}$ to generate $P$ via local measurements. (c) For an $n$-partite density matrix $\rho$, $\purirank_{\Omega}(\rho)$ characterizes the minimal $\rank_{\Omega}$ of a resource state $\ket{\psi}$ to generate $\rho$ via local quantum channels.}
\label{fig:correlationScenarios}
\end{figure}

\subsection{Sets of quantum correlations are not closed}
\label{ssec:rankCorrelation}

Positive tensor decompositions give rise to statements about (quantum) correlations. Consider a bipartite finite probability distribution represented as an entrywise nonnegative tensor $P \in \mathbb{C}^d \otimes \mathbb{C}^d$ via $$P_{ij} = P(X = i, Y =j)$$ where $X$ and $Y$ are random values taking values in $\{1,\ldots, d\}$. It is known that the non-negative rank of $P$ corresponds to the minimal range of a hidden variable $\Lambda$ necessary to sample this distribution locally. Similarly, the $\psdrank(P)$ is the minimal Schmidt rank of the bipartite state $\ket{\psi}$ necessary to generate $P$ via local measurements on each qudit \cite{Fa14, Ja13}.

We now show that these correspondences generalize to multipartite tensor decompositions as well as multipartite matrix tensor decompositions. Together with the existence of non-trivial border ranks, this entails that these sets of probability distributions as well as density matrices are topologically not closed.

\subsubsection{Entanglement in correlation scenarios and positive ranks}

We define the set $\CCorr(n,d,r)$ as the set of probability distributions arising from local distributions conditioned on a shared hidden variable taking values in $\{1,\ldots, r\}$ (see \Cref{fig:correlationScenarios} (a)), i.e.\
\begin{align*}
& P(X_1 = j_1, \ldots, X_n = j_n)\\ = &\sum_{\alpha = 1}^{r} P(\Lambda = \alpha) \prod_{i=1}^{n} P(X_i = j_i \, | \, \Lambda = \alpha)
\end{align*}
where $X_1, \ldots, X_n$ are random variables taking values in $\{1,\ldots, d\}$.

Similarly, we define the set $\CQCorr_{(\Omega,G)}(n,d,r)$ for a given hypergraph $\Omega$ and a group action $G$ on $\Omega$ as the set of all probability distributions $P$ arising as
$$P(X_1 = j_1, \ldots, X_n = j_n) = \bra{\psi} A_{j_1}^{[1]} \otimes \cdots \otimes A_{j_n}^{[n]} \ket{\psi}$$
where $\left(A^{[i]}_j\right)_{j=1}^{d}$ for $i \in \{1,\ldots, n\}$ are POVM measurements that are $G$-symmetric, i.e.\ the measurement on position $i$ coincides with the measurement on $gi$ for every $g \in G$. In other words, we have that $A_j^{[gi]} = A_j^{[i]}$ for every $g \in G$. Moreover, $\ket{\psi}$ satisfies the constraint that $\rank_{(\Omega,G)}(\ket{\psi}) \leq r$ (see \Cref{fig:correlationScenarios} (b)).

So, for example, when $\Omega = \Theta_n$ is a cycle graph of length $n$ (i.e.\ vertex $i$ is only connected to $i-1$ and $i+1$ with addition modulo $n+1$) then $\CQCorr_{\Theta_n}(n,d,r)$ is the set of all $n$-partite probability distributions obtained from an MPS $\ket{\psi}$ with $\osr(\ket{\psi}) \leq r$ via local measurements on each local space. For the cyclic group $G = C_n$, $\CQCorr_{(\Theta_n,C_n)}(n,d,r)$ is the set of probability distributions obtained from an MPS $\ket{\psi}$ with $\tiosr(\ket{\psi}) \leq r$ via identical local measurements on each local space.

These correlation scenarios also generalize to quantum states instead of probability distributions by replacing local measurement operators by local quantum channels (see \Cref{fig:correlationScenarios} (c)). We define the set $\QQCorr_{(\Omega,G)}(n,d,r)$ as the set of all density matrices arising as
$$ \rho = (\mathcal{E}_1 \otimes \cdots \otimes \mathcal{E}_n)\left(\ket{\psi} \bra{\psi}\right)$$
where $(\mathcal{E}_i)_{i=1}^{n}$ is a $G$-invariant family of quantum channels, i.e.\ $\mathcal{E}_{i} = \mathcal{E}_{gi}$. Moreover, $\ket{\psi}$ satisfies the restriction that $\rank_{(\Omega,G)}(\ket{\psi}) \leq r$.
 
Analogous to the above example, if $\Omega = \Theta_n$ is a cycle graph of length $n$, then $\QQCorr_{\Theta_n}(n,d,r)$ is the set of all $n$-partite density matrices obtained from an MPS $\ket{\psi}$ with $\osr(\ket{\psi}) \leq r$ and applying local quantum channels on each local space. The set $\QQCorr_{(\Theta_n,C_n)}(n,d,r)$ are the density matrices obtained from an MPS $\ket{\psi}$ with $\tiosr(\ket{\psi}) \leq r$ and applying identical quantum channels on each local space. 
 
Similarly to the bipartite case \cite{Ja17,Li09}, there is a correspondence between the local purification form and multipartite quantum scenarios.

\begin{theorem} \label{thm:quantumCorrMainText}
Let $P$ be a tensor representing an $n$-partite probability distribution with local dimensions $d$, and $\rho$ an $n$-partite density matrix with local dimensions $d$. Then,
\begin{enumerate}[label=(\roman*)]
	\item $P \in \CCorr(n,d,r) \quad \Leftrightarrow \quad \nnrank(P) \leq r$
	\item $P \in \CQCorr_{(\Omega,G)}(n,d,r)$\\ \hspace*{1.5cm} $\Leftrightarrow$\\ $\psdrank_{(\Omega,G)}(P) \leq r$
	\item $\rho \in \QQCorr_{(\Omega,G)}(n,d,r)$\\ \hspace*{1.5cm} $\Leftrightarrow$ \\ $\purirank_{(\Omega,G)}(\rho) \leq r$
\end{enumerate}
\end{theorem}
For a proof of the statements in this theorem, we refer to \Cref{thm:nndec_CausalStruct}, \Cref{thm:CQCorr} and \Cref{thm:QQCorr}, respectively in \Cref{app:NonnegativeDec_CausalStructures} and \Cref{app:CorrelationCorrespondenceAppendix}.

\subsubsection{Non-closedness of quantum correlation scenarios}

We now show that the correspondence in \Cref{thm:quantumCorrMainText} together with the gaps between ranks and border ranks imply that the sets of correlations are not closed. It follows that it is generally impossible to test membership of a probability distribution in these sets with a finite number of measurements.

To prove non-closedness, let $(P_k)_{k \in \mathbb{N}}$ be a sequence of tensors representing a probability distribution with $\lim_{k \to \infty} P_k = P$ and exhibiting a gap between rank and border rank, i.e.\
$$ \psdrank_{(\Omega,G)}(P_k) \leq r < \psdrank_{(\Omega,G)}(P).$$
Then $P_k \in \CQCorr_{(\Omega,G)}(n,d,r)$ for all $k \in \mathbb{N}$ while $P \notin \CQCorr_{(\Omega,G)}(n,d,r)$, i.e.\ $\CQCorr_{(\Omega,G)}(n,d,r)$ is not closed.

This implies that there does not always exist a witness detecting the rank of a resource state necessary to generate a probability distribution. We call a continuous function ${f: \left(\mathbb{R}^d\right)^{\otimes n} \to \mathbb{R}}$ a witness detecting $P \notin \CQCorr_{(\Omega,G)}(n,d,r)$ if and only if $f(P) < 0$ and $f(Q) \geq 0$ for all $Q \in \CQCorr_{(\Omega,G)}(n,d,r)$. Assuming such a function exists, a sequence exhibiting a gap between border rank and rank would violate the property that $f$ is continuous. So, it is not possible to detect the rank necessary to generate $P$ from finitely many samples generating $P$, since every $\varepsilon$-ball around $P$ intersects with $\CQCorr_{(\Omega,G)}(n,d,r)$.

According to the gaps between ranks and border ranks (see \Cref{fig:borderRankResults}) the same behavior appears in the following cases:
\begin{enumerate}[label=(\roman*)]
\itemsep0.1em
	\item Testing the standard tensor rank for $n \geq 5$.
	\item Symmetrically testing the symmetric tensor rank for $n \geq 3$.
	\item Symmetrically testing the translational invariant operator Schmidt rank for $n \geq 17$.
\end{enumerate}
Analogously, one can show that $\QQCorr_{(\Omega,G)}(n,d,r)$ is not closed in the above situations.

In contrast, the set of classical correlations $\CCorr(n,d,r)$ is closed for every choice of $n,d,r \in \mathbb{N}$. This follows from the fact that $\nnrank$ does not exhibit a gap between border rank and rank, and hence for every converging sequence of nonnegative tensors $P_k \to P$ with $\nnrank(P_k) \leq r$ we also have $\nnrank(P) \leq r$. For every $P \notin \CCorr(n,d,r)$ there also exists a separating witness since the distance between $\CCorr(n,d,r)$ and $P$ is strictly positive. Moreover, the sets of quantum correlations $\CQCorr_{\Omega}(n,d,r)$ and $\QQCorr_{\Omega}(n,d,r)$ are closed if $\Omega$ is a tree.

\subsection{Stability of separations for approximate tensor decompositions}
\label{ssec:borderSeparation} 

Various notions of positive tensor decompositions exhibit separations \cite{Fa14}, meaning that there exist families of bipartite tensors $(T_d)_{d \in \mathbb{N}}$ where $T_d \in \mathbb{C}^d \otimes \mathbb{C}^d$ such that
$$ \rank(T_d) = \textrm{const.} \quad \text{and} \quad \psdrank(T_d) \to \infty.$$
Moreover, there is also a family of bipartite tensors $(S_d)_{d \in \mathbb{N}}$ such that
$$ \psdrank(S_d) = \textrm{const.} \quad \text{and} \quad \nnrank(S_d) \to \infty.$$
Are these separations robust with respect to approximations? In \cite{De20} it is proven that for fixed approximation error $\varepsilon > 0$ and a fixed norm, the separations between $\rank$, $\psdrank$ and $\nnrank$ disappear. More, precisely $\rank^{\varepsilon}(T), \psdrank^{\varepsilon}(T), \nnrank^{\varepsilon}(T)$ can be upper bounded by a function depending only on $\varepsilon$ and $\Vert T \Vert$, independent of the dimension of the tensor product space. However, if the choice of $\varepsilon > 0$ and vector space dimension is too small, this upper bound exceeds trivial dimension-dependent upper bounds. So, the bounds are only meaningful when the dimension of the tensor product space is large.

We will now prove a ``dual'' statement. If the dimension of the tensor product space is fixed, there exists an error $\varepsilon > 0$ such that the separation between $\rank$ and $\nnrank$ persists.
\begin{figure}
\centering
\label{fig:disappearenceOfSeparations}
\begin{tikzpicture}[scale=1]
  
\draw[fill=color1, domain=0.43:5, variable=\x, draw=white] (0.43,0) -- (0.43,2.52) -- plot (\x, {1/\x+0.2}) -- (5,0.4) --  (5,0) -- cycle;

\draw[thick] (-0.2, 0) -- (5.5, 0) node[below, midway,xshift=3cm] {\scriptsize System size};
  \draw[thick] (0, -0.2) -- (0, 3.5) node[above left, xshift=0.5cm] {\scriptsize Error $\varepsilon$};
  \draw[fill = color3, opacity=0.2, domain=0.43:5, variable=\x, draw=white] (0.43, 3.4) -- plot (\x,{1.8/sqrt(\x)+0.2}) -- (5,0.3) -- (5,3.4) -- cycle;
  
  \draw[thick, -stealth, opacity=0.2] (0.22,2.18) -- (2.02,2.18);
  \draw[thick, -stealth] (0.2,2.2) -- (2,2.2) node [right] {\scriptsize Disapperance};
  \node at (3,1.8) (A) {\scriptsize of separations \cite{De20}};
  \draw[thick, -stealth, opacity=0.2] (1.52,1.48) -- (1.52,0.39);
  \draw[thick, -stealth] (1.5,1.5) -- (1.5,0.42) node [right, yshift=-0.25cm, xshift=-1cm] {\scriptsize Separations remain (\Cref{thm:Separations})};
  
\end{tikzpicture}
\caption{\textbf{When do separations persist for approximate decompositions?} Varying the system size for a fixed approximation error leads to the disappearance of separations \cite{De20}, however, when fixing the system size and choosing small enough errors, the separations persist (\Cref{thm:Separations}). The red area shows when the upper bounds derived in \cite{De20} are smaller than the upper bounds of the exact decompositions, and the blue area shows the approximation errors $\varepsilon_n$ in \Cref{thm:Separations}.}
\end{figure}

\begin{theorem}
\label{thm:Separations}
There exists a family of nonnegative tensors $(T_n)_{n \in \mathbb{N}}$ with $T_n \in \left(\mathbb{C}^d\right)^{\otimes n}$ and a family of approximation errors $\varepsilon_n > 0$ such that
$$ \nnrank^{\varepsilon_n}(T_n) = n.$$
We have also that
$$ \rank^{\varepsilon}(T_n) = \psdrank^{\varepsilon}(T_n) = 2$$
for every $\varepsilon > 0$ independent of $n$.
\end{theorem} 

Note that a similar statement can be shown for every notion of rank that \emph{does not} exhibit a gap between border rank and rank.

\begin{proof}
Let $T_n = W_n$ the family of $W_n$-states. For fixed $n \in \mathbb{N}$, we know that
$$ \bnnrank(W_n) = \nnrank(W_n) = n.$$
Therefore there exists a $\varepsilon_n > 0$ such that
$$ \nnrank^{\varepsilon_n}(W_n) = n.$$
For the second statement, recall that $$\brank(W_n) = \bpsdrank(W_n) = 2.$$ Since $$\rank^{\varepsilon}(W_n) \leq \brank(W_n) = 2$$
and
$$\psdrank^{\varepsilon}(W_n) \leq \bpsdrank(W_n) = 2$$
for every $\varepsilon > 0$, this proves the statement.
\end{proof}

\renewcommand{\arraystretch}{1.6}

\section{Conclusions and Outlook}
\label{sec:Conclusion}

In this work, we have shown that many gaps between ranks and border ranks persist when introducing positivity and invariance constraints for tensor decompositions, and explored its consequences. 
More precisely, we have proven that: 
\begin{enumerate}[label=(\roman*)]
\item 
The standard and symmetric tensor decompositions exhibit gaps between border rank and rank for the psd-decomposition and local purifications (\Cref{ssec:standardBorder}), and the gaps disappear for the nonnegative and separable decomposition (\Cref{thm:nnDecNoBorderRank});  
\item 
Most of the gaps persist for cyclic and translational invariant decompositions (\Cref{ssec:cyclic} and \Cref{ssec:tiBorder}); 
\item 
There are no gaps for tree tensor decompositions, regardless of positivity constraints  (\Cref{thm:treeBorderRank});  
\item 
Upper bounds in the various positive ranks correspond to membership in certain (quantum) correlation sets (\Cref{thm:quantumCorrMainText}), which, together with the gaps between border rank and rank, imply that certain correlation sets are not topologically closed.
\end{enumerate}

Many examples exhibiting a separation are $n$-partite tensor decompositions with $n > 3$. This leaves open the question whether gaps between border ranks and ranks exist for positive and invariant $3$-partite decompositions. Specifically, do
\begin{itemize}
	\item $\psdrank$
	\item $\tipsdosr$
	\item $\tinnosr$
\end{itemize}
exhibit a gap between border rank and rank for $n = 3$?
Moreover, it is open, whether there is a gap for $\psdosr$ for any $n$. One further step would be to ask how big the gaps between border rank and rank can be for positive decompositions. We refer to \cite{Zu17} for a study in the case of standard tensor decompositions without positivity constraints.

Other surprising properties of tensor decompositions appearing already at $n=3$ include the fact that 
tensor rank and border rank are non-additive with respect to the direct sum \cite{Sch81, Sh17c, Ch21}, 
and that they are also non-multiplicative with respect to tensor products \cite{Ch18b, Ch19d}. 
Do these properties also hold for positive and invariant decompositions? 
And what would be their implications for correlation scenarios? 

Finally, there are  positivity structures such as multipartite nonnegative and sum-of-squares polynomials that behave very similarly to multipartite positive tensors \cite{De21b}. It would be interesting if positive decompositions of polynomials exhibit gaps between border ranks and ranks.

\section{Acknowledgments}
This research was funded in part by the Austrian Science Fund (FWF) [doi:\href{https://www.doi.org/10.55776/P33122}{10.55776/P33122}]. For open access purposes, the authors have applied a CC BY public copyright license to any author accepted manuscript version arising from this submission. AK further acknowledges funding of the Austrian Academy of Sciences (\"OAW) through the DOC scholarship 26547.

\bibliographystyle{quantum}
\bibliography{Bibliography.bib}

\begin{thebibliography}{10}

\bibitem{Bi80}
D.~Bini, G.~Lotti, and F.~Romani.
\newblock ``Approximate solutions for the bilinear form computational
  problem''.
\newblock \href{https://dx.doi.org/10.1137/0209053}{SIAM J. Comput. {\bf 9},
  692--697}~(1980).

\bibitem{La11}
J.~M. Landsberg.
\newblock ``Tensors: Geometry and applications''.
\newblock \href{https://dx.doi.org/10.1090/gsm/128}{Volume 128}.
\newblock American Mathematical Soc. ~(2011).

\bibitem{Ba21}
T.~Barthel, J.~Lu, and G.~Friesecke.
\newblock ``On the closedness and geometry of tensor network state sets''.
\newblock \href{https://dx.doi.org/10.1007/s11005-022-01552-z}{Lett. Math.
  Phys. {\bf 112}, 72}~(2022).

\bibitem{Ch20}
M.~Christandl, A.~Lucia, P.~Vrana, and A.~H. Werner.
\newblock ``Tensor network representations from the geometry of entangled
  states''.
\newblock \href{https://dx.doi.org/10.21468/SCIPOSTPHYS.9.3.042}{SciPost Phys.
  {\bf 9}, 1--35}~(2020).

\bibitem{La12}
J.~M. Landsberg, Y.~Qi, and K.~Ye.
\newblock ``On the geometry of tensor network states''.
\newblock \href{https://dx.doi.org/10.26421/qic12.3-4-12}{Quantum Inf. Comput.
  {\bf 12}, 346--354}~(2012).

\bibitem{Si08}
V.~De Silva and L.~H. Lim.
\newblock ``Tensor rank and the ill-posedness of the best low-rank
  approximation problem''.
\newblock \href{https://dx.doi.org/10.1137/06066518X}{SIAM J. Matrix Anal.
  Appl. {\bf 30}, 1084--1127}~(2008).

\bibitem{La17}
J.~M. Landsberg and Mateusz Michałek.
\newblock ``Abelian tensors''.
\newblock \href{https://dx.doi.org/10.1016/j.matpur.2016.11.004}{Journal de
  Mathématiques Pures et Appliquées {\bf 108}, 333--371}~(2017).

\bibitem{Zu17}
J.~Zuiddam.
\newblock ``A note on the gap between rank and border rank''.
\newblock \href{https://dx.doi.org/10.1016/j.laa.2017.03.015}{Linear Algebra
  Appl. {\bf 525}, 33--44}~(2017).

\bibitem{Ch21c}
M.~Christandl, F.~Gesmundo, D.~Stilck França, and A.~H. Werner.
\newblock ``Optimization at the boundary of the tensor network variety''.
\newblock \href{https://dx.doi.org/10.1103/PhysRevB.103.195139}{Phys. Rev. B
  {\bf 103}, 1--9}~(2021).

\bibitem{Co20}
P.~Comon, L.-H. Lim, Y.~Qi, and K.~Ye.
\newblock ``Topology of tensor ranks''.
\newblock \href{https://dx.doi.org/10.1016/j.aim.2020.107128}{Adv. Math. {\bf
  367}, 107128}~(2020).

\bibitem{Be23}
C.~Beltr\'{a}n, P.~Breiding, and N.~Vannieuwenhoven.
\newblock ``The average condition number of most tensor rank decomposition
  problems is infinite''.
\newblock \href{https://dx.doi.org/10.1007/s10208-022-09551-1}{Found. Comp.
  Math. {\bf 23}, 433--491}~(2023).

\bibitem{De15}
G.~De las Cuevas, T.~S. Cubitt, J.~I. Cirac, M.~M. Wolf, and
  D.~P\'erez-Garc\'ia.
\newblock ``Fundamental limitations in the purifications of tensor networks''.
\newblock \href{https://dx.doi.org/10.1063/1.4954983}{J. Math. Phys. {\bf 57},
  071902}~(2016).

\bibitem{Kl14}
M.~Kliesch, D.~Gross, and J.~Eisert.
\newblock ``Matrix-product operators and states: {NP}-hardness and
  undecidability''.
\newblock \href{https://dx.doi.org/10.1103/PhysRevLett.113.160503}{Phys. Rev.
  Lett. {\bf 113}, 160503}~(2014).

\bibitem{De13c}
G.~De las Cuevas, N.~Schuch, D.~P\'erez-Garc\'ia, and J.I. Cirac.
\newblock ``Purifications of multipartite states: Limitations and constructive
  methods''.
\newblock \href{https://dx.doi.org/10.1088/1367-2630/15/12/123021}{New J. Phys.
  {\bf 15}, 123021}~(2013).

\bibitem{De19}
G.~De las Cuevas and T.~Netzer.
\newblock ``Mixed states in one spatial dimension: decompositions and
  correspondence with nonnegative matrices''.
\newblock \href{https://dx.doi.org/10.1063/1.5127668}{J. Math. Phys. {\bf 61},
  41901}~(2020).

\bibitem{De19d}
G.~De~las Cuevas, M.~Hoogsteder~Riera, and T.~Netzer.
\newblock ``Tensor decompositions on simplicial complexes with invariance''.
\newblock \href{https://dx.doi.org/10.1016/j.jsc.2024.102299}{J. Symb. Comput.
  {\bf 124}, 102299}~(2024).

\bibitem{Fa14}
H.~Fawzi, J.~Gouveia, P.~A. Parrilo, R.~Z. Robinson, and R.~R. Thomas.
\newblock ``Positive semidefinite rank''.
\newblock \href{https://dx.doi.org/10.1007/s10107-015-0922-1}{Math. Program.
  {\bf 153}, 133--177}~(2015).

\bibitem{Ja13}
R.~Jain, Y.~Shi, Z.~Wei, and S.~Zhang.
\newblock ``Efficient protocols for generating bipartite classical
  distributions and quantum states''.
\newblock \href{https://dx.doi.org/10.1109/TIT.2013.2258372}{IEEE Trans. Inf.
  Theory {\bf 59}, 5171--5178}~(2013).

\bibitem{Gl19}
I.~Glasser, R.~Sweke, N.~Pancotti, J.~Eisert, and J.~I. Cirac.
\newblock ``Expressive power of tensor-network factorizations for probabilistic
  modeling, with applications from hidden markov models to quantum machine
  learning''.
\newblock \href{https://dx.doi.org/10.48550/arXiv.1907.03741}{Adv. NeurIPS {\bf
  32}, 1498--1510}~(2019).

\bibitem{Ya91}
M.~Yannakakis.
\newblock ``Expressing combinatorial optimization problems by linear
  programs''.
\newblock \href{https://dx.doi.org/10.1016/0022-0000(91)90024-Y}{J. Comput.
  System Sci. {\bf 43}, 441--466}~(1991).

\bibitem{Go12}
J.~Gouveia, P.~A. Parrilo, and R.~R. Thomas.
\newblock ``Lifts of convex sets and cone factorizations''.
\newblock \href{https://dx.doi.org/10.1287/moor.1120.0575}{Math. Oper. Res.
  {\bf 38}, 248--264}~(2013).

\bibitem{Fi12}
S.~Fiorini, S.~Massar, S.~Pokutta, H.~R. Tiwary, and R.~De Wolf.
\newblock ``Linear vs. semidefinite extended formulations: Exponential
  separation and strong lower bounds''.
\newblock \href{https://dx.doi.org/10.1145/2213977.2213988}{Proc. ACM Symp.
  Theory of Computing}~(2012).

\bibitem{Ja17}
R.~Jain, Z.~Wei, P.~Yao, and S.~Zhang.
\newblock ``Multipartite quantum correlation and communication complexities''.
\newblock \href{https://dx.doi.org/10.1007/s00037-016-0126-y}{Comput.
  Complexity {\bf 26}, 199--228}~(2017).

\bibitem{Co93b}
J.~E. Cohen and U.~G. Rothblum.
\newblock ``Nonnegative ranks, decompositions, and factorizations of
  nonnegative matrices''.
\newblock \href{https://dx.doi.org/10.1016/0024-3795(93)90224-C}{Linear Algebra
  Appl. {\bf 190}, 149--168}~(1993).

\bibitem{Pe07}
D.~P\'erez-Garc\'ia, F.~Verstraete, M.~M. Wolf, and J.~I. Cirac.
\newblock ``Matrix product state representations''.
\newblock \href{https://dx.doi.org/10.26421/qic7.5-6-1}{Quantum Inf. Comput.
  {\bf 7}, 401--430}~(2007).

\bibitem{Te10}
K.~Temme and F.~Verstraete.
\newblock ``Stochastic matrix product states''.
\newblock \href{https://dx.doi.org/10.1103/PhysRevLett.104.210502}{Phys. Rev.
  Lett. {\bf 104}, 210502}~(2010).

\bibitem{Fa22}
A.~Fawzi et~al.
\newblock ``Discovering faster matrix multiplication algorithms with
  reinforcement learning''.
\newblock \href{https://dx.doi.org/10.1038/s41586-022-05172-4}{Nature {\bf
  610}, 47--53}~(2022).

\bibitem{De20}
G.~De las Cuevas, A.~Klingler, and T.~Netzer.
\newblock ``Approximate tensor decompositions: disappearance of many
  separations''.
\newblock \href{https://dx.doi.org/10.1063/5.0033876}{J. Math. Phys. {\bf 62},
  093502}~(2021).

\bibitem{Ec36}
C.~Eckart and G.~Young.
\newblock ``The approximation of one matrix by another of lower rank''.
\newblock \href{https://dx.doi.org/10.1007/BF02288367}{Psychometrika {\bf 1},
  211--218}~(1936).

\bibitem{Qi16}
Y.~Qi, P.~Comon, and L.~H. Lim.
\newblock ``Semialgebraic geometry of nonnegative tensor rank''.
\newblock \href{https://dx.doi.org/10.1137/16M1063708}{SIAM J. Matrix Anal.
  {\bf 37}, 1556--1580}~(2016).

\bibitem{Sy52}
J.~J. Sylvester.
\newblock ``On the principles of the calculus of forms''.
\newblock \href{https://dx.doi.org/10.1017/CBO9781139151078.009}{Cambridge and
  Dublin Math. J. {\bf 7}, 52--97}~(1852).

\bibitem{Co11}
G.~Comas and M.~Seiguer.
\newblock ``On the rank of a binary form''.
\newblock \href{https://dx.doi.org/10.1007/s10208-010-9077-x}{Found. Comput.
  Math. {\bf 11}, 65--78}~(2011).

\bibitem{Ba13}
E.~Ballico and A.~Bernardi.
\newblock ``Tensor ranks on tangent developable of segre varieties''.
\newblock \href{https://dx.doi.org/10.1080/03081087.2012.716430}{Linear
  Multilinear Algebra {\bf 61}, 881--894}~(2013).

\bibitem{Li09}
L.~H. Lim and P.~Comon.
\newblock ``Nonnegative approximations of nonnegative tensors''.
\newblock \href{https://dx.doi.org/10.1002/cem.1244}{J. Chemom. {\bf 23},
  432--441}~(2009).

\bibitem{Sa09}
M.~Sanz, D.~P\'erez-Garc\'ia, M.~M. Wolf, and J.~I. Cirac.
\newblock ``A quantum version of {W}ielandt's inequality''.
\newblock \href{https://dx.doi.org/10.1109/TIT.2010.2054552}{IEEE Trans. Inf.
  Theory {\bf 56}, 4668--4673}~(2010).

\bibitem{De17}
G.~De las Cuevas, J.~I. Cirac, N.~Schuch, and D.~P\'erez-Garc\'ia.
\newblock ``Irreducible forms of matrix product states: Theory and
  applications''.
\newblock \href{https://dx.doi.org/10.1063/1.5000784}{J. Math. Phys. {\bf 58},
  121901}~(2017).

\bibitem{Sch81}
A.~Schönhage.
\newblock ``Partial and total matrix multiplication''.
\newblock \href{https://dx.doi.org/10.1137/0210032}{SIAM J.Comput. {\bf 10},
  434--455}~(1981).

\bibitem{Sh17c}
Y.~Shitov.
\newblock ``Counterexamples to {S}trassen’s direct sum conjecture''.
\newblock \href{https://dx.doi.org/10.4310/ACTA.2019.v222.n2.a3}{Acta Math.
  {\bf 222}, 363--379}~(2019).

\bibitem{Ch21}
M.~Christandl, F.~Gesmundo, M.~Michałek, and J.~Zuiddam.
\newblock ``Border rank nonadditivity for higher order tensors''.
\newblock \href{https://dx.doi.org/10.1137/20M1357366}{SIAM J. Matrix Anal.
  Appl. {\bf 42}, 503--527}~(2021).

\bibitem{Ch18b}
M.~Christandl, A.~K. Jensen, and J.~Zuiddam.
\newblock ``Tensor rank is not multiplicative under the tensor product''.
\newblock \href{https://dx.doi.org/10.1016/j.laa.2017.12.020}{Linear Algebra
  Appl. {\bf 543}, 125--139}~(2018).

\bibitem{Ch19d}
M.~Christandl, F.~Gesmundo, and A.~K. Jensen.
\newblock ``Border rank is not multiplicative under the tensor product''.
\newblock \href{https://dx.doi.org/10.1137/18M1174829}{SIAM J. Appl. Algebra
  Geom. {\bf 3}, 231--255}~(2019).

\bibitem{De21b}
G.~{De las Cuevas}, A.~Klingler, and T.~Netzer.
\newblock ``Polynomial decompositions with invariance and positivity inspired
  by tensors''.
\newblock
  \href{https://dx.doi.org/https://doi.org/10.1016/j.laa.2024.05.025}{Linear
  Algebra Appl. {\bf 698}, 537--588}~(2024).

\bibitem{Ci20}
J.~I. Cirac, D.~P\'erez-Garc\'ia, N.~Schuch, and F.~Verstraete.
\newblock ``Matrix product states and projected entangled pair states:
  Concepts, symmetries, and theorems''.
\newblock \href{https://dx.doi.org/10.1103/RevModPhys.93.045003}{Rev. Mod.
  Phys. {\bf 93}, 045003}~(2021).

\bibitem{Ch19f}
M.~Christandl, P.~Vrana, and J.~Zuiddam.
\newblock ``Asymptotic tensor rank of graph tensors: beyond matrix
  multiplication''.
\newblock \href{https://dx.doi.org/10.1007/s00037-018-0172-8}{Computational
  Complexity {\bf 28}, 57--111}~(2019).

\bibitem{Ho85}
R.~A. Horn and C.~R. Johnson.
\newblock ``Matrix analysis''.
\newblock \href{https://dx.doi.org/10.1017/cbo9780511810817}{Cambridge
  University Press}. ~(1985).
\newblock 2nd edition.

\bibitem{Bo11}
C.~Bocci, E.~Carlini, and F.~Rapallo.
\newblock ``Perturbation of matrices and nonnegative rank with a view toward
  statistical models''.
\newblock \href{https://dx.doi.org/10.1137/110825455}{SIAM J. Matrix Anal.
  Appl. {\bf 32}, 1500--1512}~(2011).

\bibitem{Go96}
G.~H. Golub and C.~F.~Van Loan.
\newblock ``Matrix computations''.
\newblock \href{https://dx.doi.org/10.56021/9781421407944}{Johns Hopkins
  University Press}. ~(1996).
\newblock 3rd edition.

\bibitem{Or14b}
R.~Or\'us.
\newblock ``A practical introduction to tensor networks: Matrix product states
  and projected entangled pair states''.
\newblock \href{https://dx.doi.org/10.1016/j.aop.2014.06.013}{Ann. Physics {\bf
  349}, 117--158}~(2014).

\end{thebibliography}

\onecolumn\newpage
\appendix
\renewcommand{\thesubsection}{\Alph{section}.\arabic{subsection}}

\captionsetup[subfigure]{justification=justified,singlelinecheck=false}

\section*{Appendices}

The following appendices are structured as follows:
\begin{itemize}
	\item In \Cref{app:OmegaG-decompositions}, we review the notion of $(\Omega,G)$-decompositions.
	\item In \Cref{app:rankInequalities}, we prove the tensor rank inequalities used in the main text.
	\item In \Cref{app:POVMQChannel}, we present the necessary background on quantum channels and POVMs.
	\item In \Cref{app:BolzanoWeierstrass}, we present the Bolzano--Weierstraß Theorem and its consequences.
	\item In \Cref{app:NonnegativeDec_CausalStructures} and \Cref{app:CorrelationCorrespondenceAppendix}, we prove the correspondences between ranks of positive tensor decompositions and certain (quantum) correlations.
	\item In \Cref{app:nnDecNoBorderRank}, we prove that the nonnegative and separable decomposition do not exhibit a gap between rank and border rank for standard tensor decompositions.
	\item In \Cref{app:treeDecompositions}, we prove that no tree tensor decomposition exhibits a gap between border rank and rank.
\end{itemize}

\section{Tensor decompositions on general weighted simplicial complexes}
\label{app:OmegaG-decompositions}

Here, we review the framework introduced in \cite{De19d} which generalizes tensor decompositions to arbitrary multi-hypergraph (weighted simplicial complex) structures with arbitrary symmetry constraints. It is structured as follows: In \Cref{ssec:wsc}, we review weighted simplicial complexes and group actions, which determine the structure of the tensor decomposition. In \Cref{ssec:tensordec}, we review the unconstrained, the non-negative, and the positive semidefinite tensor decomposition on weighted simplicial complexes $\Omega$ together with a group action $G$. In \Cref{ssec:matrixdec}, we review the separable decomposition and the purification form as examples of decompositions on matrix tensor product spaces. We also show that the examples discussed in the main text are special instances of $(\Omega,G)$-decompositions.

\subsection{Weighted simplicial complexes and group actions}
\label{ssec:wsc}

Here we give a definition of weighted simplicial complexes (wsc) $\Omega$ and group actions $G$ \cite{De19d}. Recall that $[n]$ denotes the set $\{1, \ldots, n\}$, and $\mathcal{P}_n$ the power set $\mathcal{P}([n])$ (i.e.\ the set of all subsets of $[n]$, which has $2^{n}$ elements).

\begin{definition}
\label{def:wsc}
(i) A \emph{weighted simplicial complex (wsc) on $[n]$} is a function $$\Omega: \mathcal{P}_n \to \N$$ where $S_1 \subseteq S_2$ implies that $\Omega(S_1)$ divides $\Omega(S_2)$.

(ii) A set $S \in \mathcal{P}_n$ is called  a \emph{simplex of $\Omega$}, if $\Omega(S) \neq 0$. We will assume that each singleton set $\{i\} \in \mathcal{P}_n$ is a simplex, and the elements $i \in [n]$ are the \emph{vertices} of the wsc. A maximal simplex (with respect to inclusion of sets) is a \emph{facet of $\Omega$}. The set of facets is denoted by
$$
\mathcal{F} \coloneqq \{F \in \mathcal{P}_n: F \text{ facet of } \Omega \}.
$$
The set of facets which contain vertex $i$ is denoted by
$$\mathcal{F}_i \coloneqq \{F \in \mathcal{F}: i \in F\}.
$$
The restriction of $\Omega$ to $\mathcal{F}$ gives rise to the multiset $\facetF$ which contains $F \in \mathcal{F}$ precisely $\Omega(F)$-times. 
The multiset $\facetF_i$ for $i \in [n]$ is defined analogously. 
There exists a canonical collapse map
$$c:\facetF \to \mathcal{F} \quad \text{ and } \quad c:\facetF_i \to \mathcal{F}_i$$
mapping all copies of a facet to the underlying facet. 

(iii) Two vertices $i,j$ are \emph{neighbors} if
$$ \mathcal{F}_i \cap \mathcal{F}_j \neq \emptyset$$
Two vertices $i,j$ are \emph{connected} if there is a sequence of vertices $i = i_1, i_2, \ldots, i_k = j$ such that $i_m$ and $i_{m+1}$ are neighbors for all ${m \in [k-1]}$.
\end{definition}

Before introducing group actions on weighted simplicial complexes, we define the notions of $G$-linearity and $G$-invariance. 

\begin{definition} Let $G$ be a group acting on the sets $X$,$Y$. A function $f:X \to Y$ is called \emph{$G$-linear} if
$$f(gx) = gf(x)$$
for all $x \in X$ and $g \in G$. If $G$ acts trivially on $Y$, $f$ is \emph{$G$-invariant}.
\end{definition}

The definition of a group action on $\Omega$ is split into two parts. First, we consider a group action of $G$ on the set $[n]$. Then, we extend this group action to the set of facets of $\Omega$. Without loss of generality, $G$ can be assumed to be a subgroup of the permutation group on the set $[n]$, $S_{n}$.  Every group action on $[n]$ canonically induces a group action on $\mathcal{P}_n$. If $\Omega$ is $G$-invariant, $G$ also induces a group action on $\mathcal{F}$, since for $F \in \mathcal{F}$, $g \in G$ and $F \subsetneq S$ it holds that $\Omega(gS) = \Omega(S) = 0$. This gives rise to the following definition.
\begin{definition} A group action of $G$ on the wsc $\Omega$ consists of the following:
\begin{enumerate}[label=(\roman*)]
	\item An action of $G$ on $[n]$ such that $\Omega$ is $G$-invariant with respect to the induced action on $\mathcal{P}_n$. This induces an action of $G$ on $\mathcal{F}$.
	\item An action of $G$ on $\widetilde{\mathcal{F}}$ such that the collapse map 
	$$c: \widetilde{\mathcal{F}} \to \mathcal{F}$$ is $G$-linear (we also say the action of $G$ on $\widetilde{\mathcal{F}}$ refines the action of $G$ on $\mathcal{F}$). 
\end{enumerate}
\end{definition}

In order to obtain  a group action on a wsc, one does not only have to specify how a group acts on $\mathcal{F}$, but also how it permutes the copies of facets in the multiset $\widetilde{\mathcal{F}}$. 

\begin{example}[Weighted simplicial complexes]
\label{ex:wsc}
Let us consider three examples of wsc that lead to the decompositions studied in the main text.
\begin{enumerate}[label=(\roman*)]
	\item\label{wsc:Sigman} The ``full simplex'' $\Sigma_n$ is given by $\Sigma_n: \mathcal{P}_n \to \mathbb{N}: S \mapsto 1$. The facet set is hence given by $\mathcal{F} = \big\{\{1, \ldots, n\}\big\}$.
	For $n = 5$, it is depicted as 
	\smallskip
	
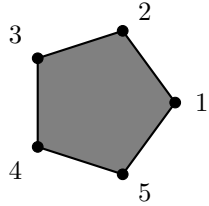
\begin{center}
	\begin{tikzpicture}

\draw[thick, fill=gray] (xyz polar cs:angle=0,radius=1)
-- (xyz polar cs:angle=72,radius=1)
-- (xyz polar cs:angle=144,radius=1)
-- (xyz polar cs:angle=216,radius=1)
-- (xyz polar cs:angle=288,radius=1)
-- cycle;

\draw[fill=black] (xyz polar cs:angle=0,radius=1) circle (2pt);
\draw[fill=black] (xyz polar cs:angle=72,radius=1) circle (2pt);
\draw[fill=black] (xyz polar cs:angle=144,radius=1) circle (2pt);
\draw[fill=black] (xyz polar cs:angle=216,radius=1) circle (2pt);
\draw[fill=black] (xyz polar cs:angle=288,radius=1) circle (2pt);

\node at (1.35,0) {$1$};
\node at (0.6,1.2) {$2$};
\node at (-1.1,0.9) {$3$};
\node at (-1.1,-0.9) {$4$};
\node at (0.6,-1.2) {$5$};

 \end{tikzpicture}
	\end{center}
	\smallskip
	where each node represents a vertex and the facet is represented by the gray area.
	\item\label{wsc:Thetan} The ``cycle'' is defined as
	$$\Theta_n: \mathcal{P}_n \to \mathbb{N}: S \mapsto \left\{\begin{array}{cl}
	1 & : \text{ if } |S| = 1 \text{ or } S = \{i, i+1\} \text{ for } i \in [n-1] \text{ or } \{n,1\}\\
	0 & : \text{ else.}
	\end{array}\right.$$
	Therefore, the set of facets is given by
	$$\mathcal{F} = \big\{\{1,2\}, \{2,3\}, \ldots, \{n-1, n\}, \{n,1\}\big\}.$$
	For $n = 5$, it is depicted as 
	\smallskip
	
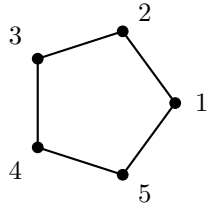
\begin{center}
	\begin{tikzpicture}

\draw[thick] (xyz polar cs:angle=0,radius=1)
-- (xyz polar cs:angle=72,radius=1)
-- (xyz polar cs:angle=144,radius=1)
-- (xyz polar cs:angle=216,radius=1)
-- (xyz polar cs:angle=288,radius=1)
-- cycle;

\draw[fill=black] (xyz polar cs:angle=0,radius=1) circle (2pt);
\draw[fill=black] (xyz polar cs:angle=72,radius=1) circle (2pt);
\draw[fill=black] (xyz polar cs:angle=144,radius=1) circle (2pt);
\draw[fill=black] (xyz polar cs:angle=216,radius=1) circle (2pt);
\draw[fill=black] (xyz polar cs:angle=288,radius=1) circle (2pt);

\node at (1.35,0) {$1$};
\node at (0.6,1.2) {$2$};
\node at (-1.1,0.9) {$3$};
\node at (-1.1,-0.9) {$4$};
\node at (0.6,-1.2) {$5$};

 \end{tikzpicture}
	\end{center}
	\smallskip
	where each line between two vertices $i$, $i+1$ represents the facet $\{i, i+1\}$.
	\item\label{wsc:Lambdan} The line $\Lambda_n$ is given by
	$$\Lambda_n: \mathcal{P}_n \to \mathbb{N}: S \mapsto \left\{\begin{array}{cl}
	1 & : \text{ if } |S| = 1 \text{ or } S = \{i, i+1\} \text{ for } i \in [n-1]\\
	0 & : \text{ else.}
	\end{array}\right.$$
	where $+$ is the ordinary addition on $\mathbb{N}$. Therefore, the set of facets is given by
	$$\mathcal{F} = \big\{\{1,2\}, \{2,3\}, \ldots, \{n-1, n\} \big\},$$
	depicted as
	\smallskip
	
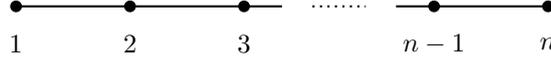
\begin{center}
	\begin{tikzpicture}

\draw[thick] (0,0) -- (3.5,0);
\draw[dotted, thick] (3.9,0) -- (4.6,0);
\draw[thick] (5,0) -- (7,0);

\draw[fill=black] (0,0) circle (2pt);
\draw[fill=black] (1.5,0) circle (2pt);
\draw[fill=black] (3,0) circle (2pt);
\draw[fill=black] (5.5,0) circle (2pt);
\draw[fill=black] (7,0) circle (2pt);

\node at (0,-0.5) {$1$};
\node at (1.5,-0.5) {$2$};
\node at (3,-0.5) {$3$};
\node at (5.5,-0.5) {$n-1$};
\node at (7,-0.5) {$n$};

 \end{tikzpicture}
	\end{center}
\end{enumerate}\demo
\end{example}

\begin{example}[Group actions on wsc]
\label{ex:GroupActions}
\begin{enumerate}[label=(\roman*)]
	\item\label{GroupActions:Sn} The group action $S_n$ of permutations on $[n]$ induces a group action on the full simplex $\Sigma_n$ since $S_n$ induces a group action of $S_n$ on $\mathcal{F}$ by mapping the simplex $\{1, \ldots, n\}$ to itself.
	\item\label{GroupActions:Cn} Let $C_n = \{e, \tau, \tau^2, \ldots, \tau^{n-1}\}$ be the group of translations on $[n]$, generated by $\tau: [n] \to [n]$ which is defined as $\tau(i) = i+1$ for $i \in [n-1]$ and $\tau(n) = 1$. $C_n$ induces a valid group action on the cycle $\Theta_n$, since facets $\{i, i+1\}$ are mapped to other facets $\{k, k+1\}$ by applying multiples of $\tau$ to every vertex. \demo
\end{enumerate}
\end{example}

In the upcoming definition of $(\Omega,G)$-decompositions (\Cref{def:tensor-dec}) we use $\beta \in \mathcal{I}^{\facetF_i}$ as summation indices. For this reason, we understand $\beta$ as a list of indices taking values in $\{1, \ldots, d\}$ where each index corresponds to a facet in $\facetF_i$, that is, $\beta = (\beta_{F_1}, \ldots, \beta_{F_n})$ where $\facetF_i = \{F_1, \ldots, F_n\}$. Moreover, we define the computational basis state corresponding to $\beta$ as
\begin{equation}
\label{eq:betaComputationalBasis}
\ket{\beta} \coloneqq \bigotimes_{F \in \facetF_i} \ket{\beta_{F}} \ \in \ \bigotimes_{i=1}^{|\facetF_i|} \mathbb{C}^{\mathcal{I}}. 
\end{equation}
For an element $g$ of a group $G$ acting on $\facetF$, and a function $(\alpha: \facetF \to \mathcal{I}) \in \mathcal{I}^{\facetF}$, we define
$$ {}^g \alpha: \facetF \to \mathcal{I}: F \mapsto \alpha(g^{-1} F)$$
We will often restrict to functions $\beta \in \mathcal{I}^{\facetF_i}$, i.e.\ defined on the facets $\facetF_i$ containing vertex $i$. In this situation, we have that ${}^{g} \beta: \facetF_{gi} \to \mathcal{I}$.

\subsection{$(\Omega,G)$-decompositions of nonnegative tensors}
\label{ssec:tensordec}

Let us now review the definition of $(\Omega,G)$-decompositions on nonnegative tensors $T \in \mathbb{C}^d \otimes \cdots \otimes \mathbb{C}^d$. For a more detailed exposition, we refer to \cite[Section 5]{De19d} and \cite{De19}. We call a tensor $T$ \emph{nonnegative} if $T_{j_1, \ldots, j_n} \geq 0$ for every $j_1, \ldots, j_n \in \{1,\ldots,d\}$. For any function $\alpha: \facetF \to \mathcal{I}$, we denote the restriction $\alpha_{|_{\sfacetF_i}}: \facetF_i \to \mathcal{I}$ by $\alpha_{|_i}$ for any $i \in [n]$.

\begin{definition}[Invariant decompositions of tensors \cite{De19d}]\label{def:tensor-dec} 
Let $T\in \mathbb{C}^{d}\otimes \cdots \otimes \mathbb{C}^{d} \cong \mathbb{C}^{d^n}$.
\begin{enumerate}[label=(\roman*),ref=(\roman*),leftmargin=*]
\item \label{def:tensor:i}
An \emph{$(\Omega,G)$-decomposition of $T$} is given by families of local tensors
$$\mathcal{T}^{[i]} = \left(\ket{T^{[i]}_{\beta}} \right)_{\beta \in \mathcal{I}^{\sfacetF_i}}$$
where $\ket{T^{[i]}_{\beta}} \in \mathbb{C}^{d}$ for all $i \in [n]$ and $\beta \in \mathcal{I}^{\facetF_i},$ such that
$$ T = \sum_{\alpha \in \mathcal{I}^{\sfacetF}} \ket{T^{[1]}_{\alpha_{\vert_1}}} \otimes \ket{T^{[2]}_{\alpha_{\vert_2}}} \otimes \cdots \otimes \ket{T^{[n]}_{\alpha_{\vert_n}}} $$
and
$$\ket{T^{[gi]}_{{}^g \beta}} = \ket{T^{[i]}_{\beta}}$$
for all $i \in [n]$ and $\beta \in \mathcal{I}^{\widetilde{\mathcal{F}}_i}$. The minimal cardinality of $\mathcal{I}$ among all $(\Omega,G)$-decompositions is called the \emph{$(\Omega,G)$-rank} of $T$, denoted  $\rank_{(\Omega,G)}(T)$. 
\item 
A \emph{nonnegative $(\Omega,G)$-decomposition of $T$} is an $(\Omega,G)$-decomposition of $T$ where all local vectors $\ket{T_{\beta}^{[i]}}$ have nonnegative entries in the standard basis, i.e.\ $\Cbraket{j}{T_{\beta}^{[i]}} \geq 0$ for all $j \in \{1, \ldots, d\}$.
The corresponding rank is called the \emph{nonnegative $(\Omega,G)$-rank} of $T$, denoted  $\nnrank_{(\Omega,G)}(T).$
\item 
A \emph{positive semidefinite $(\Omega,G)$-decomposition of $T$} consists of positive semidefinite matrices (indexed by $\beta, \beta' \in \mathcal{I}^{\facetF_i}$) $$ E_j^{[i]}\in \mathcal{M}_{\mathcal I^{\sfacetF_i}}^{+}(\mathbb{C})$$
 for $i \in [n]$ and $j \in \{1,\ldots, d\}$ such that  $$\bra{{}^g\beta} E_j^{[gi]} \ket{{}^g\beta'}= \bra{\beta} E_j^{[i]} \ket{\beta'}$$ for all $i,g,j,\beta,\beta'$, and $$T_{j_1, \ldots, j_n}=\sum_{\alpha,\alpha'\in\mathcal I^{\sfacetF}} \left(E_{j_1}^{[1]}\right)_{\alpha_{\mid_1}, \alpha'_{\mid_1}} \cdots \left(E_{j_n}^{[n]}\right)_{\alpha_{\mid_n}, \alpha'_{\mid_n}} = \sum_{\alpha,\alpha'\in\mathcal I^{\sfacetF}} \bra{\alpha_{|_1}} E_{j_1}^{[1]} \ket{\alpha'_{|_1}} \cdots \bra{\alpha_{|_n}} E_{j_n}^{[n]} \ket{\alpha'_{|_n}}$$ for all $j_1, \ldots, j_n$. 
The smallest cardinality of $\mathcal{I}$ among all positive semidefinite $(\Omega,G)$-decompositions is called the \emph{positive semidefinite $(\Omega,G)$-rank} of $T$, denoted ${\psdrank}_{(\Omega,G)}(T).$
\end{enumerate} 
\end{definition}

\begin{example}[Examples of $(\Omega,G)$-decompositions] Let us now present examples of $(\Omega,G)$-decompositions appearing in the main text.

\begin{enumerate}[label=(\roman*)]
	\item Let $\Omega = \Sigma_n$ be the full $n$-simplex (see \Cref{ex:wsc} \ref{wsc:Sigman}), i.e.\ the set of facets is a singleton $\mathcal{F} = \big\{ \{1, \ldots, n\} \big\}$.
	There is precisely one index in the decomposition; hence, a $\Sigma_n$-decomposition is given by
	$$\ket{T} = \sum_{\alpha = 1}^{r} \ket{T_\alpha^{[1]}} \otimes  \ket{T_\alpha^{[2]}} \otimes \cdots \otimes  \ket{T_\alpha^{[n]}}$$
	where $\ket{T_{\alpha}^{[i]}} \in \mathbb{C}^d$. The minimal natural number $r$ among all $\Sigma_n$-decompositions is the $\Sigma_n$-rank.	If, in addition, $\ket{T_{\alpha}^{[i]}}$ is a nonnegative vector in the standard basis for every $i \in \{1,\ldots, n\}$ and $\alpha \in \{1, \ldots, r\}$, this gives rise to a nonnegative $\Sigma_n$-decomposition. A positive semidefinite $\Sigma_n$-decomposition is given by
	$$T_{j_1, \ldots, j_n} = \sum_{\alpha, \alpha' = 1}^{r} \left(E_{j_1}^{[1]}  \right)_{\alpha, \alpha'} \cdots \left(E_{j_n}^{[n]}  \right)_{\alpha, \alpha'}$$
	with $E_j^{[i]} \geqslant 0$ being psd for every $i \in \{1, \ldots, n\}$ and $j \in \{1,\ldots, d\}$.
	
	When considering $\Sigma_n$-decompositions, we refer to their corresponding ranks via $\rank, \nnrank$ and $\psdrank$ without the index $\Sigma_n$.	
	
	\item Let $G = S_n$ be the full permutation group acting on the simplex $\Sigma_n$  (see \Cref{ex:GroupActions} \ref{GroupActions:Sn}).
	A $(\Sigma_n, S_n)$-decomposition of $\ket{T}$ is given by the additional constraint $\ket{T_{\alpha}^{[i_1]}} = \ket{T_{\alpha}^{[i_2]}} \eqqcolon \ket{T_{\alpha}}$ for every $i_1, i_2 \in \{1, \ldots, n\}$, i.e.\
	$$\ket{T} = \sum_{\alpha = 1}^{r} \ket{T_\alpha} \otimes  \ket{T_\alpha} \otimes \cdots \otimes  \ket{T_\alpha}.$$
	Similarly, a psd $(\Sigma_n, S_n)$-decomposition is given by 	
	$$T_{j_1, \ldots, j_n} = \sum_{\alpha, \alpha' = 1}^{r} \left(E_{j_1}\right)_{\alpha, \alpha'} \cdots \left(E_{j_n}\right)_{\alpha, \alpha'},$$
	i.e.\ the psd matrices $E_j^{[i]}$ coincide for every $i \in \{1, \ldots, n\}$.
	
	When considering $(\Sigma_n,S_n)$-decompositions, we refer to their corresponding ranks by $\symmrank, \symmnnrank$ and $\symmpsdrank$, and we do not use the subscript $(\Sigma_n,S_n)$.	
	
	\item Let $\Theta_n$ be a cycle of length $n$ (see \Cref{ex:wsc} \ref{wsc:Thetan}).
	There are $n$ indices $\alpha_i$ in the decomposition, each corresponding to one facet $\{i, i+1\}$. Hence, a $\Theta_n$-decomposition is given by
	$$T = \sum_{\alpha_1, \ldots, \alpha_n = 1}^{r} \ket{T_{\alpha_1, \alpha_2}^{[1]}} \otimes  \ket{T_{\alpha_2, \alpha_3}^{[2]}} \otimes \cdots \otimes  \ket{T_{\alpha_n, \alpha_1}^{[n]}}$$
	where $\ket{T_{\alpha_{i}, \alpha_{i+1}}^{[i]}} \in \mathbb{C}^d$. The minimal natural number $r$ among all $\Theta_n$-decompositions is the $\Theta_n$-rank.
	A $\Theta_n$-decomposition corresponds to the Matrix Product State (MPS) decomposition with closed boundary conditions \cite{Pe07, Ci20} and the rank corresponds to the bond dimension. If, in addition, $\ket{T_{\alpha, \beta}^{[i]}}$ is a nonnegative vector in the computational basis for every $i \in \{1,\ldots, n\}$, $\alpha, \beta \in \{1, \ldots, r\}$, this gives rise to a nonnegative $\Theta_n$-decomposition.
	
	A positive semidefinite $\Theta_n$-decomposition is given by
	$$T_{j_1, \ldots, j_n} = \sum_{\substack{\alpha_1, \ldots, \alpha_n,\\ \alpha'_1, \ldots, \alpha'_n = 1}}^{r} \bra{\alpha_1, \alpha_2} E_{j_1}^{[1]} \ket{\alpha'_1, \alpha'_2} \cdot \bra{\alpha_2, \alpha_3} E_{j_2}^{[2]} \ket{\alpha'_2, \alpha'_3} \cdots \bra{\alpha_n, \alpha_1} E_{j_n}^{[n]} \ket{\alpha'_n, \alpha'_1}$$
	with $E_j^{[i]} \geqslant 0$ for every $i \in \{1, \ldots, n\}$ and $j \in \{1,\ldots, d\}$.
	
	When considering $\Theta_n$-decompositions, we denote their corresponding ranks by $\osr, \nnosr$ and $\psdosr$, without the index $\Theta_n$.	
	
	\item Let $C_n$ the group of translations (see \Cref{ex:GroupActions} \ref{GroupActions:Cn}). A $(\Theta_n, C_n)$-decomposition is given by the additional constraints $\ket{T^{[i]}_{\alpha, \beta}} = \ket{T^{[j]}_{\alpha, \beta}} \eqqcolon \ket{T_{\alpha, \beta}}$, i.e.\ the decomposition is explicitly translational invariant. Hence,
	$$T = \sum_{\alpha_1, \ldots, \alpha_n = 1}^{r} \ket{T_{\alpha_1, \alpha_2}} \otimes  \ket{T_{\alpha_2, \alpha_3}} \otimes \cdots \otimes \ket{T_{\alpha_n, \alpha_1}}.$$
	This is known as a translational invariant MPS decomposition \cite{De19, Pe07} or the uniform Matrix Product State \cite{Ci20}.

	If, in addition, $\ket{T_{\alpha, \beta}}$ are nonnegative vectors in the computational basis, it is called a nonnegative $(\Theta_n, C_n)$-decomposition. Moreover,
	$$T_{j_1, \ldots, j_n} = \sum_{\substack{\alpha_1, \ldots, \alpha_n,\\ \alpha'_1, \ldots, \alpha'_n = 1}}^{r} \bra{\alpha_1, \alpha_2} E_{j_1} \ket{\alpha'_1, \alpha'_2} \cdot \bra{\alpha_2, \alpha_3} E_{j_2} \ket{\alpha'_2, \alpha'_3} \cdots \bra{\alpha_n, \alpha_1} E_{j_n} \ket{\alpha'_n, \alpha'_1}$$
	is called a psd $(\Theta_n,C_n)$-decomposition.
	
	When considering $(\Theta_n, C_n)$-decompositions, we denote their corresponding ranks by $\tiosr, \tinnosr$ and $\tipsdosr$. \demo
\end{enumerate}

\end{example}

\subsection{$(\Omega,G)$-decompositions of positive semidefinite matrices}
\label{ssec:matrixdec}

Let us now review the definition of $(\Omega,G)$-decompositions on matrices $\rho \in \mathcal{M}_d(\mathbb{C})^{n}$ that are psd or separable. Recall that $\rho$ is separable if there exists a decomposition
$$\rho = \sum_{\alpha = 1}^{n} \rho_{\alpha}^{[1]} \otimes \cdots \otimes \rho_{\alpha}^{[n]}$$
such that $\rho_{\alpha}^{[i]} \geqslant 0$. For a more detailed exposition, see \cite[Section 3]{De19d} and \cite{De19}.

\begin{definition}[Invariant decompositions of matrices \cite{De19d}]\label{def:matrix-dec} 
Let $\rho \in \mathcal{M}_d(\mathbb{C}) \otimes \cdots \otimes \mathcal{M}_d(\mathbb{C})$.
\begin{enumerate}[label=(\roman*),ref=(\roman*),leftmargin=*]
\item \label{def:tensor:i}
An \emph{$(\Omega,G)$-decomposition of $\rho$} is given by families
$$\mathcal{A}^{[i]} = \left(A^{[i]}_{\beta} \right)_{\beta \in \mathcal{I}^{\sfacetF_i}}$$
where $A^{[i]}_{\beta} \in \mathcal{M}_d(\mathbb{C})$ for all $i \in [n]$ and $\beta \in \mathcal{I}^{\mathcal{\widetilde{F}}_i},$ such that
$$ \rho = \sum_{\alpha \in \mathcal{I}^{\sfacetF}} A^{[1]}_{\alpha_{\vert_1}} \otimes A^{[2]}_{\alpha_{\vert_2}} \otimes \cdots \otimes A^{[n]}_{\alpha_{\vert_n}} $$
and
$$A^{[gi]}_{{}^g \beta} = A^{[i]}_{\beta}$$
for all $i \in [n]$ and $\beta \in \mathcal{I}^{\facetF_i}$. The minimal cardinality of $\mathcal{I}$ among all $(\Omega,G)$-decompositions is called the \emph{$(\Omega,G)$-rank} of $\rho$, denoted  $\rank_{(\Omega,G)}(\rho)$. 
\item A \emph{separable $(\Omega,G)$-decomposition of $\rho$} is an $(\Omega,G)$-decomposition with the additional constraint that 
$$A^{[i]}_{\beta} \geqslant 0 \quad \text{ for every } i \in \{1, \ldots, n\} \text{ and } \beta \in \mathcal{I}^{\facetF_i}.$$
The minimal cardinality of $\mathcal{I}$ among all separable $(\Omega,G)$-decompositions is called the \emph{separable $(\Omega,G)$-rank} of $\rho$, denoted  $\seprank_{(\Omega,G)}(\rho)$.

\item An \emph{$(\Omega,G)$-purification of $\rho$} is a factorization $\rho = L L^{\dag}$ where
$$L \in \mathcal{M}_{d, d'}(\mathbb{C}) \otimes \cdots  \otimes \mathcal{M}_{d, d'}(\mathbb{C})$$
together with an $(\Omega,G)$-decomposition of $L$.
The minimal cardinality of $\mathcal{I}$ in the decomposition of $L$ among all $(\Omega,G)$-purifications is called the \emph{$(\Omega,G)$ purification-rank} of $\rho$, denoted  $\purirank_{(\Omega,G)}(\rho)$. 
\end{enumerate}
\end{definition}

\begin{example}[$(\Omega,G)$-decompositions of multipartite matrices] We now present examples of $(\Omega,G)$-decompositions which appear in the main text.

\begin{enumerate}[label=(\roman*)]
	\item Let $\Omega$ be the full $n$-simplex $\Sigma_n$. A $\Sigma_n$-decomposition is given by
	$$\rho = \sum_{\alpha = 1}^{r} A_\alpha^{[1]} \otimes  A_\alpha^{[2]} \otimes \cdots \otimes  A_\alpha^{[n]}$$
	where $A_{\alpha}^{[i]} \in \mathcal{M}_{d}(\mathbb{C})$. The minimal natural number $r$ among all $\Sigma_n$-decompositions is the $\Sigma_n$-rank.
		
	If, in addition, $A_{\alpha}^{[i]} \geqslant 0$ for every $i \in \{1,\ldots, n\}$ and $\alpha \in \{1, \ldots, r\}$, the decomposition is a separable $\Sigma_n$-decomposition.
	
	We refer to the corresponding ranks of $\Sigma_n$-decompositions by $\rank, \seprank$ and $\purirank$, without the subscript $\Sigma_n$.	
	
	\item A $(\Sigma_n, S_n)$-decomposition of $\rho$ is given by the additional constraint $A_{\alpha}^{[i_1]} = A_{\alpha}^{[i_2]} \eqqcolon A_{\alpha}$,
	$$\rho = \sum_{\alpha = 1}^{r} A_\alpha \otimes  A_\alpha \otimes \cdots \otimes  A_\alpha$$
	Similarly, a $(\Sigma_n, S_n)$-purification is given by a factorization $\rho = L L^{\dag}$ with
	$$L = \sum_{\alpha = 1}^{r} L_\alpha \otimes  L_\alpha \otimes \cdots \otimes  L_\alpha$$
	
	We refer to the corresponding ranks of $(\Sigma_n,S_n)$-decompositions by $\symmrank, \symmseprank$ and $\symmpurirank$, without the index $(\Sigma_n,S_n)$.	
	
	\item Let $\Theta_n$ be a cycle of length $n$. A $\Theta_n$-decomposition is given by
	$$\rho = \sum_{\alpha_1, \ldots, \alpha_n = 1}^{r} A_{\alpha_1, \alpha_2}^{[1]} \otimes  A_{\alpha_2, \alpha_3}^{[2]} \otimes \cdots \otimes  A_{\alpha_n, \alpha_1}^{[n]}$$
	where $A_{\alpha,\beta}^{[i]} \in \mathcal{M}_d(\mathbb{C})$. The minimal natural number $r$ among all $\Theta_n$-decompositions is the $\Theta_n$-rank.
	
	A $\Theta_n$-decomposition corresponds to a Matrix Product Operator (MPO) decomposition \cite{De19}. The rank is known as the bond dimension or sometimes also called the \emph{operator Schmidt rank} (osr). This decomposition is represented with tensor networks in \Cref{fig:1Ddecompositions}.
	
	If, in addition, $A_{\alpha, \beta}^{[i]} \geqslant 0$ for every $i \in \{1,\ldots, n\}$ and $\alpha, \beta \in \{1, \ldots, r\}$, it gives rise to a separable $\Theta_n$-decomposition.
	
	A $\Theta_n$-purification is given by a factorization $\rho = L L^{\dag}$ together with
	$$L = \sum_{\alpha_1, \ldots, \alpha_n = 1}^{r} L_{\alpha_1, \alpha_2}^{[1]} \otimes  L_{\alpha_2, \alpha_3}^{[2]} \otimes \cdots \otimes  L_{\alpha_n, \alpha_1}^{[n]}.$$
	A $\Theta_n$-decomposition is better known as a local purification form in 1D \cite{De19,De19d}, depicted via tensor networks in \Cref{fig:1Ddecompositions}.
	
	We refer to the ranks of $\Theta_n$-decompositions by $\osr, \seposr$ and $\puriosr$, without using the index $\Theta_n$.

	\item A $(\Theta_n, C_n)$-decomposition is given by the additional constraints $A^{[i]}_{\alpha, \beta} = A^{[j]}_{\alpha, \beta} \eqqcolon A_{\alpha, \beta}$, i.e.\
	$$\rho = \sum_{\alpha_1, \ldots, \alpha_n = 1}^{r} A_{\alpha_1, \alpha_2} \otimes  A_{\alpha_2, \alpha_3} \otimes \cdots \otimes  A_{\alpha_n, \alpha_1}$$
	In words, the decomposition is explicitly translational invariant. This decomposition is also known as translational invariant operator Schmidt decomposition \cite{De19}.

	If, in addition, $A_{\alpha, \beta} \geqslant 0$, it is a separable $(\Theta_n, C_n)$-decomposition. A $(\Theta_n, C_n)$-purification is given by a factorization $\rho = L L^{\dag}$ together with a $(\Theta_n, C_n)$-decomposition
	$$L = \sum_{\alpha_1, \ldots, \alpha_n = 1}^{r} L_{\alpha_1, \alpha_2} \otimes  L_{\alpha_2, \alpha_3} \otimes \cdots \otimes  L_{\alpha_n, \alpha_1}.$$
	
	We refer to the corresponding ranks of $(\Theta_n, C_n)$-decompositions by $\tiosr, \tiseposr$ and $\tipuriosr$. \demo
\end{enumerate}

\end{example}

\begin{figure}
\centering
\begin{subfigure}[b]{0.44\textwidth}
\centering
\begin{tikzpicture}[scale=0.6]

\begin{scope}[xshift=-0.6cm]

\draw[thick, fill=color3] (-3.5,1) rectangle (-1,2) node[midway] {$\rho$};

\draw[thick] (-3.25,1) -- (-3.25,0.5);
\draw[thick] (-2.75,1) -- (-2.75,0.5);
\draw[thick] (-1.25,1) -- (-1.25,0.5);
\node at (-2,0.75) {$\cdots$};

\draw[thick] (-3.25,2.5) -- (-3.25,2);
\draw[thick] (-2.75,2.5) -- (-2.75,2);
\draw[thick] (-1.25,2.5) -- (-1.25,2);
\node at (-2,2.25) {$\cdots$};

\node at (-0.5,1.5) {$=$};
\end{scope}

\draw[very thick] (1,1.5) -- (2,1.5);
\draw[very thick] (3,1.5) -- (3.5,1.5);
\draw[very thick] (4.5,1.5) -- (5,1.5);

\node at (1.5,1.25) {\scriptsize $\alpha_2$};
\node at (3.35,1.25) {\scriptsize $\alpha_3$};
\node at (4.65,1.25) {\scriptsize $\alpha_n$};

\draw[very thick, rounded corners=1.5mm] (0,1.5) -- (-0.5,1.5) -- (-0.5, 2.25) -- (6.5, 2.25) -- (6.5,1.5) -- (6,1.5);
\node at (4,2.5) {\scriptsize $\alpha_1$};

\node at (4,1.5) {$\cdots$};

\draw[thick] (0.5,1) -- (0.5,0.5);
\draw[thick] (2.5,1) -- (2.5,0.5);
\draw[thick] (5.5,1) -- (5.5,0.5);

\draw[line width=2mm, white] (0.5,2) -- (0.5,2.5);
\draw[line width=2mm, white] (2.5,2) -- (2.5,2.5);
\draw[line width=2mm, white] (5.5,2) -- (5.5,2.5);

\draw[thick] (0.5,2) -- (0.5,2.5);
\draw[thick] (2.5,2) -- (2.5,2.5);
\draw[thick] (5.5,2) -- (5.5,2.5);

\draw[thick, fill=color3] (0,1) rectangle (1,2) node[midway] {$A^{[1]}$};
\draw[thick, fill=color3] (2,1) rectangle (3,2) node[midway] {$A^{[2]}$};
\draw[thick, fill=color3] (5,1) rectangle (6,2) node[midway] {$A^{[n]}$};

\draw[white,thick] (0.5,0) -- (0.5,-0.5);
\draw[white,thick] (0.5,3) -- (0.5,3.5);

 \end{tikzpicture}
\label{1Ddec:MPO}
\end{subfigure}
\hspace*{1.2cm}
\begin{subfigure}[b]{0.44\textwidth}
\centering
\begin{tikzpicture}[scale=0.6]

\begin{scope}[xshift=-0.6cm]
\draw[thick, fill=color3] (-3.5,1) rectangle (-1,2) node[midway] {$\rho$};

\draw[thick] (-3.25,1) -- (-3.25,0.5);
\draw[thick] (-2.75,1) -- (-2.75,0.5);
\draw[thick] (-1.25,1) -- (-1.25,0.5);
\node at (-2,0.75) {$\cdots$};

\draw[thick] (-3.25,2.5) -- (-3.25,2);
\draw[thick] (-2.75,2.5) -- (-2.75,2);
\draw[thick] (-1.25,2.5) -- (-1.25,2);
\node at (-2,2.25) {$\cdots$};

\node at (-0.4,1.5) {$=$};

\end{scope}

\draw[rectangle, fill=color1, draw=none] (-0.6, -0.3) rectangle (7.5,1.5);
\node at (7, 0.5) {$L$};

\draw[very thick, rounded corners=1.5mm] (0,0.5) -- (-0.5,0.5) -- (-0.5, 1.25) -- (6.5, 1.25) -- (6.5,0.5) -- (6,0.5);
\draw[very thick, rounded corners=1.5mm] (0,2.5) -- (-0.5,2.5) -- (-0.5, 3.25) -- (6.5, 3.25) -- (6.5,2.5) -- (6,2.5);

\draw[very thick] (1,0.5) -- (2,0.5);
\draw[very thick] (3,0.5) -- (3.5,0.5);
\draw[very thick] (4.5,0.5) -- (5,0.5);

\draw[very thick] (1,2.5) -- (2,2.5);
\draw[very thick] (3,2.5) -- (3.5,2.5);
\draw[very thick] (4.5,2.5) -- (5,2.5);

\node at (4,3.5) {\scriptsize $\alpha_1$};

\node at (1.5,2.25) {\scriptsize $\alpha_2$};
\node at (3.35,2.25) {\scriptsize $\alpha_3$};
\node at (4.65,2.25) {\scriptsize $\alpha_n$};

\node at (4,1.5) {\scriptsize $\alpha'_1$};

\node at (1.5,0.2) {\scriptsize $\alpha'_2$};
\node at (3.35,0.2) {\scriptsize $\alpha'_3$};
\node at (4.65,0.2) {\scriptsize $\alpha'_n$};

\node at (4,0.5) {$\cdots$};
\node at (4,2.5) {$\cdots$};

\draw[line width=2mm, white] (0.5,1) -- (0.5,2);
\draw[line width=2mm, white] (2.5,1) -- (2.5,2);
\draw[line width=2mm, white] (5.5,1) -- (5.5,2);

\draw[line width=2mm, color1] (0.5,1) -- (0.5,1.5);
\draw[line width=2mm, color1] (2.5,1) -- (2.5,1.5);
\draw[line width=2mm, color1] (5.5,1) -- (5.5,1.5);

\draw[thick] (0.5,1) -- (0.5,2);
\draw[thick] (2.5,1) -- (2.5,2);
\draw[thick] (5.5,1) -- (5.5,2);

\draw[thick] (0.5,0) -- (0.5,-0.5);
\draw[thick] (2.5,0) -- (2.5,-0.5);
\draw[thick] (5.5,0) -- (5.5,-0.5);

\draw[line width=2mm, white] (0.5,3) -- (0.5,3.5);
\draw[line width=2mm, white] (2.5,3) -- (2.5,3.5);
\draw[line width=2mm, white] (5.5,3) -- (5.5,3.5);

\draw[thick] (0.5,3) -- (0.5,3.5);
\draw[thick] (2.5,3) -- (2.5,3.5);
\draw[thick] (5.5,3) -- (5.5,3.5);

\draw[thick, fill=color3] (0,0) rectangle (1,1) node[midway] {\small $L^{[1]}$};
\draw[thick, fill=color3] (2,0) rectangle (3,1) node[midway] {\small $L^{[2]}$};
\draw[thick, fill=color3] (5,0) rectangle (6,1) node[midway] {\small $L^{[n]}$};

\draw[thick, fill=color3] (0,2) rectangle (1,3) node[midway] {\small $\bar{L}^{[1]}$};
\draw[thick, fill=color3] (2,2) rectangle (3,3) node[midway] {\small $\bar{L}^{[2]}$};
\draw[thick, fill=color3] (5,2) rectangle (6,3) node[midway] {\small $\bar{L}^{[n]}$};

 \end{tikzpicture}
\label{1Ddec:puri}
\end{subfigure}
\caption{(left) The Matrix product density operator (MPDO) form and (right) the local purification form drawn as a tensor network for $\Omega = \Theta_n$. Each line represents an index and each contraction a summation. The summation indices $\alpha_i$ form a cycle.}
\label{fig:1Ddecompositions}
\end{figure}

\begin{remark}
\label{rem:diagonalDec}
$(\Omega,G)$-decompositions of matrices generalize $(\Omega,G)$-decompositions on tensors. Given a tensor $T \in \mathbb{C}^d \otimes \cdots \otimes \mathbb{C}^d$, the diagonal matrix
$$ \rho_{T} \coloneqq \Diag(T) \coloneqq \sum_{j_1, \ldots, j_n = 1}^{d} T_{j_1, \ldots, j_n} \ket{j_1, \ldots, j_n} \bra{j_1, \ldots, j_n}$$
is psd if and only if $T$ is entrywise nonnegative. Moreover, every nonnegative $(\Omega,G)$-decomposition of $T$ corresponds to a separable $(\Omega,G)$-decomposition of $\rho_T$, and every psd $(\Omega,G)$-decomposition of $T$ corresponds to a $(\Omega,G)$-purification of $\rho_T$. In particular, we have that \cite[Theorem 43]{De19d}
\def\arraystretch{1.5}
\setlength{\tabcolsep}{0.5em}
\begin{center}
\begin{tabular}{r r c l}
(i) & $\rank_{(\Omega,G)}(T)$ & $ = $ & $\rank_{(\Omega,G)}(\rho_{T})$\\
(ii) & $\nnrank_{(\Omega,G)}(T)$ & $=$ & $\seprank_{(\Omega,G)}(\rho_{T})$\\
(iii) & $\psdrank_{(\Omega,G)}(T)$ & $=$ & $\purirank_{(\Omega,G)}(\rho_{T})$
\end{tabular}
\end{center}
\demo
\end{remark}

\subsection{The structure tensor $\ket{\Omega_r}$}
\label{ssec:structureTensor}

We now introduce for every wsc $\Omega$ a corresponding structure tensor $\ket{\Omega_r}$ which has $\rank_{\Omega}(\ket{\Omega_r}) \leq r$. This tensor allows for a more compact representation of $(\Omega,G)$-decompositions, which we will use in the proofs of \Cref{thm:CQCorr} and \Cref{thm:QQCorr}. Similar structure tensors, based on hyper-graphs, have been introduced in \cite{Ch19f, Ch21}.

Namely, for a given wsc $\Omega$, we define
$$\ket{\Omega_r} \coloneqq \sum_{\alpha \in \mathcal{I}^{\sfacetF}} \ket{\alpha_{|_1}} \otimes \cdots \otimes \ket{\alpha_{|_n}} \ \in \ \bigotimes_{i=1}^{n} \mathbb{C}^{r_i}$$
where $\mathcal{I} = \{1,\ldots, r\}$ and $r_i = |\mathcal{I}^{\facetF_i}|$. Since $\beta \in \mathcal{I}^{\facetF_i}$ can be understood as an array of indices in $\{1,\ldots, r\}$, $\ket{\beta}$ is defined as in Equation \eqref{eq:betaComputationalBasis}. 

Choosing for example the cycle $\Theta_n$, we obtain the $n$-fold matrix mutliplication (MaMu) tensor
$$ \ket{\Theta_{n,r}} = \sum_{\alpha_1, \ldots, \alpha_n = 1}^{r} \ket{\alpha_1, \alpha_2} \otimes \ket{\alpha_2, \alpha_3} \otimes \cdots \otimes \ket{\alpha_n, \alpha_1} \ \in \ \left(\mathbb{C}^r \otimes \mathbb{C}^r \right)^{\otimes n}.$$
For the $n$-fold simplex $\Sigma_n$, we obtain the unnormalized $n$-fold $r$-dimensional GHZ-state
$$ \ket{\Sigma_{n,r}} = \sum_{\alpha = 1}^{r} \ket{\alpha}^{\otimes n} \ \in \ \left(\mathbb{C}^r\right)^{\otimes n}.$$
Note that every $(\Omega,G)$-decomposition of rank $r$ can be written using $\ket{\Omega_r}$ as
\begin{equation}
\label{eq:OmegaRDecomposition}
\ket{T} = \sum_{\alpha \in \mathcal{I}^{\sfacetF}} \ket{v_{\alpha_{|_1}}^{[1]}} \otimes \cdots \otimes \ket{v_{\alpha_{|_n}}^{[n]}} = W^{[1]} \otimes \cdots \otimes W^{[n]} \ket{\Omega_r} \quad \text{ with } \quad W^{[i]} = \sum_{\beta \in \mathcal{I}^{\sfacetF_i}} \ket{v_{\beta}^{[i]}} \bra{\beta}.
\end{equation}
In this situation, $G$-invariance of $v^{[i]}_{\beta}$ translates to $W^{[gi]} \ket{{}^g \beta} = W^{[i]} \ket{\beta}$. Moreover, every psd $(\Omega,G)$-decomposition can be written as
$$T_{j_1, \ldots, j_n} = \sum_{\alpha, \alpha' \in \mathcal{I}^{\sfacetF}} \left(B_{j_1}^{[1]}\right)_{\alpha_{|_1}, \alpha'_{|_1}} \cdots \left(B_{j_n}^{[n]}\right)_{\alpha_{|_n}, \alpha'_{|_n}} = \bra{\Omega_r} B_{j_1}^{[1]} \otimes \cdots \otimes B_{j_n}^{[n]} \ket{\Omega_r}.$$
Note that in both examples, the corresponding $(\Omega,G)$-rank is given by the minimal parameter $r$ in the structure tensor that admits such a decomposition.

\section{Tensor rank inequalities}
\label{app:rankInequalities}
In the following, we will prove three inequalities used in the main text: that the rank of the $W_n$-state is $n$ (\Cref{ssec:rankWn}), the relation between $\rank$ and $\psdrank$ (\Cref{ssec:rank_psdrank}), and a lower bound for the symmetric psd-rank of $W_n$ (\Cref{ssec:symmetric_psdrank}).

\subsection{The rank of the $W_n$-state}
\label{ssec:rankWn}

Let us review a well-known example of a gap between border-rank and rank for the standard tensor rank. This statement has been proven for $n = 3$ in \cite{Si08} and generalizes to larger $n$. We prove it here for completeness.
The $W_n$-state has been studied from many different perspectives, for example, also bounds on the rank of $W_3^{\otimes k}$ have been established in \cite{Zu17}.

\begin{proposition}
\label{prop:WnLoweBound}
For $n \geq 2$, we have that $\rank(W_n) = n$.
\end{proposition}

\begin{proof}
That $\rank(W_n) \leq n$ is clear by the definition of $W_n$. We prove that $\rank(W_n) \geq n$ by induction. The case $n = 2$ is clear, since $W_2 \in \mathbb{C}^2 \otimes \mathbb{C}^2$ corresponds to the matrix 
$$W_2 = \ket{0}\bra{1} + \ket{1} \bra{0}$$
via the correspondence presented in Equation \eqref{eq:MatrixRankTensorRank}. Therefore $W_2$ has rank $2$.

For the induction step $n \to n+1$, suppose that $W_{n+1}$ has $\rank(W_{n+1}) \leq n$ with a decomposition
$$W_{n+1} = \sum_{\alpha = 1}^{n} \ket{v_{\alpha}^{[1]}} \otimes \cdots \otimes \ket{v_{\alpha}^{[n]}}.$$
For system $1$ we will prove that
\begin{enumerate}[label=(\alph*)]
	\item\label{indep} The vectors $\{\ket{v_{\alpha}^{[1]}}\}_{\alpha = 1,\ldots,n}$ span $\mathbb{C}^2$.
	\item\label{dep} $\ket{v_{\beta}^{[1]}} = c_{\beta} \ket{0}$ for every $\beta \in \{1, \ldots, n\}$.
\end{enumerate}
These two conditions contradict each other, hence proving the statement of the proposition.

To prove \ref{indep} assume that the family $\{\ket{v_{\alpha}^{[1]}}\}_{\alpha = 1,\ldots,n}$ does not span $\mathbb{C}^2$. Then there exists a non-zero vector $\ket{x} \in \mathbb{C}^2$ such that $\Cbraket{x}{v_{\alpha}^{[1]}} = 0$ for every $\alpha$. Appyling $\bra{x}$ to the first tensor factor of $W_{n+1}$ leads to
$$0 = \Cbraket{x}{0} W_n + \Cbraket{x}{1} \ket{0,0,0,\ldots,0}.$$
Since $\ket{W_n}$ and $\ket{0,\ldots,0}$ are linearly independent this implies that $\ket{x} = 0$, which is a contradiction.

To prove \ref{dep}, note that 
$$\rank(W_n + b \ket{0,\ldots, 0}) \geq \rank(W_n) \geq n$$ for every  $b \in \mathbb{R}$
since
$$W_n = A^{\otimes n} \Big(W_n + b \ket{0,0,0,\ldots,0}\Big)$$
with
$$A: \ket{0} \mapsto \ket{0}, \quad \ket{1} \mapsto \ket{1} - \frac{b}{n} \ket{0}.$$
This shows that $\rank(W_n + b \ket{0,\ldots,0}) \geq \rank(W_n)$ since the rank is non-increasing under local operations. Now let $\beta \in \{1, \ldots, r\}$ be fixed and choose $\ket{x} \in \mathbb{C}^2$ such that $\Cbraket{x}{v_{\beta}^{[1]}} = 0$. Applying $\bra{x}$ to the first tensor factor of $W_{n+1}$ we obtain

$$ \sum_{\alpha = 1, \alpha \neq \beta}^{n} \Cbraket{x}{v_{\alpha}^{[1]}} \ket{v_{\alpha}^{[2]}} \otimes \cdots \otimes \ket{v_{\alpha}^{[n]}} = \Cbraket{x}{0} W_n + \Cbraket{x}{1} \ket{0,0,0,\ldots,0}$$
Since the sum on the left hand side contains $n-1$ elementary tensors and the right hand side has rank at least $n$, if $\Cbraket{x}{0} \neq 0$, it follows that $\Cbraket{x}{0} = 0$. But this implies that $\ket{v_{\beta}^{[1]}} = c_\beta \ket{0}$.
\end{proof}

\begin{corollary}
For $n \geq 3$, we have that
$$ \brank(W_n) = 2 < n = \rank(W_n)$$
\end{corollary}

\subsection{Relating $\rank$ and $\psdrank$}
\label{ssec:rank_psdrank}

We now review the inequality between rank and the psd-rank shown in \cite[Corollary 44]{De19d}. For completeness, we also provide a proof of this statement.

\begin{lemma}[Relation between rank and psd-rank]
\label{lem:rank-psdrank-bound_general}
Let $\Omega$ be a wsc and $G$ a group action on $\Omega$.
For every nonnegative tensor $T$, we have
$$\rank_{(\Omega,G)}(T) \leq \psdrank_{(\Omega,G)}(T)^2$$
\end{lemma}

\begin{proof}
Let
$$T_{j_1, \ldots, j_n} = \sum_{\alpha, \alpha' \in \mathcal{I}^{\sfacetF}} \left(A_{j_1}^{[1]}\right)_{\alpha_{|_1}, \alpha'_{|_1}} \cdots \left(A_{j_n}^{[n]}\right)_{\alpha_{|_n}, \alpha'_{|_n}} $$
be a psd-decomposition of $T$ with $\psdrank_{(\Omega,G)}(T) = |\mathcal{I}|$.
Consider the new index set $\mathcal{L} \coloneqq \mathcal{I} \times \mathcal{I}$ with the two projections $p_{1,2}: \mathcal{L} \to \mathcal{I}$ and define
$$\Cbraket{j}{v_{\beta}^{[i]}} \coloneqq \left(A_{j}^{[i]}\right)_{p_1 \circ \beta, p_2 \circ \beta}$$
for every $i \in [n], j \in \{1, \ldots, d\}$ and $\beta \in \mathcal{L}^{\facetF_i}$. The new vectors $\{\ket{v_{\beta}}\}_{\beta \in \mathcal{L}^{\sfacetF_i}}$ provide a $(\Omega,G)$-decomposition of $T$ with
$$\rank_{(\Omega,G)}(T) \leq |\mathcal{L}| = |\mathcal{I}|^2 = \psdrank_{(\Omega,G)}(T)^2$$
which proves the statement.
\end{proof}

\subsection{A border rank gap for symmetric psd-decompositions in a tripartite system}
\label{ssec:symmetric_psdrank}

We now prove that the there is a gap between border rank and rank for the symmetric psd-rank already for $n = 3$.

\begin{proposition}
\label{prop:symmetryAnalysisWState}
There is a gap between border rank and rank of $W_3$ for the symmetric psd-rank. Specifically we have
$$\bsymmpsdrank(W_3) = 2 < 3 \leq \symmpsdrank(W_3)$$
\end{proposition}

\begin{proof}
That $\bsymmpsdrank(W_3) = 2$ is proven in \Cref{eq:defAmatrices}.

Now assume that $\symmpsdrank(W_3) = 2$.
Then there exists a symmetric psd-decomposition
$$W_{j_1, j_2, j_3} = \sum_{\alpha, \beta = 1}^{2} \left(A_{j_1}\right)_{\alpha, \beta} \cdot \left(A_{j_2}\right)_{\alpha, \beta} \cdot \left(A_{j_3}\right)_{\alpha, \beta}.$$
This decomposition can be expressed equivalently as
$$ W_{j_1, j_2, j_3} = \bra{M} A_{j_1} \star A_{j_2} \star A_{j_3} \ket{M}$$
where $\ket{M} = (1, \ldots, 1)^t$ and $\star$ is the Hadamard product, i.e.\
$$ \left(X \star Y\right)_{\alpha,\beta} = X_{\alpha, \beta} \cdot Y_{\alpha,\beta}.$$ 
We claim that $A_0$ and $A_1$ in the decomposition have rank $1$. Assume for example that $A_0$ has full rank, it is positive definite, therefore $A_0 \star A_0 \star A_0$ is positive definite by Schur's product theorem (see \cite[Theorem 7.5.3.]{Ho85}). But this implies that
$$0 = W_{0,0,0} = \bra{M} A_{0} \star A_{0} \star A_{0} \ket{M} > 0.$$
The same applies to $A_1$.

Since $A_0, A_1$ have rank $1$, we can parametrize them as
$$ A_{j} = \left(\begin{array}{cc} a_{j,0} & \sqrt{a_{j,0} a_{j,1}} \exp(i 2 \pi \varphi_j) \\  \sqrt{a_{j,0} a_{j,1}} \exp(-i 2 \pi \varphi_j) & a_{j,1} \end{array} \right)$$
where $a_{j,0}, a_{j,1} \geq 0$.
Since $W_{0,0,0} = W_{1,1,1} = 0$, we have that $a_{j,0} = a_{j,1}$ for $j = 0,1$ as well as $\varphi_j = 1/2$ which implies that $W_{j_1,j_2,j_3} = 0$ for all $j_1, j_2, j_3 \in \{0,1\}$.
\end{proof}

\section{Positive Operator Valued Measures and Quantum Channels}
\label{app:POVMQChannel}

We call a collection of matrices $r \times r$ matrices $(A_j)_{j = 1}^d$ a \emph{positive operator-valued measure} (POVM), if $A_j \geqslant 0$ and
$$ \sum_{j=1}^{d} A_j = \one_r.$$
Note that if $(A_{j_1})_{j_1 \in [d_1]}$, $(B_{j_2})_{j_2 \in [d_2]}$ are POVMs, then the tensor product
$$\left(A_{j_1} \otimes B_{j_2}\right)_{j_1 \in [d_1], j_2 \in [d_2]}$$
is also a POVM. The definition of a POVM guarantees that for any quantum state $\ket{\psi} \in \mathbb{C}^{r}$ the $d$-ary array
$$\big(p_j \coloneqq \bra{\psi} A_j \ket{\psi} \big)_{j=1, \ldots, d}$$
forms a probability distribution.

A linear map $\mathcal{E}: \mathcal{M}_{d_1}(\mathbb{C}) \to \mathcal{M}_{d_2}(\mathbb{C})$ is called \emph{completely positive} (cp), if $\id_{d'} \otimes \mathcal{E}$ is a positive map for every $d' \in \mathbb{N}$, i.e.\
$$(\id_{d'} \otimes \mathcal{E})(A) \geqslant 0 \quad \text{ for every psd matrix } A \in \mathcal{M}_{d'}(\mathbb{C}) \otimes \mathcal{M}_{d_1}(\mathbb{C}).$$
It is additionally called \emph{completely positive trace preserving} (cptp) or \emph{quantum channel} if it is cp and preserves the trace, i.e.\ $\tr(\mathcal{E}(A)) = \tr(A)$ for every $A \in \mathcal{M}_{d}(\mathbb{C})$.

Note that every POVM $(A_j)_{j=1}^{d}$ gives rise to the following quantum channel:
\begin{equation} 
\label{eq:DiagCPTP}
\mathcal{E}: \rho \mapsto \sum_{j=1}^{d} \ket{j} \bra{j} \tr(A_j \rho).
\end{equation}
In words, it maps every state to the classical state (i.e.\ a probability distribution) that is obtained when performing the measurement.
Conversely, every quantum channel whose image consists only of diagonal matrices gives rise to a POVM via Equation \eqref{eq:DiagCPTP}.

\section{The Bolzano--Weierstraß Theorem and its Applications}
\label{app:BolzanoWeierstrass}

In this part, we review a version of the Bolzano--Weierstraß Theorem. This allows to show that certain $(\Omega,G)$-decompositions do not have a gap between rank and border rank.

Let $(\mathcal{V}, \Vert \cdot \Vert)$ be a finite-dimensional normed vector space. A set $\mathcal{S} \subseteq \mathcal{V}$ is \emph{bounded} if there exists a constant $C \geq 0$ such that $\Vert x \Vert \leq C$ for every $x \in \mathcal{S}$. $\mathcal{S}$ is \emph{closed} if for every sequence $(s_k)_{k \in \mathbb{N}} \in \mathcal{S}^{\mathbb{N}}$ that converges to a point $s \in \mathcal{V}$ with respect to $\Vert \cdot \Vert$, we have that $s \in \mathcal{S}$. $\mathcal{S}$ is \emph{compact} if it is closed and bounded.

\begin{theorem}[Bolzano--Weierstraß for finite-dimensional vector spaces]
\label{thm:BolzanoWeierstrass}
Let $\mathcal{S} \subseteq \mathcal{V}$ be a compact subset of a finite-dimensional vectorspace $\mathcal{V}$. Then \emph{every} sequence $(s_i)_{i \in \mathbb{N}} \in \mathcal{S}^\mathbb{N}$ has a convergent subsequence, i.e.\ there is strictly increasing sequence $(k_\ell)_{\ell \in \mathbb{N}}$ in $\mathbb{N}$ such that
$$ \lim_{\ell \to \infty} s_{k_\ell} = s \in \mathcal{S}.$$
\end{theorem}

We now apply it to the space of quantum channels. We denote $\CPTP(d_1,d_2) \subseteq \Lin\left(\mathcal{M}_{d_1}(\mathbb{C}), \mathcal{M}_{d_2}(\mathbb{C})\right)$ the set of quantum channels in $\mathcal{M}_{d_1}(\mathbb{C}) \to \mathcal{M}_{d_2}(\mathbb{C})$.

\begin{lemma}
\label{lem:cptpBounded}
$\CPTP(d_1, d_2)$ is compact in $\Lin\left(\mathcal{M}_{d_1}(\mathbb{C}), \mathcal{M}_{d_2}(\mathbb{C})\right)$.
\end{lemma}
\begin{proof}
Equipping the space $\Lin(\mathcal{M}_{d_1}(\mathbb{C}), \mathcal{M}_{d_2}(\mathbb{C}))$ with the norm
$$\Vert \mathcal{E} \Vert \coloneqq \max_{\Vert \rho \Vert_{1} \leq 1} \Vert \mathcal{E}(\rho) \Vert_{1}$$
where $\Vert \, \cdot \, \Vert_{1}$ is the trace-norm on $\mathcal{M}_{d_i}(\mathbb{C})$, we obtain that
$\Vert \mathcal{E} \Vert \leq 1$ for every $\mathcal{E} \in \CPTP(d_1, d_2)$ which shows the boundedness.

Moreover, $\CPTP(d_1, d_2)$ is closed since it can be characterized by the closed conditions $\id_{n} \otimes \mathcal{E}(A) \geqslant 0$ for every psd $A \in \mathcal{M}_{d_1 \cdot n}(\mathbb{C})$ and $\tr(\mathcal{E}(\rho)) = \tr(\rho)$ for every $\rho \in \mathcal{M}_{d_1}(\mathbb{C})$. Since intersections of closed sets are closed, the statement follows.
\end{proof}

\begin{corollary}
\label{cor:cptpConvSeq}
Every sequence of quantum channels has a convergent subsequence.
\end{corollary}
\begin{proof}
This is an immediate consequence of \Cref{thm:BolzanoWeierstrass} and \Cref{lem:cptpBounded}.
\end{proof}

\section{Nonnegative $\Sigma_n$-decompositions and causal structures}
\label{app:NonnegativeDec_CausalStructures}

We now show that characterizing the nonnegative $\Sigma_n$-rank of a probability tensor $P$ is equivalent to $P$ arising from a specific Bayesian network with hidden complexity constraints. This statement has been proven in the bipartite case  \cite{Co93b} as well as in the multipartite case in \cite{Li09}.

\begin{theorem}[The nonnegative rank and classical correlations]
\label{thm:nndec_CausalStruct}
 Let $P \in \left(\mathbb{R}^{d}\right)^{\otimes n}$ be a nonnegative vector representing a probability distribution, i.e. 
 $$P_{j_1, \ldots, j_n} = P(X_1 = j_1, \ldots, X_n = j_n).$$
The following are equivalent:
\begin{enumerate}[label=(\roman*)]
	\item\label{nnrankI} $\nnrank(P) \leq r$
	\item\label{nnrankII} There exists a random variable $Z$ taking values in $\{1, \ldots, r\}$ such that
	$$P(X_1 = j_1, \ldots, X_n = j_n) = \sum_{\alpha=1}^r P(X_1 = j_1 \ | \ Z = \alpha) \cdots P(X_n = j_n \ | \ Z = \alpha) \cdot P(Z = \alpha)$$
\end{enumerate}

The same equivalence holds for $\symmnnrank$ with the additional constraint that the conditional probability distributions $P(X_i = - | Z = z)$ are identical for every $i \in \{1, \ldots, n\}$.
\end{theorem}
As in the main text, we denote the set of all probability distributions satisfying \ref{nnrankII} by $\CCorr(n,d,r)$.
\begin{proof}
We show the equivalence only for $\symmnnrank$ as the other follows analogously.

\ref{nnrankI} $\Longrightarrow$ \ref{nnrankII}: Since $\symmrank(P) \leq r$ there is a nonnegative decomposition
\begin{equation}
\label{eq:symmProb}
P = \sum_{\alpha = 1}^{r} \ket{v_{\alpha}} \otimes \cdots \otimes \ket{v_{\alpha}}.
\end{equation}
Define
$$ P(X_i = j | Z = \alpha) \coloneqq \frac{\Cbraket{j}{v_{\alpha}}}{\sum_{j = 1}^{d} \Cbraket{j}{v_{\alpha}}}$$
and
$$ P(Z = \alpha) = \left(\sum_{j = 1}^{d} \Cbraket{j}{v_{\alpha}}\right)^n.$$
By definition $P(X_i = - | Z = \alpha)$ is a probability distribution. Moreover, $P(Z = -)$ is a probability distribution since
\begin{align*} \sum_{\alpha = 1}^{r} P(Z = \alpha) &= \sum_{\alpha = 1}^{r} \left(\sum_{j = 1}^{d} \Cbraket{j}{v_{\alpha}}\right)^n = \sum_{\alpha = 1}^{r} \sum_{j_1, \ldots, j_n = 1}^{d} \bra{j_1, \ldots, j_n} \left(\ket{v_{\alpha}}\right)^{\otimes n} \\ &= \sum_{j_1, \ldots, j_n = 1}^{d} \bra{j_1, \ldots, j_n} \left(\sum_{\alpha = 1}^{r} \ket{v_{\alpha}}^{\otimes n} \right) = \sum_{j_1, \ldots, j_n = 1}^{d} P(X_1 = j_1, \ldots, X_n = j_n) = 1
\end{align*}

\ref{nnrankII} $\Longrightarrow$ \ref{nnrankI}: Let
	$$P(X_1 = j_1, \ldots, X_n = j_n) = \sum_{\alpha=1}^r P(X_1 = j_1 | Z = \alpha) \cdots P(X_n = j_n | Z = \alpha) \cdot P(Z = \alpha).$$
Defining
$$ \ket{v_{\alpha}^{[i]}} \coloneqq \sum_{j=1}^{r} P(X_i = j| Z=\alpha) \cdot P(Z= \alpha)^{\frac{1}{n}} \ket{j}$$
gives rise to nonnegative vectors in the computational basis.
Since all conditional distributions $P(X_i = - | Z= \alpha)$ are identical, we have that $\ket{v_{\alpha}^{[i]}} =  \ket{v_{\alpha}^{[j]}} \eqqcolon \ket{v_{\alpha}}$ for every $i,j \in \{1,\ldots, n\}$. It is immediate that Equation \eqref{eq:symmProb} holds.
\end{proof}

\section{The purification-rank corresponds to minimal entanglement in a correlation}
\label{app:CorrelationCorrespondenceAppendix}

We now prove the characterization theorems for the psd-rank and for the purification-rank on the level of arbitrary $(\Omega,G)$-decompositions (see \Cref{thm:quantumCorrMainText} in the main text for a simplified version). We first define the correlation sets $\CQCorr_{(\Omega,G)}$ and $\QQCorr_{(\Omega,G)}$ (\Cref{ssec:corrSets}) and subsequently prove the characterization of $\psdrank$ (\Cref{thm:CQCorr}) in \Cref{ssec:psdrankCorrespondence} and the characterization of $\purirank$ (\Cref{thm:QQCorr}) in \Cref{ssec:purirankCorrespondence}.

\subsection{The correlation sets}
\label{ssec:corrSets}

In the following, we define the set $\CQCorr_{(\Omega,G)}(n,d,r)$ and the set $\QQCorr_{(\Omega,G)}(n,d,r)$. Intuitively, every element in either set is generated by an $n$-partite quantum state as a resource that satisfies $\rank_{(\Omega,G)}(\ket{\psi}) \leq r$. In other words, $\ket{\psi}$ is limited in its entanglement structure. An $n$-partite probability distribution with local dimension $d$ in $\CQCorr_{(\Omega,G)}(n,d,r)$ is then generated via a POVM on each tensor product factor, and an $n$-partite mixed state in $\CQCorr_{(\Omega,G)}(n,d,r)$ is generated via quantum channel on each tensor product factor. Moreover, if $G$ is non-trivial, then the POVMs and the quantum channels satisfy additionally a symmetry constraint.

\begin{definition}[Quantum correlation scenarios for $(\Omega,G)$-structures]\phantom{*}\\ Let $\Omega$ be a wsc and $G$ be a group acting on $\Omega$.

\begin{enumerate}[label=(\roman*)]
	\item The set of \emph{multipartite quantum-correlation probability distributions} $\CQCorr_{(\Omega,G)}(n,d,r)$ is the set of all $n$-fold $d$-dimensional probability distributions $\left(P_{j_1, \ldots, j_n}\right)_{j_1, \ldots, j_n = 1}^{d}$ such that
	$$P_{j_1, \ldots, j_n} = \tr\left(A_{j_1}^{[1]} \otimes \cdots \otimes A_{j_n}^{[n]} \ket{\psi} \bra{\psi} \right)$$
	where $\ket{\psi}$ is a normalized state with $\rank_{(\Omega,G)}(\ket{\psi}) \leq r$ and $\left(A_{j}^{[i]}\right)_{j=1,\ldots,d}$ with $i=1, \ldots, n$ are $G$-invariant collection of POVMs, i.e.\ for $g \in G$ we have that $$A^{[gi]}_j = A^{[i]}_j.$$
	\item The set of \emph{multipartite quantum-correlation quantum states} $\QQCorr_{(\Omega,G)}(n,d,r)$ is the set of all $n$-fold density matrices $\rho \in \mathcal{M}_d(\mathbb{C})^{\otimes n}$ such that
	$$\rho = \left(\mathcal{E}_1 \otimes \cdots \otimes \mathcal{E}_n\right)\big(\ket{\psi} \bra{\psi}\big)$$
	where $\ket{\psi}$ is a normalized state with $\rank_{(\Omega,G)}(\ket{\psi}) \leq r$ and $\mathcal{E}_1, \ldots, \mathcal{E}_n$ are $G$-invariant family of quantum channels, i.e.\ for all $g \in G$ we have $$\mathcal{E}_{gi} = \mathcal{E}_i.$$
\end{enumerate}
\end{definition}

We will prove the correspondence between these sets of correlations and sets of low-rank tensors for arbitrary wsc $\Omega$ and a certain subclass of group actions on $\Omega$, called \emph{external} group actions.

\begin{definition}[External group action]
Let $\Omega$ be a wsc. We call a group action of $G$ on $\Omega$ \emph{external}, if 
$$\forall g \in G \text{ s.t. } gi = i \text{ we have that } gF = F \text{ for every } F \in \facetF_i.$$
\end{definition}

Intuitively, an external group action of $G$ does not give rise to any local constraints on the tensors. It only gives rise to constraints between tensors of different local systems.

\begin{example}[Examples of (non-)external group actions] All group actions used in the main text (see \Cref{fig:borderRankResults}) are external:
\begin{enumerate}[label=(\roman*)]
	\item The group action $C_n$ on the $n$-cycle $\Theta_n$ is external. This follows since the only $g \in C_n$ such that $gi = i$ for some $i \in [n]$ is the neutral element which also keeps the facets fixed.
	\item The group action $S_n$ on the $n$-simplex $\Sigma_n$ is external. In contrast to (i), there are $g \in S_n$ such that $gi = i$ for some $i \in [n]$. However, since $\facetF = \facetF_i$ are singletons the facet keeps trivially fixed.
\end{enumerate}

For a group action that is not external, consider the line $\Lambda_3$ with $3$ vertices (see \Cref{ex:wsc}) 
together with the group action $C_3$ generated from the reflection $\tau: 1 \mapsto 3, 2 \mapsto 2, 3 \mapsto 1$ and its induced action on the facets. $C_3$ is not external on $\Lambda_3$ since $\tau(2) = 2$ while $\tau F_1 = F_2$.
The corresponding $(\Lambda_3, C_3)$-decomposition is given by
$$ T = \sum_{\alpha, \beta = 1}^{r} \ket{v_{\alpha}} \otimes \ket{w_{\alpha, \beta}} \otimes \ket{v_{\alpha}}$$
with the additional constraint $\ket{w_{\alpha,\beta}} = \ket{w_{\beta, \alpha}}$ for every $\alpha,\beta \in \{1, \ldots, r\}$. The appearing ``internal'' symmetry constraint on the local family of tensors $\{\ket{w_{\alpha,\beta}}\}_{\alpha,\beta}$ motivates the name of the definition. \demo
\end{example}

\subsection{Correspondence between $\psdrank$ and $\CQCorr$}
\label{ssec:psdrankCorrespondence}

We now prove the first part of \Cref{thm:quantumCorrMainText}, namely that elements of $\CQCorr_{(\Omega,G)}(n,d,r)$ are precisely these tensors with $\psdrank_{(\Omega,G)}(P) \leq r$. A similar statement in a special case (namely $\Omega = \Sigma_n$ and $G = \{e\}$) is proven by Jain et al \cite[Theorem 13]{Ja17}.

\begin{theorem}[The psd-rank and quantum correlation scenarios]
\label{thm:CQCorr}
Let $\Omega$ be a wsc and $G$ an external group action on $\Omega$. Further, let $P \in \left(\mathbb{R}^d\right)^{\otimes n}$ be an $n$-fold probability distribution. The following are equivalent:
\begin{enumerate}[label=(\roman*)]
	\item \label{thm:QCorrI} $P \in \CQCorr_{(\Omega,G)}(n,d,r)$.
	\item \label{thm:QCorrII} $\psdrank_{(\Omega,G)}(P) \leq r$.
\end{enumerate}
\end{theorem}
We first need a preparatory lemma.
\begin{lemma}[$G$-symmetric matrix diagonalization]
\label{lem:GsymmetricDiagonalization}
Let $\Omega$ be a wsc and $G$ an external group action on $\Omega$. Let $K^{[i]} \in \mathcal{M}_{\mathcal{I}^{\sfacetF_i}}(\mathbb{C})$ for $i = 1, \ldots, n$ be a family of Hermitian matrices such that
$$\bra{{}^g \beta} K^{[gi]} \ket{{}^g \beta'} = \bra{\beta}  K^{[i]} \ket{\beta'} \quad \text{ for all } \beta,\beta' \in \mathcal{I}^{\facetF_i}$$
Then, there exists a compatible eigendecomposition of all matrices $K^{[i]}$ given by
$$K^{[i]} = \sum_{\ell = 1}^{m} \lambda_{\ell}^{[i]} \ket{w_{\ell}^{[i]}} \bra{w_{\ell}^{[i]}} \quad \text{ such that } \ \Cbraket{{}^g \beta}{w_{\ell}^{[gi]}} =  \Cbraket{\beta}{w_{\ell}^{[i]}} \textrm{ and } \lambda_{\ell}^{[gi]} = \lambda_{\ell}^{[i]}$$
\end{lemma}
\begin{proof}
Choose $i_1, \ldots, i_m \in [n]$ representations of the different orbits of the group actions. Computing the eigenvectors and eigenvalues of $K^{[i_1]}, \ldots, K^{[i_m]}$ we obtain a generating set of eigen-decompositions for every matrix $K^{[i]}$ by setting
$$\lambda^{[i]}_{\ell} = \lambda_{\ell}^{[g i_k]} \text{ and } \ket{w_{\ell}^{[i]}} = \sum_{\beta \in \mathcal{I}^{\sfacetF_{i_k}}} \ket{{}^g \beta} \Cbraket{\beta}{w_{\ell}^{[i_k]}}$$
for $g \in G$ and a representative $i_k$ such that $i = g i_k$.

Since the action is external, this is independent of the choice of $g$.
\end{proof}

\begin{proof}[Proof of \Cref{thm:CQCorr}]
\ref{thm:QCorrI} $\Rightarrow$ \ref{thm:QCorrII}: Let $P \in \CQCorr_{(\Omega,G)}(n,d,r)$. By definition, there exist a state
$$\ket{\psi} = \sum_{\alpha \in \mathcal{I}^{\sfacetF}} \ket{v^{[1]}_{\alpha_{|_1}}} \otimes \cdots \otimes \ket{v^{[n]}_{\alpha_{|_n}}} $$
with $|\mathcal{I}| \leq r$
and $G$-invariant POVMs $\left(A^{[i]}_{j}\right)_{j=1}^{d}$ such that
$$ P_{j_1, \ldots, j_n} = \tr\left(A^{[1]}_{j_1} \otimes \cdots \otimes A^{[n]}_{j_n} \ket{\psi} \bra{\psi}\right).$$
Define
$$B^{[i]}_j \coloneqq \left(X^{[i]}\right)^{\dag} A^{[i]}_j X^{[i]} \quad \text{ with } \quad X^{[i]} = \sum_{\beta \in \mathcal{I}^{\sfacetF_i}} \ket{v_{\beta}^{[i]}} \bra{\beta}. $$
Note that $B^{[i]}_j \in \mathcal{M}^{+}_{\mathcal{I}^{{\sfacetF_i}}}(\mathbb{C})$. Moreover, we have
$$\bra{{}^g \beta} B^{[gi]}_j \ket{{}^g \beta'} = \bra{v^{[gi]}_{{}^{g} \beta}} A^{[gi]}_j \ket{v^{[gi]}_{{}^g \beta'}} = \bra{v^{[i]}_{\beta}} A^{[i]}_j \ket{v^{[i]}_{\beta'}} = \bra{\beta} B^{[i]}_j \ket{\beta'}$$
where we have used that $\ket{v_{\beta}^{[i]}}$ forms a $(\Omega,G)$-decomposition and that $A_j^{[i]}$ are $G$-invariant. Moreover,
\begin{align*} \sum_{\alpha, \alpha' \in \mathcal{I}^{\sfacetF}}  \left(B^{[1]}_{j_1}\right)_{\alpha_{|_1}, \alpha'_{|_1}} \cdots \left(B^{[n]}_{j_n}\right)_{\alpha_{|_n}, \alpha'_{|_n}} =  \bra{\psi} A^{[1]}_{j_1} \otimes \cdots \otimes A^{[n]}_{j_n} \ket{\psi} = P_{j_1, \ldots, j_n}
\end{align*}
which proves that $\textrm{psd-rank}_{(\Omega,G)}(P) \leq r$.

\ref{thm:QCorrII} $\Rightarrow$ \ref{thm:QCorrI}:
Let 
\begin{equation}
\label{eq:decDef}
 P_{j_1, \ldots, j_n} =  \sum_{\alpha, \alpha' \in \mathcal{I}^{\sfacetF}}  \left(B^{[1]}_{j_1}\right)_{\alpha_{|_1}, \alpha'_{|_1}} \cdots \left(B^{[n]}_{j_n}\right)_{\alpha_{|_n}, \alpha'_{|_n}} = \bra{\Omega_r} B_{j_1}^{[1]} \otimes \cdots \otimes B_{j_n}^{[n]} \ket{\Omega_r}
\end{equation}
be a psd $(\Omega,G)$-decomposition of $P$ with $\psdrank_{(\Omega,G)}(P) \leq r = |\mathcal{I}|$.
As the last expression in \eqref{eq:decDef} already suggests, we use $B_j^{[i]}$ to construct a POVM and $\ket{\Omega_r}$ to construct a state whose combination lead to $P$. While the matrices $B^{[i]}_j$ are psd, they do not form a POVM in general since
$$\sum_{j = 1}^{k} B^{[i]}_{j} \neq \one_{r_i}$$
with $r_i = |\mathcal{I}^{\facetF_i}|$.
To this end, define
$$ S^{[i]} \coloneqq \sum_{j=1}^{d} B^{[i]}_{j} = \sum_{\ell=1}^{m_i} \lambda_{\ell}^{[i]} \ket{w_\ell^{[i]}} \bra{w_\ell^{[i]}}$$
with $\lambda_{\ell}^{[i]} > 0$ being only the positive eigenvalues of $S^{[i]}$ and $\ket{w_{\ell}^{[i]}}$ being the $G$-invariant eigenvectors of the family $S^{[1]}, \ldots, S^{[n]}$ according to \Cref{lem:GsymmetricDiagonalization}. Define
$$ T^{[i]} = \sum_{\ell=1}^{m_i} \left(\lambda_{\ell}^{[i]}\right)^{-1/2} \ket{w_{\ell}^{[i]}} \bra{\ell} \quad \textrm{ and } \quad W^{[i]} = \sum_{\ell=1}^{m_i} \left(\lambda_{\ell}^{[i]}\right)^{1/2} \ket{\ell} \bra{w_{_\ell}^{[i]}}.$$
Note that $T^{[i]} \cdot W^{[i]}$ is a projector on the subspace $\textrm{span}(\{\ket{w_1^{[i]}}, \ldots, \ket{w_{m_i}^{[i]}}\})$. Therefore, we have that
\begin{equation}
\label{eq:BiProj} B_j^{[i]} = \left(T^{[i]} \cdot W^{[i]}\right)^{\dag} \cdot B_j^{[i]} \cdot \left(T^{[i]} \cdot W^{[i]}\right). 
\end{equation}
We have that
\begin{equation}
\label{eq:TW_symmetry}
\bra{{}^g \beta} T^{[gi]} = \bra{\beta} T^{[i]} \quad \textrm{ and } \quad W^{[gi]} \ket{{}^g \beta} = W^{[i]} \ket{\beta}
\end{equation}
since the vectors $\ket{w_{\ell}^{[i]}}$ are $G$-invariant. We now define a POVM $(A_j^{[i]})_{j=1}^{d}$ via
$$ A_j^{[i]} = \left(T^{[i]}\right)^{\dag} \cdot B_j^{[i]} \cdot T^{[i]}.$$
We have that $A_j^{[i]}$ is psd and
$$\sum_{j=1}^{d} A_j^{[i]} = \one_{m_i}$$
which shows that $A^{[i]} \coloneqq \left(A_j^{[i]}\right)_{j=1, \ldots, d}$ is indeed a POVM for $i \in \{1, \ldots, n\}$. Moreover, $\left(A^{[i]}\right)_{i=1}^{n}$ is a $G$-invariant family since
\begin{align*}
A_j^{[gi]} &= \left(T^{[gi]}\right)^{\dag} \cdot B_j^{[gi]} \cdot T^{[gi]} = \sum_{\beta, \beta' \in \mathcal{I}^{\sfacetF_i}} \left(\bra{\beta} T^{[gi]} \right)^{\dag} \bra{\beta} B_j^{[gi]} \ket{\beta'} \bra{\beta'} T^{[gi]} \\
&=  \sum_{\beta, \beta' \in \mathcal{I}^{\sfacetF_i}} \left(\bra{{}^g \beta} T^{[gi]} \right)^{\dag} \bra{{}^g \beta} B_j^{[gi]} \ket{{}^g \beta'} \bra{{}^g \beta'} T^{[gi]} \\
&= \sum_{\beta, \beta' \in \mathcal{I}^{\sfacetF_i}} \left(\bra{\beta} T^{[i]} \right)^{\dag} \bra{\beta} B_j^{[i]} \ket{\beta'} \bra{\beta'} T^{[i]} = A_j^{[i]}
\end{align*}
where we have used that $\beta \mapsto {}^g \beta$ is a bijection between $\mathcal{I}^{\facetF_i}$ and $\mathcal{I}^{\facetF_{gi}}$ in the third step, and Equation \eqref{eq:TW_symmetry} in the fourth step. Moreover, defining 
$$ \ket{\psi} = W^{[1]} \otimes \cdots \otimes W^{[n]} \ket{\Omega_r}$$
leads to a normalized state with $\rank_{(\Omega, G)}(\ket{\psi}) \leq r$ since
\begin{align*}
\Cbraket{\psi}{\psi} &= \bra{\Omega_r} \left(W^{[1]}\right)^{\dag} W^{[1]} \otimes \cdots \otimes \left(W^{[n]}\right)^{\dag} W^{[n]}  \ket{\Omega_r} = \bra{\Omega_r} S^{[1]} \otimes \cdots \otimes S^{[n]} \ket{\Omega_r}\\
&= \sum_{j_1, \ldots, j_n = 1}^{d} \sum_{\alpha, \alpha \in \mathcal{I}^{\sfacetF}} \left(B^{[1]}_{j_1}\right)_{\alpha_{|_1}, \alpha'_{|_1}} \cdots \left(B^{[n]}_{j_n}\right)_{\alpha_{|_n}, \alpha'_{|_n}} = \sum_{j_1, \ldots, j_n = 1}^{d} P_{j_1, \ldots, j_n} = 1
\end{align*}
where we have used that $\left(P_{j_1, \ldots, j_n}\right)_{j_1, \ldots, j_n = 1}^{d}$ represents a probability distribution in the last step. Finally, the defined POVMs $\left(A_j^{[i]}\right)_{j=1}^{d}$ for $i \in \{1,\ldots, n\}$ and the state $\ket{\psi}$ generate the probability distribution $P$, since
$$ \bra{\psi} A_{j_1}^{[1]} \otimes \cdots \otimes A_{j_n}^{[n]} \ket{\psi} = \sum_{\alpha, \alpha' \in \mathcal{I}^{\sfacetF}} \left(B_{j_1}^{[1]}\right)_{\alpha_{|_1}, \alpha'_{|_1}} \cdots \left(B_{j_n}^{[n]}\right)_{\alpha_{|_n}, \alpha'_{|_n}} = P_{j_1, \ldots, j_n}$$
where we have used Equation \eqref{eq:BiProj} in the first step and Equation \eqref{eq:decDef} in the second step.
\end{proof}

\subsection{Correspondence between $\purirank$ and $\QQCorr$}
\label{ssec:purirankCorrespondence}

We now prove the quantum version of \Cref{thm:CQCorr}, namely that elements of $\QQCorr_{(\Omega,G)}(n,d,r)$ are precisely psd matrices $\rho$ with $\tr(\rho) = 1$ and $\purirank_{(\Omega,G)}(\rho) \leq r$.
The proof of this statement is similar to that of \Cref{thm:CQCorr}. In particular, it first restricts to density matrices which are diagonal in the computational basis according to \Cref{rem:diagonalDec} and then uses the fact that quantum channels whose image are diagonal states correspond to POVMs (see Equation \eqref{eq:DiagCPTP}).

\begin{theorem}[The puri-rank and quantum correlation scenarios]\label{thm:QQCorr}
Let $\Omega$ be a wsc and $G$ an external group action. Further, let $\rho \in \mathcal{M}_d(\mathbb{C})^{\otimes n}$ be psd and $\tr(\rho) = 1$. The following are equivalent:
\begin{enumerate}[label=(\roman*)]
	\item \label{thm:QQCorrI} $\rho \in \QQCorr_{(\Omega,G)}(n,d,r)$.
	\item \label{thm:QQCorrII} $\purirank_{(\Omega,G)}(\rho) \leq r$.
\end{enumerate}
\end{theorem}

The construction of \Cref{thm:QQCorr} \ref{thm:QQCorrII} $\Longrightarrow$ \ref{thm:QQCorrI} is depicted as a tensor network in \Cref{fig:normalization_process} for one-dimensional purification forms, i.e.\ a $\Lambda_n$-purification. 

\begin{proof}
\ref{thm:QQCorrI} $\Longrightarrow$ \ref{thm:QQCorrII}: Let $\rho$ be a density matrix in $\textsf{QQCorr}_{(\Omega,G)}(n,d,r)$. By definition, there exists a state
$$ \ket{\psi} = \sum_{\alpha \in \mathcal{I}^{\sfacetF}} \ket{v^{[1]}_{\alpha_{|_1}}} \otimes \cdots \otimes \ket{v^{[n]}_{\alpha_{|_n}}} $$
such that $\textrm{rank}_{(\Omega,G)}(\ket{\psi}) \leq r = |\mathcal{I}|$ and $G$-invariant quantum channels
\begin{equation}
\label{eq:MiDef}
\mathcal{E}_i(-) \coloneqq \sum_{k = 1}^{d_i} \left(A^{[i]}_k\right) \cdot \ - \ \cdot \left(A^{[i]}_k\right)^{\dag}
\end{equation}
with the condition that $A^{[i]}_k = A^{[gi]}_k$ that generate $\rho$. We now define $L \in \mathcal{M}_{d, d_0}(\mathbb{C}) \otimes \cdots \otimes \mathcal{M}_{d, d_n}(\mathbb{C})$ such that
\begin{enumerate}[label=(\alph*)]
	\item \label{pr:puridec} $\rho = L L^{\dag}$
	\item $\textrm{rank}_{(\Omega,G)}(L) \leq r$
\end{enumerate}
which proves \ref{thm:QQCorrII}. For $i \in \{1, \ldots, n\}$ and $\beta \in \mathcal{I}^{\facetF_i}$ let
\begin{equation}
\label{eq:LiDef}
L^{[i]}_{\beta} \coloneqq \sum_{k=1}^{d_i} A^{[i]}_k \ket{v^{[i]}_{\beta}} \bra{k}.
\end{equation}
Further, set 
$$L = \sum_{\alpha \in \mathcal{I}^{\sfacetF}} L^{[1]}_{\alpha_{|_1}} \otimes \cdots \otimes L^{[n]}_{\alpha_{|_n}}.$$
By definition, we have that $\textrm{rank}_{(\Omega,G)}(L) \leq r$. It only remains to prove \ref{pr:puridec}. This follows from
\begin{align*} L L^{\dag} = \sum_{k_1, \ldots, k_n = 1}^{d} \left(A_{k_1}^{[1]} \otimes \cdots \otimes  A_{k_n}^{[n]}\right) \ket{\psi} \bra{\psi} \left(A_{k_1}^{[1]} \otimes \cdots \otimes  A_{k_n}^{[n]}\right)^{\dag} = (\mathcal{E}_1 \otimes \cdots  \otimes \mathcal{E}_n)(\ket{\psi} \bra{\psi}) = \rho
\end{align*}
where we have used Equation \eqref{eq:LiDef} in the first step and Equation \eqref{eq:MiDef} in the second step.

\ref{thm:QQCorrII} $\Longrightarrow$ \ref{thm:QQCorrI}: Let $\rho = L L^{\dag}$ where
$$ L = \sum_{\alpha \in \mathcal{I}^{\sfacetF}} L_{\alpha_{|_1}}^{[1]} \otimes \cdots \otimes L_{\alpha_{|_n}}^{[n]} $$
be an $(\Omega,G)$-purificiation with $\textrm{puri-rank}_{(\Omega,G)}(\rho) \leq r = |\mathcal{I}|$. 

Defining the completely positive maps
$$\mathcal{N}_i(-) \coloneqq \sum_{k=1}^{d'} \left(B_k^{[i]}\right) \cdot \ - \ \cdot \left(B_k^{[i]}\right)^{\dag} \quad \text{ with } \quad \left(B_k^{[i]}\right)_{\ell, \beta} = \left(L_{\beta}^{[i]}\right)_{\ell, k}$$
we have that
\begin{equation} 
\label{eq:unnormalizedCP}
\rho = (\mathcal{N}_1 \otimes \cdots \otimes \mathcal{N}_n)(\ket{\Omega_r} \bra{\Omega_r}).
\end{equation}
where $\ket{\Omega_r}$ is the structure tensor defined in \Cref{ssec:structureTensor}.
However, $\mathcal{N}_i$ is neither trace-preserving nor $G$-invariant and $\ket{\Omega_r}$ is not normalized. For this reason, define
$$S^{[i]} \coloneqq \sum_{k=1}^{d'} \left(B_k^{[i]}\right)^{\dag} \left(B_k^{[i]}\right) = \sum_{\ell = 1}^{m_i} \lambda_{\ell}^{[i]} \ket{w_{\ell}^{[i]}} \bra{w_{\ell}^{[i]}}$$
where $\ket{w_{\ell}^{[i]}}$ is a $G$-invariant eigendecomposition of the family $S^{[1]}, \ldots, S^{[n]}$ according to \Cref{lem:GsymmetricDiagonalization}. Similarly to the proof of \Cref{thm:CQCorr} we define 
\begin{equation} \label{eq:TiWi}
T^{[i]} = \sum_{\ell=1}^{m_i} \left(\lambda_{\ell}^{[i]}\right)^{-1/2} \ket{w_{\ell}^{[i]}} \bra{\ell} \quad \textrm{ and } \quad W^{[i]} = \sum_{\ell=1}^{m_i} \left(\lambda_{\ell}^{[i]}\right)^{1/2} \ket{\ell} \bra{w_{_\ell}^{[i]}}.
\end{equation}
and completely positive maps
\begin{equation}
\label{eq:cptpMap}
\mathcal{E}_i(-) = \sum_{k=1}^{d'} \left(A_k^{[i]}\right) \cdot \ - \ \cdot \left(A_k^{[i]}\right)^{\dag} \quad \text{ with } \quad A_k^{[i]} \coloneqq B_k^{[i]} \cdot T^{[i]}.
\end{equation}
Note that $(\mathcal{E}_i)_{i=1,\ldots,n}$ is by definition a $G$-invariant family of quantum channels. Moreover, by the reasoning of the proof of \Cref{thm:CQCorr},
\begin{equation} 
\label{eq:normalizedState} 
\ket{\psi} = W^{[1]} \otimes \cdots \otimes W^{[n]} \ket{\Omega_r} \end{equation}
defines a normalized state with $\rank_{(\Omega,G)}(\ket{\psi}) \leq r$. Moreover,
$$\left(\mathcal{E}_1 \otimes \cdots \otimes \mathcal{E}_n\right)(\ket{\psi} \bra{\psi}) = \left(\mathcal{N}_1 \otimes \cdots \otimes \mathcal{N}_n\right)(\ket{\Omega_r} \bra{\Omega_r}) = \rho$$
which proves the statement.
\end{proof}

\begin{figure}[h!]
\centering
\begin{tikzpicture}[scale=0.95]

\draw (-1,5) rectangle (16,9);

\draw[line width=1cm, white] (2.4,3.3) -- (1,5.5);
\draw[thick, -stealth] (2.4,3.3) -- (1,5.5);

\draw[line width=1cm, white] (11.5,3.4) -- (13.7,6.2);
\draw[thick, stealth-] (11.5,3.4) -- (13.7,6.2);

\begin{scope}

\draw[draw=none, fill=gray!10!white] (-0.5,1.4) rectangle (5.5,3.1);

\draw[thick, gray] (0.25, 0.5) -- (0.25,2);
\draw[thick, gray] (1.75, 0.5) -- (1.75,2);
\draw[thick, gray] (3.25, 0.5) -- (3.25,2);
\draw[thick, gray] (4.75, 0.5) -- (4.75,2);

\draw[very thick] (0.5,2.25) -- (1.5,2.25);
\draw[very thick] (2,2.25) -- (3,2.25);
\draw[very thick] (3.5,2.25) -- (4.5,2.25);
\draw[very thick] (0.5,0.25) -- (1.5,0.25);
\draw[very thick] (2,0.25) -- (3,0.25);
\draw[very thick] (3.5,0.25) -- (4.5,0.25);

\draw[thick] (0.25,2.5) -- (0.25, 3);
\draw[thick] (1.75,2.5) -- (1.75, 3);
\draw[thick] (3.25,2.5) -- (3.25, 3);
\draw[thick] (4.75,2.5) -- (4.75, 3);

\draw[thick] (0.25,0) -- (0.25, -0.5);
\draw[thick] (1.75,0) -- (1.75, -0.5);
\draw[thick] (3.25,0) -- (3.25, -0.5);
\draw[thick] (4.75,0) -- (4.75, -0.5);

\draw[thick, fill=red!50!white] (-0.1,1.9) rectangle (0.6,2.6) node[midway] {\small $L^{[1]}$};
\draw[thick, fill=red!50!white] (1.4,1.9) rectangle (2.1,2.6) node[midway] {\small $L^{[2]}$};
\draw[thick, fill=red!50!white] (2.9,1.9) rectangle (3.6,2.6) node[midway] {\small $L^{[3]}$};
\draw[thick, fill=red!50!white] (4.4,1.9) rectangle (5.1,2.6) node[midway] {\small $L^{[4]}$};

\draw[thick, fill=red!50!white] (-0.1,-0.1) rectangle (0.6,0.6) node[midway] {\small $\bar{L}^{[1]}$};
\draw[thick, fill=red!50!white] (1.4,-0.1) rectangle (2.1,0.6) node[midway] {\small $\bar{L}^{[2]}$};
\draw[thick, fill=red!50!white] (2.9,-0.1) rectangle (3.6,0.6) node[midway] {\small $\bar{L}^{[3]}$};
\draw[thick, fill=red!50!white] (4.4,-0.1) rectangle (5.1,0.6) node[midway] {\small $\bar{L}^{[4]}$};

\node at (2.5,-1.3) {(a)};

\end{scope}


\begin{scope}[xshift=8.5cm, yshift=0.4cm]

\node at (-2,0.9) {$\stackrel{?}{=}$};

\draw[draw=none, fill=gray!25!white] (0.4,0.9) rectangle (6.5,1.65);
\node at (6.1,1.25) {$\ket{\psi}$};

\draw[draw=none, fill=gray!10!white] (1.2,1.8) -- (1.2,2.7) -- (-0.8,2.7) -- (-0.8,-0.9) -- (1.2,-0.9) -- (1.2,0) -- (0.3,0) -- (0.3, 1.8) -- cycle;
\node at (-0.3,0.9) {$\mathcal{M}_1$};

\draw[thick, red] (0.75,1) -- (0.75,2);
\draw[thick, red] (2.25,1) -- (2.25,2);
\draw[thick, red] (3.75,1) -- (3.75,2);
\draw[thick, red] (5.25,1) -- (5.25,2);

\draw[thick, gray] (0.25, -0.2) -- (0.25,2);
\draw[thick, gray] (1.75, -0.2) -- (1.75,2);
\draw[thick, gray] (3.25, -0.2) -- (3.25,2);
\draw[thick, gray] (4.75, -0.2) -- (4.75,2);

\draw[thick, fill=blue!50!white] (0,2) rectangle (1,2.5) node[midway] {\small $A^{[1]}$};
\draw[thick, fill=blue!50!white] (1.5,2) rectangle (2.5,2.5) node[midway] {\small $A^{[2]}$};
\draw[thick, fill=blue!50!white] (3,2) rectangle (4,2.5) node[midway] {\small $A^{[3]}$};
\draw[thick, fill=blue!50!white] (4.5,2) rectangle (5.5,2.5) node[midway] {\small $A^{[4]}$};

\draw[thick] (0.25,2.5) -- (0.25, 2.8);
\draw[thick] (1.75,2.5) -- (1.75, 2.8);
\draw[thick] (3.25,2.5) -- (3.25, 2.8);
\draw[thick] (4.75,2.5) -- (4.75, 2.8);

\draw[line width = 0.2cm, gray!25!white] (1,1.25) -- (5,1.25);
\draw[very thick] (1,1.25) -- (5,1.25);

\draw[thick, fill=white] (0.5,1) rectangle (1,1.5)  node[midway] {\scalebox{.6}{$W^{[1]}$}};

\draw[thick, fill=white] (2,1) rectangle (2.5,1.5)  node[midway] {\scalebox{.6}{$W^{[2]}$}};

\draw[thick, fill=white] (3.5,1) rectangle (4,1.5)  node[midway] {\scalebox{.6}{$W^{[3]}$}};

\draw[thick, fill=white] (5,1) rectangle (5.5,1.5)  node[midway] {\scalebox{.6}{$W^{[4]}$}};

\draw[thick, red] (0.75,0.8) -- (0.75,-0.2);
\draw[thick, red] (2.25,0.8) -- (2.25,-0.2);
\draw[thick, red] (3.75,0.8) -- (3.75,-0.2);
\draw[thick, red] (5.25,0.8) -- (5.25,-0.2);

\draw[thick] (0.25,-0.7) -- (0.25, -1);
\draw[thick] (1.75,-0.7) -- (1.75, -1);
\draw[thick] (3.25,-0.7) -- (3.25, -1);
\draw[thick] (4.75,-0.7) -- (4.75, -1);

\draw[thick, fill=blue!50!white] (0,-0.2) rectangle (1,-0.7) node[midway] {\small $\bar{A}^{[1]}$};
\draw[thick, fill=blue!50!white] (1.5,-0.2) rectangle (2.5,-0.7) node[midway] {\small $\bar{A}^{[2]}$};
\draw[thick, fill=blue!50!white] (3,-0.2) rectangle (4,-0.7) node[midway] {\small $\bar{A}^{[3]}$};
\draw[thick, fill=blue!50!white] (4.5,-0.2) rectangle (5.5,-0.7) node[midway] {\small $\bar{A}^{[4]}$};

\draw[line width = 0.2cm, white] (1,0.55) -- (5,0.55);
\draw[very thick] (1,0.55) -- (5,0.55);

\draw[thick, fill=white] (0.5,0.8) rectangle (1,0.3);

\draw[thick, fill=white] (2,0.8) rectangle (2.5,0.3);

\draw[thick, fill=white] (3.5,0.8) rectangle (4,0.3);

\draw[thick, fill=white] (5,0.8) rectangle (5.5,0.3);

\node at (2.5,-1.7) {(e)};

\end{scope}


\begin{scope}[yshift=6cm, xshift=-0.5cm, scale=0.6]

\node at (2.5,4) {(b)};

\draw[draw=none, fill=gray!10!white] (0,-0.5) -- (7,-0.5) -- (7, 0.5) -- (0,0.5) -- cycle;
\node at (6,0) {$\ket{\Omega_r}$};

\draw[thick, fill=red!50!white] (0,2) rectangle (0.5,2.5);
\draw[thick, fill=red!50!white] (1.5,2) rectangle (2,2.5);
\draw[thick, fill=red!50!white] (3,2) rectangle (3.5,2.5);
\draw[thick, fill=red!50!white] (4.5,2) rectangle (5,2.5);

\draw[thick, gray] (0.25, 1.5) -- (0.25,2);
\draw[thick, gray] (1.75, 1.5) -- (1.75,2);
\draw[thick, gray] (3.25, 1.5) -- (3.25,2);
\draw[thick, gray] (4.75, 1.5) -- (4.75,2);

\draw[thick] (0.25,2.5) -- (0.25, 2.8);
\draw[thick] (1.75,2.5) -- (1.75, 2.8);
\draw[thick] (3.25,2.5) -- (3.25, 2.8);
\draw[thick] (4.75,2.5) -- (4.75, 2.8);

\draw[very thick] (0.5,2.25) -- (0.75,2.25) -- (0.75,1) -- (0.35,1) -- (0.35,0);
\draw[very thick] (1.5,2.25) -- (1.25,2.25) -- (1.25,1) -- (1.65,1) -- (1.65,0);
\draw[very thick] (0.35,0) -- (1.65,0);

\draw[very thick] (2,2.25) -- (2.25,2.25) -- (2.25,1) -- (1.85,1) -- (1.85,0);
\draw[very thick] (3,2.25) -- (2.75,2.25) -- (2.75,1) -- (3.15,1) -- (3.15,0);
\draw[very thick] (1.85,0) -- (3.15,0);

\draw[very thick] (3.5,2.25) -- (3.75,2.25) -- (3.75,1) -- (3.35,1) -- (3.35,0);
\draw[very thick] (4.5,2.25) -- (4.25,2.25) -- (4.25,1) -- (4.65,1) -- (4.65,0);
\draw[very thick] (3.35,0) -- (4.65,0);

\node at (8.25,1.5) {$=$};

\end{scope}

\begin{scope}[xshift=5.5cm, yshift=6cm, scale=0.6]

\node at (2.5,4) {(c)};

\draw[dotted, fill=blue!10!white] (-0.25,0.125) -- (1,0.125) -- (1, 3) -- (-0.25,3) -- cycle;

\draw[thick, fill=red!50!white] (0,2) rectangle (0.5,2.5);
\draw[thick, fill=red!50!white] (1.5,2) rectangle (2,2.5);
\draw[thick, fill=red!50!white] (3,2) rectangle (3.5,2.5);
\draw[thick, fill=red!50!white] (4.5,2) rectangle (5,2.5);

\draw[thick, gray] (0.25, 1.5) -- (0.25,2);
\draw[thick, gray] (1.75, 1.5) -- (1.75,2);
\draw[thick, gray] (3.25, 1.5) -- (3.25,2);
\draw[thick, gray] (4.75, 1.5) -- (4.75,2);

\draw[thick] (0.25,2.5) -- (0.25, 2.8);
\draw[thick] (1.75,2.5) -- (1.75, 2.8);
\draw[thick] (3.25,2.5) -- (3.25, 2.8);
\draw[thick] (4.75,2.5) -- (4.75, 2.8);

\draw[very thick] (0.5,2.25) -- (0.75,2.25) -- (0.75,1) -- (0.35,1) -- (0.35,0.75);
\draw[very thick] (1.5,2.25) -- (1.25,2.25) -- (1.25,1) -- (1.65,1) -- (1.65,0.75);
\draw[very thick] (0.35, -0.5) -- (0.35,-0.75) -- (1.65,-0.75) -- (1.65,-0.5);

\draw[very thick] (2,2.25) -- (2.25,2.25) -- (2.25,1) -- (1.85,1) -- (1.85,0.75);
\draw[very thick] (3,2.25) -- (2.75,2.25) -- (2.75,1) -- (3.15,1) -- (3.15,0.75);
\draw[very thick] (1.85, -0.5) -- (1.85,-0.75) -- (3.15,-0.75) -- (3.15,-0.5);

\draw[very thick] (3.5,2.25) -- (3.75,2.25) -- (3.75,1) -- (3.35,1) -- (3.35,0.75);
\draw[very thick] (4.5,2.25) -- (4.25,2.25) -- (4.25,1) -- (4.65,1) -- (4.65,0.75);
\draw[very thick] (3.35,-0.5) -- (3.35,-0.75) -- (4.65,-0.75) -- (4.65,-0.5);

\draw[thick, red] (0.35,0) -- (0.35,0.25);
\draw[thick, red] (1.75,0) -- (1.75,0.25);
\draw[thick, red] (3.25,0) -- (3.25,0.25);
\draw[thick, red] (4.65,0) -- (4.65,0.25);

\draw[thick, fill=white] (0.1,0.25) rectangle (0.6,0.75);
\draw[thick, fill=white] (0.1,-0.5) rectangle (0.6,0);

\draw[thick, fill=white] (1.5,0.25) rectangle (2,0.75);
\draw[thick, fill=white] (1.5,-0.5) rectangle (2,0);

\draw[thick, fill=white] (3,0.25) rectangle (3.5,0.75);
\draw[thick, fill=white] (3,-0.5) rectangle (3.5,0);

\draw[thick, fill=white] (4.4,0.25) rectangle (4.9,0.75);
\draw[thick, fill=white] (4.4,-0.5) rectangle (4.9,0);

\draw [thick, draw=gray!50!black, decorate, decoration = {brace}] (5.25,0.85) -- (5.25,-0.6);
\node[anchor=west] at (5.5, 0.17) {\scriptsize $T^{[i]} \cdot W^{[i]} = P$};

\node at (8.25,1.5) {$=$};

\end{scope}

\begin{scope}[xshift=12cm, yshift=6cm, scale=0.6]

\node at (2.5,4) {(d)};

\draw[thick, red] (0.75,1) -- (0.75,2);
\draw[thick, red] (2.25,1) -- (2.25,2);
\draw[thick, red] (3.75,1) -- (3.75,2);
\draw[thick, red] (5.25,1) -- (5.25,2);

\draw[thick, fill=blue!50!white] (0,2) rectangle (1,2.5);
\draw[thick, fill=blue!50!white] (1.5,2) rectangle (2.5,2.5);
\draw[thick, fill=blue!50!white] (3,2) rectangle (4,2.5);
\draw[thick, fill=blue!50!white] (4.5,2) rectangle (5.5,2.5);

\draw[thick, gray] (0.25, 1.5) -- (0.25,2);
\draw[thick, gray] (1.75, 1.5) -- (1.75,2);
\draw[thick, gray] (3.25, 1.5) -- (3.25,2);
\draw[thick, gray] (4.75, 1.5) -- (4.75,2);

\draw[thick] (0.25,2.5) -- (0.25, 2.8);
\draw[thick] (1.75,2.5) -- (1.75, 2.8);
\draw[thick] (3.25,2.5) -- (3.25, 2.8);
\draw[thick] (4.75,2.5) -- (4.75, 2.8);

\draw[line width = 0.2cm, white] (1,1.25) -- (5,1.25);
\draw[very thick] (1,1.25) -- (5,1.25);

\draw[thick, fill=white] (0.5,1) rectangle (1,1.5);

\draw[thick, fill=white] (2,1) rectangle (2.5,1.5);

\draw[thick, fill=white] (3.5,1) rectangle (4,1.5);

\draw[thick, fill=white] (5,1) rectangle (5.5,1.5);

\end{scope}

\end{tikzpicture}
\caption{Proof of \Cref{thm:QQCorr} $(ii) \Longrightarrow (i)$ on a 1d chain, i.e. proving the equality of the expressions (a) and (e). (a) is the local purification form with $\puriosr(\rho) \leq r$. (b) When rearranging the wires we obtain the definition of a $\Omega$-decomposition using the structure-tensor $\ket{\Omega_r}$ according to Equation \eqref{eq:OmegaRDecomposition} (in this setting $\ket{\Omega_r}$ is a MaMu tensor). This decomposition can also be understood as applying a cp map to $\ket{\Omega_r}$ according to Equation \eqref{eq:unnormalizedCP}. In (c) we insert a projector $P^{[i]}$ of the space where the tensor $L^{[i]}$ acts non-trivially and factorize it into a product $T^{[i]} \cdot W^{[i]}$ according to Equation \eqref{eq:TiWi}. To obtain (d) we merge the upper box ($T^{[i]}$) with the red box (Equation \eqref{eq:cptpMap}). This gives rise to a normalized state (Equation \eqref{eq:normalizedState}) together with local quantum channels (e).}
\label{fig:normalization_process}
\end{figure}

\section{Nonnegative and separable decompositions and the proof of \Cref{thm:nnDecNoBorderRank}}
\label{app:nnDecNoBorderRank}

We now prove that the nonnegative and the separable $\Sigma_n$-decompositions do not exhibit a gap between border rank and rank. This generalizes the result that the nonnegative matrix rank and the nonnegative tensor rank are lower semicontinuous \cite{Bo11, Qi16}.

\begin{theorem}
\label{thm:seprankNoBorderRank}
Let $(\rho_n)_{n \in \mathbb{N}}$ be a sequence of separable matrices with $\rho_n \to \rho$ and $\seprank(\rho_n) \leq r$. Then
$$\seprank(\rho) \leq r $$
The same statement applies to $\symmseprank$, $\nnrank$ and $\symmnnrank$.
\end{theorem}
This entails the following.
\begin{corollary}
\label{cor:noNNBorderRank}
There is no gap between border rank and rank for $\nnrank$, $\symmnnrank$, $\seprank$ and $\symmseprank$.
\end{corollary}

To prove \Cref{thm:seprankNoBorderRank} we need the following preparatory lemma.
\begin{lemma}
\label{lem:psdInequality}
Let $A, B \in \mathcal{M}_{d}^{+}(\mathbb{C})$. Then,
$$\max \big\{\lambda_{\max}(A), \lambda_{\max}(B) \big\} \leq \lambda_{\max}(A + B)$$
\end{lemma}

\begin{proof}
Let
$$R_X(x) \coloneqq \frac{\bra{x}X\ket{x}}{\Cbraket{x}{x}}$$
for $x \in \mathbb{C}^d$. We have that $R_{A}(x) + R_{B}(x) = R_{A + B}(x)$ and since $A,B$ are psd, we have that $R_{A}(x), R_{B}(x) \geq 0$ for every $x$. This implies that
$$\max \{R_{A}(x), R_{B}(x) \} \leq R_{A+B}(x).$$
Since $\lambda_{\max}(X) = \max_{x \in \mathbb{C}^d} R_X(x)$, the result follows.
\end{proof}

\begin{proof}[Proof of \Cref{thm:seprankNoBorderRank}]
We prove it for $\symmseprank$. The proof for $\seprank$ is analogous, and the proof for $\nnrank$ and $\symmnnrank$ follows from restricting to diagonal matrices and the fact that $\nnrank(T) = \seprank(\Diag(T))$ (see \Cref{rem:diagonalDec}).

Let $(\rho_k)_{k \in \mathbb{N}}$ be a sequence of separable matrices with $\symmseprank(\rho_k) \leq r$, i.e.\ with a separable decomposition
$$ \rho_k = \sum_{\alpha = 1}^{r} \rho_{\alpha,k} \otimes \cdots \otimes \rho_{\alpha,k} $$
with $\rho_{\alpha,k}$ psd. Since all elementary tensors are themselves psd, we have that for all $\alpha, k$
$$ \Vert \rho_{\alpha,k} \Vert_\infty^n = \Vert \rho_{\alpha,k}^{\otimes n} \Vert_{\infty} \leq \Vert \rho_k \Vert_{\infty} \leq \Vert \rho \Vert_{\infty} + C $$
for some constant $C \in \mathbb{N}$ where the first equality is true since $\lambda_{\max}(\rho^{\otimes n}) = \lambda_{\max}(\rho)^n$, the first inequality follows by \Cref{lem:psdInequality}, and the last inequality follows from the convergence of $\rho_k$ to $\rho$.

This implies that $(\rho_{\alpha,k})_{k \in \mathbb{N}}$ is a bounded sequence. By the Bolzano--Weierstraß Theorem (\Cref{thm:BolzanoWeierstrass}) there is a subsequence $(k_\ell)_{\ell \in \mathbb{N}}$ such that $\rho_{\alpha, k_\ell}$ converges to a limiting point $\rho_{\alpha}$ which is again psd. Since $\rho_k \to \rho$ by assumption, we have that
$$ \rho = \sum_{\alpha=1}^{r} \rho_{\alpha} \otimes \cdots \otimes \rho_{\alpha}, $$
i.e.\ $\symmseprank(\rho) \leq r$, which shows the statement.
\end{proof}

\section{Tree tensor decompositions and the proof of \Cref{thm:treeBorderRank}}
\label{app:treeDecompositions}

Here we prove that no decomposition shows gap between border rank and rank whenever $\Omega$ is a tree. While this is known for unconstrained tensor decompositions (see \cite{Ba21, La12}), this was not known for positive tensor decompositions, to the best of our knowledge.

We call a wsc $\Omega$ a \emph{tree} if it satisfies the two following conditions:
\begin{enumerate}[label=(\roman*)]
	\item The facets of $\Omega$ represent a graph, i.e.\ $\Omega$ is a simplicial complex and for every $F \in \facetF$, we have $|F| \leq 2$.
	\item There are no loops in $\Omega$, i.e.\ for every sequence $i_1, \ldots, i_n$ of vertices such that
	$$\{i_1, i_2\}, \{i_2, i_3\}, \ldots, \{i_{n-1}, i_n\}, \{i_n, i_1\} \in \facetF$$
	we have that $i_1 = i_2 = \ldots = i_n$.
\end{enumerate}

In \Cref{ssec:unconstrainedTree} we review the result for unconstrained tensor decompositions, following \cite{Ba21}. In \Cref{ssec:sepTree}, we prove the main statement for separable and nonnegative tensor decompositions, and in \Cref{ssec:puriTree} we show the main result for psd-decompositions and local purification forms.

\subsection{Unconstrained tree tensor decompositions}
\label{ssec:unconstrainedTree}

In this part, we review the result that unconstrained $\Omega$-decompositions on trees $\Omega$ do not exhibit a gap between border-rank and rank, i.e.\
$$\brank_{\Omega}(T) = \rank_{\Omega}(T)$$
The idea is as follows. A tensor decomposition where a index only joins two local spaces, such as
$$T = \sum_{\alpha = 1}^{r} \ket{v_{\alpha}} \otimes \ket{w_{\alpha}}$$
corresponds to a matrix factorization $T = A \cdot B$ with $A \in \mathcal{M}_{d,r}(\mathbb{C})$ and $B \in \mathcal{M}_{r,d}(\mathbb{C})$, where each column of $A$ is given by a vector $\ket{v_{\alpha}}$ and each row of $B$ is given by a vector $\ket{w_{\alpha}}$.
Note that there is a ``gauge freedom'' in these decompositions, as specifically, for every $X \in \mathcal{M}_{r,r}(\mathbb{C})$ invertible, $\widetilde{A} = A \cdot X^{-1}$ and $\widetilde{B} = X \cdot B$ give rise to a new decomposition of $T$ of the same rank. Computing a thin (or reduced) $QR$-decomposition of $A$ \cite[Chapter 5]{Go96}, we obtain $A = Q \cdot R$ with $Q$ being an isometry in $\mathcal{M}_{d,r}(\mathbb{C})$ and $R \in \mathcal{M}_{r}(\mathbb{C})$ being an invertible matrix. Hence, $\widetilde{A} \coloneqq Q$ and $\widetilde{B} \coloneqq R \cdot B$ give rise to a decomposition where all tensor factors in the first part form an orthonormal basis, and the local vectors satisfy normalization conditions with respect to the Hilbert-Schmidt norm
$$ \Vert X \Vert_2 \coloneqq \sqrt{\tr\left(X^{\dag} X\right)} = \sqrt{\sum_{i,j = 1}^{d} |X_{i,j}|^2},$$
namely $\Vert \widetilde{A} \Vert_2 = \sqrt{r}$ and 
$$\Vert T \Vert_2 = \Vert \widetilde{A} \widetilde{B}\Vert_2 = \sqrt{\tr\left(\widetilde{B}^{\dag} Q^{\dag} Q \widetilde{B}\right)} = \sqrt{\tr\left(\widetilde{B}^{\dag} \widetilde{B}\right)} = \Vert \widetilde{B} \Vert_2$$
Similarly, one shows that for any tree $\Omega$ there exists a normalized $\Omega$-decomposition. Such decompositions are known as left-canonical forms in the tensor network literature (see \cite{Or14b, Ba21}).

\begin{lemma}
\label{lem:unconstrainedBounded}
Let $\ket{\psi} \in \mathbb{C}^d \otimes \cdots \otimes \mathbb{C}^d$ and $\Omega$ be a tree with $\rank_{\Omega}(\psi) \leq r$. There exists a decomposition
$$\ket{\psi}= W^{[1]} \otimes \cdots \otimes W^{[n]} \ket{\Omega_r}$$
such that
$$\Vert W^{[i]} \Vert_2 = \sqrt{r} \quad \text{for } i=1, \ldots, n-1,\text{ and } \quad \Vert W^{[n]} \Vert_2 = \sqrt{\Cbraket{\psi}{\psi}}$$
\end{lemma}

\begin{proof}
Follows directly from the proof in \cite[Proposition 1]{Ba21}.

\end{proof}

\Cref{lem:unconstrainedBounded} entails that there is no gap between border rank and rank for unconstrained $\Omega$-decompositions whenever $\Omega$ is a tree.

\begin{theorem}
\label{thm:NoBorderRanksUnconstrained}
If $\Omega$ is a tree, then $\rank_{\Omega} = \brank_{\Omega}$.
\end{theorem}
\begin{proof}
Let $\ket{\psi_k}$ be a sequence of states with $\ket{\psi_k} \rightarrow \ket{\psi}$ such that $\rank_{\Omega}(\ket{\psi_k}) \leq r$. We show that $\rank_{\Omega}(\ket{\psi}) \leq r$. By \Cref{lem:unconstrainedBounded} there exists tensor decomposition
$$\ket{\psi_k} = W^{[1]}_k \otimes \cdots \otimes W^{[n]}_k \ket{\Omega_r}$$
sucht that $\Vert W^{[i]}_k \Vert_2 = \sqrt{r}$ for $i=1, \ldots, n-1$ and $\Vert W^{[n]}_k \Vert_2 = \sqrt{\Cbraket{\psi_k}{\psi_k}}$. Since $\ket{\psi_k} \to \ket{\psi}$ there exists a constant $C$ such that
$$\sqrt{\Cbraket{\psi_k}{\psi_k}} \leq \sqrt{\Cbraket{\psi}{\psi}} + C$$
which implies that $(W^{[i]}_k)_{k \in \mathbb{N}}$ is a bounded sequence for every $i \in \{1, \ldots, n\}$. By the Bolzano--Weierstraß Theorem (\Cref{thm:BolzanoWeierstrass}), there exists subsequence $(W^{[i]}_{k_{\ell}})_{\ell \in \mathbb{N}}$ converging to a matrix $W^{[i]}$ for every $i \in \{1, \ldots,n\}$ which implies that
$$ \ket{\psi} = W^{[1]} \otimes \cdots \otimes W^{[n]} \ket{\Omega_r}.$$
\end{proof}

\subsection{Nonnegative and separable tree tensor decompositions} 
\label{ssec:sepTree}

Here we prove that nonnegative and separable tensor decompositions on trees $\Omega$ exhibit no gap between rank and border rank. The proof strategy is similar to that of $\Sigma_n$-decompositions: We show that there exist a decomposition of minimal rank with a bounding constraint, and then apply the Bolzano--Weierstraß Theorem. In this part, we prove the statement only for separable $\Omega$-decompositions; the result for nonnegative $\Omega$-decompositions follows by restricting to diagonal matrices and the fact that $\seprank_{\Omega}(\Diag(T)) = \nnrank_{\Omega}(T)$ (see \Cref{rem:diagonalDec}). We start by constructing a tensor decomposition of minimal rank where every local element satisfies a certain bounding constraint.

\begin{lemma}
\label{lem:boundSepDecTreeAppendix}
Let $\Omega$ be a tree and $\rho \in \mathcal{M}_{d}(\mathbb{C})^{\otimes n}$ a separable matrix with $\seprank_{\Omega}(\rho) \leq r$.
There exists a separable $\Omega$-decomposition of rank $r$
$$ \rho = \sum_{\alpha \in \mathcal{I}^{\sfacetF}} \rho_{\alpha_{|_1}}^{[1]} \otimes \cdots \otimes \rho_{\alpha_{|_n}}^{[n]}$$
such that $\tr(\rho_{\beta}^{[i]}) \leq 1$ for every $\beta \in \mathcal{I}^{\facetF_i}$ and $i \in \{1, \ldots, n-1\}$, and $\sum_{\beta \in \mathcal{I}^{\sfacetF_i}} \tr(\rho_{\beta}^{[n]}) = \tr(\rho)$.
\end{lemma}
\begin{proof}
We prove a stronger statement by induction over the number of vertices $n$. Specifically, we show that for every family $(\rho_{\delta})_{\delta \in \mathcal{I}}$ with a joint $\Omega$-decomposition
\begin{equation}
\label{eq:rho_dec}
 \rho_{\delta} = \sum_{\alpha \in \mathcal{I}^{\sfacetF}} \rho^{[1]}_{\alpha_{|_1}} \otimes \cdots \otimes \rho^{[n]}_{\alpha_{|_n}, \delta}
\end{equation}
the local tensors can be chosen such that $\tr(\rho^{[i]}_{\beta}) = 1$ for $\beta \in \mathcal{I}^{\facetF_i}$ and $i \in \{1, \ldots, n-1\}$, and $\sum_{\beta \in \mathcal{I}^{\sfacetF_n}} \tr(\rho^{[n]}_{\beta, \delta}) = \tr(\rho_{\delta})$. Setting $\delta = 1$ proves the claim. The idea of the induction step is shown in \Cref{fig:proofSketchTrees}.

For $n = 1$ (i.e.\ a single vertex) the statement is trivial.

For the induction step $n-1 \to n$, choose a joint $\Omega$-decomposition according to Equation \eqref{eq:rho_dec} without normalization constraints.
We assume without loss of generality that vertex $n$ is connected to precisely two other vertices.\footnote{If it is connected to more or less vertices the proof works analogously.} We denote the vertices of the first subtree $\Omega_1$ by $\{1, \ldots, k_1\}$, and the vertices on the second subtree $\Omega_2$ by $\{k_1 + 1, \ldots, n-1\}$. Moreover, vertices $k_1$ and $n-1$ are connected to vertex $n$ (\Cref{fig:proofSketchTrees}). For this reason, we can rewrite the separable $\Omega$-decomposition $\rho_{\delta}$ as
$$ \rho_{\delta} = \sum_{\gamma, \eta \in \mathcal{I}} \rho^{[1,\ldots, k_1]}_{\gamma} \otimes \rho^{[k_1 + 1, \ldots, n-1]}_{\eta} \otimes \rho^{[n]}_{\gamma, \eta, \delta}$$ 
with
$$ \rho^{[1,\ldots, k_1]}_{\gamma} =  \sum_{\alpha \in \mathcal{I}^{\sfacetG}} \rho_{\alpha_{|_1}}^{[1]} \otimes \cdots \otimes \rho_{\alpha_{|_{k_1}}, \gamma}^{[k_1]} \quad \text{ and } \quad \rho^{[k_1 + 1, \ldots, n-1]}_{\eta} =  \sum_{\alpha \in \mathcal{I}^{\sfacetH}} \rho_{\alpha_{|_{{k_1 + 1}}}}^{[k_1 + 1]} \otimes \cdots \otimes \rho_{\alpha_{|_{n-1}}, \eta}^{[n-1]} $$
where $\facetG$ and $\facetH$ are the sets of facets of $\Omega_1$ and $\Omega_2$ respectively. By applying the induction hypothesis to $\rho_{\gamma}^{[1,\ldots, k_1]}$ and $\rho_{\eta}^{[k_1+1,\ldots, n-1]}$, we obtain that all tensor factors have trace one, except the tensor factors at position $k_1$ and $n-1$. There, we have
$$ \sum_{\beta \in \mathcal{I}^{\sfacetF_{k_1}}} \tr(\rho^{[k_1]}_{\beta, \gamma}) = \tr(\rho^{[1,\ldots, k_1]}_{\gamma}) \quad \text{ and } \quad \sum_{\beta' \in \mathcal{I}^{\sfacetF_{n-1}}} \tr(\rho^{[n-1]}_{\beta', \eta}) = \tr(\rho^{[k_1+1, \ldots, n-1]}_{\eta}).$$
Defining
$$ \widetilde{\rho}_{\beta, \gamma}^{[k_1]} \coloneqq  \frac{1}{\tr(\rho^{[1,\ldots, k_1]}_{\gamma})} \rho_{\beta, \gamma}^{[k_1]},$$
$$\widetilde{\rho}_{\beta', \eta}^{[n-1]} \coloneqq  \frac{1}{\tr(\rho^{[k_1 + 1, \ldots, n-1]}_{\eta})} \rho^{[n-1]}_{\beta, \eta},$$
and
$$\widetilde{\rho}_{\gamma, \eta, \delta}^{[n]} \coloneqq \tr(\rho^{[1 \ldots k_1]}_{\gamma}) \cdot \tr(\rho^{[k_1 + 1 \ldots n-1]}_{\eta}) \cdot \rho_{\gamma, \eta, \delta}^{[n]}$$
we obtain a joint $\Omega$-decomposition 
\begin{align*}
\rho_{\delta} &= \sum_{\alpha \in \mathcal{I}^{\sfacetG}} \sum_{\alpha' \in \mathcal{I}^{\sfacetH}} \sum_{\gamma, \eta \in \mathcal{I}} \rho_{\alpha_{|_1}}^{[1]} \otimes \cdots \otimes \widetilde{\rho}_{\alpha_{|_{k_1}}, \gamma}^{[k_1]} \otimes \rho_{\alpha'_{|_{{k_1 + 1}}}}^{[k_1 + 1]} \cdots \otimes \widetilde{\rho}_{\alpha'_{|_{n-1}}, \eta}^{[n-1]} \otimes \widetilde{\rho}_{\gamma, \eta, \delta}^{[n]} \\
&= \sum_{\alpha \in \mathcal{I}^{\sfacetF}} \rho_{\alpha_{|_1}}^{[1]} \otimes \cdots \otimes \widetilde{\rho}_{\alpha_{|_{k_1}}}^{[k_1]} \otimes \rho_{\alpha_{|_{k_1 + 1}}}^{[k_1 + 1]} \otimes \cdots \otimes \widetilde{\rho}_{\alpha_{|_{n-1}}}^{[n-1]} \otimes \widetilde{\rho}_{\alpha_{|_{n}}, \delta}^{[n]}
\end{align*}
that satisfies the desired properties. Since every tree arises by sequentially attaching vertices in the described way, this proves the statement.
\end{proof}

\begin{figure}[t!]
\centering
\label{fig:proofStrategySepBounded}
\begin{tikzpicture}

\draw [fill=color1, draw=white] plot [smooth cycle, tension=0.5] coordinates {(-2,1.8) (0,1.8) (0,3) (-1,4.2) (-2,3)};
\draw [fill=color3, draw=white] plot [smooth cycle, tension=0.5] coordinates {(0.8,1.8) (2.2,1.8) (2.2,3.5) (0.8,4.2)};

\node at (-0.8,4.5) {$\gamma$};
\node at (0.8,4.5) {$\eta$};
\node[gray!50!white] at (0,5.5) {$\delta$};

\draw[fill=gray!20!white,draw=gray!20!white] (0,5) circle (1.5pt);

\draw[thick,gray!20!white] (-1,4) -- (0,5);
\draw[thick,gray!20!white] (1,4) -- (0,5);
\draw[thick,gray!20!white] (0,5) -- (0,5.25);

\draw[thick] (-1,4) -- (-0.75, 4.25);
\draw[thick] (1,4) -- (0.75, 4.25);

\draw[thick] (-1,4) -- (-1.5,3) -- (-1.75,2);
\draw[thick] (-1.5,3) -- (-1.25,2);
\draw[thick] (-1,4) -- (-1,3);
\draw[thick] (-1,4) -- (-0.5,3);

\draw[thick] (1,4) -- (1.5,3) -- (1.75,2);
\draw[thick] (1.5,3) -- (1.25,2);

\node[anchor=east] at (-1.2,4.2) {$[k_1]$};
\node[anchor=west] at (1.2,4.2) {$[n-1]$};
\node[anchor=west, gray!50!white] at (0.2,5.2) {$[n]$};
\node at (-0.25,2) {$\Omega_1$};
\node at (2,3.4) {$\Omega_2$};

\draw[fill=black, draw=color1, line width=0.3mm] (-1,4) circle (3pt);
\draw[fill=black, draw=color1, line width=0.3mm] (-1.5,3) circle (2pt);
\draw[fill=black, draw=color1, line width=0.3mm] (-1,3) circle (2pt);
\draw[fill=black, draw=color1, line width=0.3mm] (-0.5,3) circle (2pt);
\draw[fill=black, draw=color1, line width=0.3mm] (-1.75,2) circle (2pt);
\draw[fill=black, draw=color1, line width=0.3mm] (-1.25,2) circle (2pt);

\draw[fill=black, draw=color3, line width=0.3mm] (1,4) circle (3pt);
\draw[fill=black, draw=color3, line width=0.3mm] (1.5,3) circle (2pt);
\draw[fill=black, draw=color3, line width=0.3mm] (1.25,2) circle (2pt);
\draw[fill=black, draw=color3, line width=0.3mm] (1.75,2) circle (2pt);

\draw[thick, -stealth] (2.75,3.5) -- (4.75,3.5) node[midway, above]{Induction} node[midway, below]  {step};

\begin{scope}[xshift=7cm]

\draw[thick, color3, -stealth, xshift=0.1cm, yshift=-0.1cm] (-0.8,4.2) -- (-0.2,4.8);
\draw[thick, color3, -stealth, xshift=-0.1cm, yshift=-0.1cm] (0.8,4.2) -- (0.2,4.8);

\node at (-0.8,4.5) {$\gamma$};
\node at (0.8,4.5) {$\eta$};
\node at (0,5.5) {$\delta$};

\draw[thick] (-1,4) -- (0,5);
\draw[thick] (1,4) -- (0,5);
\draw[thick] (0,5) -- (0,5.25);

\draw[fill=black, draw=white, line width=0.3mm] (0,5) circle (3pt);

\draw[thick] (-1,4) -- (-0.75, 4.25);
\draw[thick] (1,4) -- (0.75, 4.25);

\draw[thick] (-1,4) -- (-1.5,3) -- (-1.75,2);
\draw[thick] (-1.5,3) -- (-1.25,2);
\draw[thick] (-1,4) -- (-1,3);
\draw[thick] (-1,4) -- (-0.5,3);

\draw[thick] (1,4) -- (1.5,3) -- (1.75,2);
\draw[thick] (1.5,3) -- (1.25,2);

\draw[fill=black, draw=white, line width=0.3mm] (-1,4) circle (2pt);
\draw[fill=black, draw=white, line width=0.3mm] (-1.5,3) circle (2pt);
\draw[fill=black, draw=white, line width=0.3mm] (-1,3) circle (2pt);
\draw[fill=black, draw=white, line width=0.3mm] (-0.5,3) circle (2pt);
\draw[fill=black, draw=white, line width=0.3mm] (-1.75,2) circle (2pt);
\draw[fill=black, draw=white, line width=0.3mm] (-1.25,2) circle (2pt);

\draw[fill=black, draw=white, line width=0.3mm] (1,4) circle (2pt);
\draw[fill=black, draw=white, line width=0.3mm] (1.5,3) circle (2pt);
\draw[fill=black, draw=white, line width=0.3mm] (1.25,2) circle (2pt);
\draw[fill=black, draw=white, line width=0.3mm] (1.75,2) circle (2pt);

\end{scope}

\end{tikzpicture}
\caption{Sketch of the induction step in the proof of \Cref{lem:boundSepDecTreeAppendix}. We assume that a normalized decomposition on every subtree $\Omega_1$, $\Omega_2$ exists. This implies that all local elements at the small nodes have trace $1$. The large nodes represent local elements whose normalization is given by the global element. In the induction step, we shift the global normalization constraint of node $k_1$ and $n-1$ to node $n$.}
\label{fig:proofSketchTrees}
\end{figure}

\begin{corollary}
Let $\Omega$ be a tree. Then, $\bnnrank_{\Omega} = \nnrank_{\Omega}$ and $\bseprank_{\Omega} = \seprank_{\Omega}.$
\end{corollary}
\begin{proof}
The proof is analogous to \Cref{cor:noNNBorderRank}. We prove it again only for separable decompositions. Let $(\rho_k)_{k \in \mathbb{N}}$ be a sequence of separable matrices such that $\seprank(\rho_k) \leq r$ and $\rho_k \to \rho$. We show that $\seprank(\rho) \leq r$. To this end, let
$$\rho_k = \sum_{\alpha \in \mathcal{I}^{\sfacetF}} \rho^{[1]}_{\alpha_{|_1}, k} \otimes \cdots \otimes \rho^{[n]}_{\alpha_{|_n}, k}$$
be a normalized decomposition according to \Cref{lem:boundSepDecTree}. We have that $\tr(\rho^{[i]}_{\beta, k}) = 1$ for every $i \in \{1, \ldots, n-1\}$ and $\tr(\rho^{[n]}_{\beta, k}) \leq \tr(\rho) + C$ for a suitable choice of $C$ due to the convergence $\rho_k \to \rho$. Hence, every tensor factor is a bounded sequence which has a convergent subsequence $\rho^{[i]}_{\beta, k_\ell} \to \rho^{[i]}_{\beta}$ for $\ell \to \infty$ due to \Cref{thm:BolzanoWeierstrass}. Since $\rho_k \to \rho$, we have that
$$ \rho = \sum_{\alpha \in \mathcal{I}^{\sfacetF}} \rho^{[1]}_{\alpha_{|_1}} \otimes \cdots \otimes \rho^{[n]}_{\alpha_{|_n}}$$
which shows that $\seprank(\rho) \leq r$.
\end{proof}

\subsection{The purification rank and the positive semidefinite rank on trees}
\label{ssec:puriTree}

Here we prove that for every tree $\Omega$, neither psd $\Omega$-decompositions nor $\Omega$-purifications exhibit a gap between rank and border rank. The proof strategy is similar to other cases without gaps: We use that there is a bounded decomposition with the same expressiveness and then apply the Bolzano--Weierstraß Theorem. In this case, we additionally use the correspondence to correlation scenarios (\Cref{thm:QQCorr}) and \Cref{thm:NoBorderRanksUnconstrained}.

\begin{theorem}
\label{thm:puriBorderRank}
Let $\Omega$ be a tree. Then, $\bpsdrank_{\Omega} = \psdrank_{\Omega}$ and $\bpurirank_{\Omega} = \purirank_{\Omega}$.
\end{theorem}
\begin{proof}
We prove the statement only for $\purirank_{\Omega}$ as the case of $\psdrank_{\Omega}$ works similarly. Let $(\rho_k)_{k \in \mathbb{N}}$ be a sequence of psd matrices such that $\purirank_{\Omega}(\rho_k) \leq r$ and $\rho_k \to \rho$. To see that $\purirank_{\Omega}(\rho) \leq r$, by \Cref{thm:QQCorr} there exists a sequence of states $\ket{\psi_k}$ with $\rank_{\Omega}(\ket{\psi_k}) \leq r$ and a sequence of quantum channels $\mathcal{E}^{(k)}_i$ for every $i \in [n]$ such that
$$\rho_{k} = \left(\mathcal{E}_1^{(k)} \otimes \cdots \otimes \mathcal{E}_n^{(k)}\right)(\ket{\psi_k} \bra{\psi_k}).$$
Since the space of quantum states is compact (we have that $\Cbraket{\psi}{\psi} = 1$ for every $\ket{\psi}$), and since the space of cptp maps is compact by \Cref{lem:cptpBounded}, there exists a joint subsequence $k_\ell$ such that $\mathcal{E}_i \coloneqq \lim_{\ell \to \infty} \mathcal{E}_i^{(k_{\ell})}$ and $\ket{\psi} \coloneqq \lim_{\ell \to \infty} \ket{\psi_{k_{\ell}}}$, which implies that
$$\rho = (\mathcal{E}_1 \otimes \cdots \otimes \mathcal{E}_n)(\ket{\psi} \bra{\psi}).$$
Since $\rank_{\Omega} = \brank_{\Omega}$, we have that $\rank_{\Omega}(\ket{\psi}) \leq r$, which proves that $\purirank_{\Omega}(\rho) \leq r$.

The proof for the psd-rank similarly uses \Cref{thm:CQCorr} and the fact that every sequence of a POVM has a convergent subsequence that converges to a POVM by the Bolzano--Weierstraß Theorem. \end{proof}

\end{document}